	\RequirePackage{silence} % :-\
	\documentclass[10pt,paper=b5,footinclude,headinclude]{scrbook} % KOMA-Script book
  \pdfoutput=1
\PassOptionsToPackage{utf8}{inputenc}
  \usepackage{inputenc}

\PassOptionsToPackage{T1}{fontenc} % T2A for cyrillics
  \usepackage{fontenc}

\PassOptionsToPackage{
  drafting=false,    % print version information on the bottom of the pages
  tocaligned=false, % the left column of the toc will be aligned (no indentation)
  dottedtoc=true,  % page numbers in ToC flushed right
  eulerchapternumbers=true, % use AMS Euler for chapter font (otherwise Palatino)
  linedheaders=true,       % chapter headers will have line above and beneath
  floatperchapter=true,     % numbering per chapter for all floats (i.e., Figure 1.1)
  eulermath=false,  % use awesome Euler fonts for mathematical formulae (only with pdfLaTeX)
  beramono=true,    % toggle a nice monospaced font (w/ bold)
  palatino=false,    % deactivate standard font for loading another one, see the last section at the end of this file for suggestions
  style=arsclassica % classicthesis, arsclassica
}{classicthesis}

\newcommand{\myTitle}{Current groups and the Hamiltonian anomaly}%\xspace}
\newcommand{\myName}{Ossi Niemimäki}%\xspace}

\newcommand{\myFaculty}{Faculty of Science}%\xspace}
\newcommand{\myDepartment}{Department of Mathematics and Statistics}%\xspace}
\newcommand{\myUni}{University of Helsinki}%\xspace}
\newcommand{\myVersion}{v.~1.0 : pre-print}
\newcommand{\myDefence}{Doctoral dissertation to be presented for public examination with the permission of the Faculty of Science of the University of Helsinki, in Auditorium~107, Athena Building, on the 27th of May 2020 at 14 o'clock.}

\providecommand{\mLyX}{L\kern-.1667em\lower.25em\hbox{Y}\kern-.125emX\@}

\newcommand{\eg}{e.\,g.}

\usepackage[drafting=false]{classicthesis}

\usepackage[ttscale=0.8]{libertine}
	\useosf				% figure lining: old-skool -- for libertine/biolinum

\usepackage{url}
\usepackage{hyperref}
\usepackage{xcolor}

\hypersetup{%
  colorlinks=true, linktocpage=true, pdfstartpage=3, pdfstartview=FitV,%
  breaklinks=true, pageanchor=true,%
  pdfpagemode=UseNone, %
  plainpages=false, bookmarksnumbered, bookmarksopen=true, bookmarksopenlevel=1,%
  hypertexnames=true, pdfhighlight=/O,%nesting=true,%frenchlinks,%
  urlcolor=CTurl, linkcolor=CTlink, citecolor=CTcitation, %pagecolor=RoyalBlue,%
  pdftitle={\myTitle},%
  pdfauthor={\textcopyright\ \myName, \myUni, \myFaculty},%
  pdfsubject={},%
  pdfkeywords={},%
  pdfcreator={pdfLaTeX},%
  pdfproducer={LaTeX with hyperref and classicthesis}%
}

\makeatletter
\@ifpackageloaded{babel}%
  {%
    \addto\extrasamerican{%
    }%
    \addto\extrasngerman{%
    }%
    }{\relax}
\makeatother
\usepackage[backend=biber,style=alphabetic,maxbibnames=4]{biblatex}
\usepackage{csquotes}
\usepackage[english]{babel}

\usepackage[b5paper,inner=12mm,outer=30mm,top=24mm,bottom=30mm,bindingoffset=8mm]{geometry}
\usepackage{marginnote}

\usepackage{ragged2e}

\usepackage{amsfonts}
\usepackage{amsmath}
\usepackage{amssymb}
\usepackage{amsthm}
\usepackage{mathtools}

\usepackage[super]{nth}

\usepackage{slashed}
\usepackage[normalem]{ulem}

\usepackage{enumitem}
\setlist[enumerate]{itemsep=0.25\baselineskip}

\usepackage{array}

\usepackage{mdframed}
\usepackage{thmtools}

\usepackage{etoolbox}
\makeatletter
\patchcmd{\ttlh@hang}{\parindent\z@}{\parindent\z@\leavevmode}{}{}
\patchcmd{\ttlh@hang}{\noindent}{}{}{}
\makeatother
\DeclareOldFontCommand{\bf}{\normalfont\bfseries}{\mathbf}
\usepackage{graphicx}
\usepackage{tikz}
  \usetikzlibrary{arrows,calc,decorations.text,positioning,shapes.multipart,tikzmark}
\usepackage{tikz-cd}

\let\oldFootnote\footnote
\newcommand\nextToken\relax
\renewcommand\footnote[1]{%
    \oldFootnote{#1}\futurelet\nextToken\isFootnote}
\newcommand\isFootnote{%
    \ifx\footnote\nextToken\textsuperscript{,}\fi}

\usepackage[makeindex]{imakeidx}
\makeindex

\newcommand{\idx}[1]{\index{#1|linkki}}
\newcommand{\idxsee}[2]{\index{#1|see {#2}}}
\newcommand{\idxalku}[1]{\index{#1|(linkki}}
\newcommand{\idxloppu}[1]{\index{#1|)linkki}}

\newcommand{\mots}[1]{\texorpdfstring{${#1}$}{}}

\newcommand{\R}{\mathbb{R}}
\newcommand{\C}{\mathbb{C}}

\newcommand{\N}{\mathbb{N}}

\newcommand{\Z}{\mathbb{Z}}
\newcommand{\dom}{\textrm{dom}} % domain
\newcommand{\Hs}{\mathcal{H}} % hilbert space
\newcommand{\I}{\mathcal{I}} % index set
\newcommand{\J}{\mathcal{J}} % index set
\newcommand{\M}{\mathcal{M}} % family of manifolds
\newcommand{\Ca}{\mathcal{C}} % algebra (general)
\newcommand{\bnd}[1]{\mathcal{B}(#1)} % bounded linear operators
\newcommand{\map}[1]{\ensuremath{\textrm{Map}(#1)}} % smooth maps
\newcommand{\A}{\mathcal{A}} % potential space
\newcommand{\U}{\mkern.5\thinmuskip\mathcal{U}} % open sets
\newcommand{\V}{\mkern.5\thinmuskip\mathcal{V}} % open sets
\newcommand{\Os}{\mkern.5\thinmuskip\mathcal{O}} % open sets
\newcommand{\sheaf}[1]{\underline{#1}} % sheaf of maps
\newcommand{\forms}{\ensuremath{\Lambda}}
\newcommand{\F}{\mathcal{F}} % Fock bundle
\newcommand{\G}{\mathcal{G}} % gauge transformations
\newcommand{\Lb}{\mathcal{L}} % universal bundle
\newcommand{\spin}{\mathrm{Spin}} % String group
\newcommand{\strng}{\mathrm{String}} % String group
\newcommand{\spec}{\mathrm{spec}} % spectrum
\newcommand{\Top}[1]{\ensuremath{\wedge^{\mkern-2mu#1}\mkern1mu}} % top power
\newcommand{\Gbig}{\G_S} % large gauge transformation space
\newcommand{\Abig}{\A_S} % large connection space
\newcommand{\redK}{\tilde{K}} % reduced K-group
\newcommand{\coh}{\ensuremath{\mathrm{H}}} % cohomology group
\newcommand{\coc}{\ensuremath{\mathrm{Z}}} % cocycle group
\newcommand{\coch}{\ensuremath{\mathrm{C}}} % cochain group
\newcommand{\cobo}{\ensuremath{\mathrm{B}}} % coboundary group
\newcommand{\ccoh}{\ensuremath{\breve{\coh}}} % Cech cohomology
\newcommand{\aut}{\ensuremath{\mathrm{Aut}}} % automorphisms
\newcommand{\extensions}{\ensuremath{\mathrm{Ext}}} % group extensions
\newcommand{\Det}{\ensuremath{\mathrm{Det}}} % determinant bundle

\newcommand{\exact}[3]{\ensuremath{0 \to #1 \to #2 \to #3 \to 0}}
\newcommand{\gext}[1]{\ensuremath{\widehat{#1}}} % group/algebra extension

\newcommand{\Lie}{\mathrm{Lie}} % Lie algebra of a group
\newcommand{\lie}[1]{\ensuremath{\mathfrak{#1}}} % lie algebra
\newcommand{\sct}[1]{\ensuremath{\Gamma \left( #1 \right)}} % section of a bundle
\newcommand{\embedd}{\ensuremath{\hookrightarrow}} % embedding

\newcommand{\paths}[1]{\ensuremath{\mathcal{P}(#1)}}
\newcommand{\loopg}[1]{\ensuremath{L{#1}}}
\newcommand{\bloopg}[1]{\ensuremath{\Omega{#1}}}
\newcommand{\flatloopg}[1]{\ensuremath{\Omega_{\flat}{#1}}}
\newcommand{\qloopg}[1]{\ensuremath{Q{#1}}}
\newcommand{\flatqloopg}[1]{\ensuremath{Q_{\flat}{#1}}}
\newcommand{\flatqbloopg}[1]{\ensuremath{Q_{\flat,e}{#1}}}

\newcommand{\bnloopg}[2]{\ensuremath{\Omega^{#1}{#2}}}
\newcommand{\bnballg}[2]{\ensuremath{B_{\flat}^{#1}{#2}}}
\newcommand{\threeloop}{\ensuremath{\bnloopg{3}{G}}}
\newcommand{\threeloopalg}{\ensuremath{\bnloopg{3}{\lie{g}}}}
\newcommand{\threeball}{\ensuremath{\bnballg{3}{G}}}
\newcommand{\threemaps}{\ensuremath{\map{\threeball,S^1}}}

\newcommand{\actg}[2]{\ensuremath{#1{\mkern-1mu/\mkern-3mu/}#2}}

\newcommand{\gerbe}[3]{%
  \begin{tikzcd}[ampersand replacement=\&]
      #1 \arrow{d} \& \\
      {#2}^{[2]} \arrow[r, shift left] \arrow[r, shift right] \& #2 \arrow{d} \\
      \& #3
  \end{tikzcd}%
  }

\DeclareMathOperator{\cker}{coker}	% cokernel
\DeclareMathOperator{\bdr}{\partial\mkern-3mu}	% boundary
\DeclareMathOperator{\ind}{ind}   % operator index
\DeclareMathOperator{\tr}{tr} % trace
\DeclareMathOperator{\id}{id} % identity
\DeclareMathOperator{\ad}{ad}		% adjoint
\DeclareMathOperator{\vol}{vol} % volume form
\newcommand{\drc}{\slashed{D}} % Dirac operator
\DeclareMathOperator{\Hom}{Hom} % morphisms
\DeclareMathOperator{\End}{End}

\newcommand{\isom}{\ensuremath{\cong}} % isomorphism
\newcommand{\chern}{\ensuremath{\mathrm{ch}}} % chern class
\newcommand{\ahatgenus}{\ensuremath{\hat{A}}} % A-hat genus

\newcommand{\conn}{\ensuremath{\nabla}} % connection
\newcommand{\Auni}{\Gamma} % universal connection
\newcommand{\cob}{\ensuremath{\delta}} % coboundary
\newcommand{\pull}[1]{\ensuremath{\pi^*(#1)}} % pullback
\newcommand{\ext}{\ensuremath{{\mkern.25mu}\mathrm{d}\mkern-.25mu}} % exterior derivative
\newcommand{\lder}{\ensuremath{\mathcal{L}}} % Lie derivative
\newcommand{\inv}[1]{\ensuremath{{#1}^{-1}}} % inverse
\newcommand{\abs}[1]{\ensuremath{\lvert #1 \rvert}}	% norm
\newcommand{\defeq}{\coloneqq}
\newcommand{\comm}[2]{\ensuremath{\left[ #1, #2 \right]}} % commutator
\newcommand{\acomm}[2]{\ensuremath{\left\{ #1, #2 \right\}}} % anticommutator

\newcommand{\norm}[1]{\ensuremath{\lVert #1 \rVert}} % norm
\newcommand{\ip}[1]{\ensuremath{\left\langle #1 \right\rangle}} % inner product
\newcommand{\idoperator}{\ensuremath{\mathrm{I}}} % identity operator

\newcommand{\creat}{\ensuremath{\mathbf{a^{\dagger}}}} % creation operator
\newcommand{\annih}{\ensuremath{\mathbf{a}}} % annihilation operator

\newcommand{\term}[1]{{\bfseries{#1}}} % in text
\newcommand{\termd}[1]{{\bfseries{#1}}} % in definition-environment

\newcommand{\laina}[1]{``#1''}

\newcommand{\bautoref}[1]{\bfseries{\autoref{#1}}}

\theoremstyle{plain}
\newtheorem{lemma}[equation]{Lemma}%[section]
\newtheorem{proposition}[lemma]{Proposition}

\theoremstyle{definition}

\newtheorem{example}{Example}[chapter]
\newtheorem*{example*}{Example}
\newtheorem{assumption}{Assumption}

\theoremstyle{remark}
\newtheorem{remark}{Remark}[chapter]
\definecolor{thmbg}{rgb}{0.975,0.95,0.95}
\declaretheoremstyle[%
headfont=\normalfont\bfseries,%
notefont=\mdseries, notebraces={(}{)},%
bodyfont=\normalfont\itshape,%
postheadspace=0.5em,%
preheadhook={\begin{mdframed}[backgroundcolor=thmbg, %
  innertopmargin =0\baselineskip , %
  skipbelow=.5\baselineskip, skipabove=.5\baselineskip, %
  hidealllines=true,usetwoside=true,nobreak=true]},%
postfoothook=\end{mdframed}%
]{varjo}
\declaretheorem[sibling=lemma,style=varjo]{definition}
\declaretheorem[sibling=definition,style=varjo]{theorem}
\declaretheoremstyle[%
  spaceabove=\topsep,%
  spacebelow=\topsep%,
  headfont=\normalfont,%
  headpunct={},%
  headformat={\hspace*{-.275em}\NOTE\hspace{.025em}.\hspace{.5em}},%
  headindent={},%
  postheadspace=0em,%
  notefont=\bfseries,%
  notebraces={}{},% punctuation before and after the note
]{oletustoisto}
\declaretheorem[style=oletustoisto,title=""]{assumptioN}

\newcommand\mlnode[1]{\mbox{\begin{tabular}{@{}c@{}}#1\end{tabular}}}

\begin{document}

\pagenumbering{roman}

% tyyli etusivulle
	\thispagestyle{empty}
	%!TEX root = teesirunko-arxiv.tex
%*******************************************************
% Titlepage
%*******************************************************
\begin{titlepage}
    \pdfbookmark[1]{\myTitle}{titlepage}
    % if you want the titlepage to be centered, uncomment and fine-tune the line below (KOMA classes environment)
    %\begin{addmargin}[10mm]{30mm}
    \newgeometry{hmargin=24mm}
    \begin{center}
        \large

        \hfill

        \vfill

        \begingroup
            % fontcolor for e-version!
            \color{CTtitle}
            \spacedallcaps{\myTitle} \\ \bigskip
        \endgroup

        \spacedlowsmallcaps{\myName} \\

        \vfill

        \small

        \begin{minipage}{0.6\textwidth}
          \centering\textit{\myDefence}
        \end{minipage} \\ \bigskip

        \vfill

        %\includegraphics[width=6cm]{gfx/TFZsuperellipse_bw} \\ \medskip

        %\mySubtitle \\ \medskip
        %\myDegree \\
        \myDepartment \\
        \myFaculty \\
        \myUni \\ \smallskip % \\ \bigskip

        \textsc{Helsinki 2020}

        %\myTime -- \myVersion
		    %\today{} -- \myVersion \\ %\smallskip
        %{\small\myEmail}

        %\vfill

    \end{center}
  	%\end{addmargin}
  	\restoregeometry
\end{titlepage}

%!TEX root = teesirunko-arxiv.tex
%*******************************************************
% After titlepage
%*******************************************************

\thispagestyle{empty}

{\setlength{\parindent}{0pt} %\small

  %\hfill
  \phantom{these go down!} 
  \vfill

  \textsc{isbn} 978-951-51-5976-2 (paperback) \\
  \textsc{isbn} 978-951-51-5977-9 (\textsc{pdf}) \\ \smallskip

  \url{http://ethesis.helsinki.fi} \\ 
  \url{http://urn.fi/URN:ISBN:978-951-51-5977-9} \\  
  \smallskip

  Unigrafia Oy \\
  Helsinki 2020
}

\cleardoublepage

%!TEX root = teesirunko-arxiv.tex

%*******************************************************
% Dedication
%*******************************************************
\thispagestyle{empty}
\phantomsection
\pdfbookmark[1]{Roll the dice}{Roll the Dice}

%\phantom{gnaa}
\vspace*{3cm}

\makebox[\textwidth][r]{
  \noindent
  \begin{minipage}{0.5\textwidth}
    {\itshape
    if you’re going to try,\\
    go all the way. \\

    there is no other feeling like \\
    it. \\
    you will be alone with the \\
    gods \\
    and the nights will flame with \\
    fire. \\ %\smallskip
    % all the way. \\
    % you will ride life straight to \\
    % perfect laughter, its \\
    % the only good fight \\
    % there is.  \\
    } \medskip

    {\footnotesize%
    --- Charles Bukowski}
  \end{minipage}
}

%\end{center}

%\medskip

\cleardoublepage

%!TEX root = teesirunko-arxiv.tex

%*******************************************************
% Abstract
%*******************************************************
%\renewcommand{\abstractname}{Abstract}
\pdfbookmark[1]{Abstract}{Abstract}
% \addcontentsline{toc}{chapter}{\tocEntry{Abstract}}
\begingroup
\let\clearpage\relax
\let\cleardoublepage\relax
\let\cleardoublepage\relax

\chapter*{Abstract}

Gauge symmetry invariance is an indispensable aspect of the field-theoretic models in classical and quantum physics. Geometrically this symmetry is often modelled with current groups and current algebras, which are used to capture both the idea of gauge invariance and the algebraic structure of gauge currents related to the symmetry.

The Hamiltonian anomaly is a well-known problem in the quantisation of massless fermion fields, originally manifesting as additional terms in current algebra commutators. The appearance of these anomalous terms is a signal of two things: that the gauge invariance of quantised Hamiltonian operators is broken, and that consequently it is not possible to coherently define a vacuum state over the physical configuration space of equivalent gauge connections.

In this thesis we explore the geometric and topological origins of the Hamiltonian anomaly, emphasising the usefulness of higher geometric structures in the sense of category theory. Given this context we also discuss higher versions of the gauge-theoretic current groups. These constructions are partially motivated by the $2$-group models of the abstract string group, and we extend some of these ideas to current groups on the three-sphere $S^3$.

The study of the Hamiltonian anomaly utilises a wide variety of tools from such fields as differential geometry, group cohomology, and operator K-theory.
We gather together many of these approaches and apply them in the standard case involving the time components of the gauge currents. We then proceed to extend the analysis to the general case with all space-time components. We show how the anomaly terms for these generalised current algebra commutators are derived from the same topological foundations; namely, from the Dixmier-Douady class of the anomalous bundle gerbe. As an example we then compute the full set of anomalous commutators for the three-sphere $S^3$ as the physical space.

\endgroup

\vfill

\cleardoublepage

\pdfbookmark[1]{\contentsname}{tableofcontents}
\tableofcontents
\cleardoublepage

%!TEX root = teesirunko-arxiv.tex

%*******************************************************
% Preface
%*******************************************************
\renewcommand{\prefacename}{Preface}
%\pdfbookmark[1]{Preface}{Preface}
\phantomsection
\addcontentsline{toc}{chapter}{\tocEntry{Preface}}
\begingroup
\let\clearpage\relax
\let\cleardoublepage\relax
\let\cleardoublepage\relax

\chapter*{Preface}

Around 2012 CE I was struck by a peculiar malady of the head: to leave behind the ever-green pastures of electrical engineering and venture into the wilderness of pure mathematics. Figuring out the \emph{hows} proved rather tricky, but the \emph{whys} were clear. Why spend my time not engaging in the most beautiful forms of art known to the human being? Now, many twists and turns later that path culminates -- or at least comes to a little rest -- here, in this thesis.

To this end I am grateful for the funding provided by the Emil Aaltonen Foundation%
\footnote{Grant nos. 160183~N and 170185~N.}%
, the Niilo Helander Foundation%
\footnote{Grant no. 180084}%
, and the Jenny and Antti Wihuri Foundation%
\footnote{Grant no. 00180261.}%
, without which this work would never have seen the light of day.

I am deeply indebted to my advisors Antti Kupiainen and Jouko Mickelsson: Antti for all the support and help with any and all practical issues I have had, and Jouko for all the \sout{madness} \emph{ideas} and gentle guidance throughout the endeavour. I also wish to thank Michael Murray and Richard Szabo for their pertinent review of the thesis during the pre-examination.

The Department of Mathematics and Statistics at the University of Helsinki has been my second home for almost five years now, for which I am ever grateful. I wish to thank especially those of the teaching crew with whom I have had collaboration; the delightful pain and toil of mathematics is only surpassed when aiding others to attain and overcome it, too. I also have fond memories of the local hospitality during my visits to the Centre for Quantum Geometry of Moduli Spaces in Aarhus, and to the Erwin Schrödinger Institute%
  \footnote{A part of this thesis was conceived during the ESI workshop \emph{Bivariant K-theory in Geometry and Physics} in 2018.}
in Vienna.

Beyond direct influence on the thesis and the subject matter, I owe pints to Janne Salo and Arto Hiltunen for their companionship, and a whole shelf of jars filled with preserved goods to Preppers for never staying the madness.
And where would I be without the support of my parents, brothers and sisters?

No wanderer finds their direction without a sun to guide them. Above all I owe to Erica, without whom I would have been lost a long time ago.

\endgroup

\vspace*{4\baselineskip}

\hfill%
\begin{minipage}{0.223\textwidth}
  Ossi Niemimäki\\
  Helsinki, April 2020
\end{minipage}

\phantom{gnaa}

\vfill

\cleardoublepage

\pagenumbering{arabic}

% line numbering
%\linenumbers

% intro
%!TEX root = teesirunko-arxiv.tex
\chapter{Introduction}

%\begin{quote}
%Veit ec at ec hecc vindga meiði a\\
%netr allar nío,\\
%geiri vndaþr oc gefinn Oðin,\\
%sialfr sialfom mer,\\
%a þeim meiþi, er mangi veit, hvers hann af rótom renn.\\

%Við hleifi mic seldo ne viþ hornigi,\\
%nysta ec niþr,\\
%nam ec vp rvnar,\\
%opandi nam,\\
%fell ec aptr þaðan.
%\end{quote}

The term \emph{Hamiltonian} is used in this thesis in two distinguished but not entirely disconnected respects. For one, \term{Hamiltonian quantisation} describes the operator formalism as the geometric field theory in continuum is transformed -- quantised -- to an algebraic theory of quantum operators.  The second meaning refers to the time-evolution of the physical systems, which can be written explicitly as a \term{Hamiltonian operator} acting on the relevant entities. Here we treat time as a simple parameter space characterising these dynamics; while acknowledging that this in all likelihood is not a fundamental aspect of physics, it nevertheless is useful simplification from a variety of perspectives.%
  \footnote{See \cite{ROVELLI2011} for a take on how time could be an emergent property of a quantum system.}

The switch from field geometry to operator algebra is not the only change brought in by the quantisation: some effects are also seen in the symmetries of the system. The geometric field theory incorporates a simple but far-reaching idea of invariance, which means that some internal components of the theory can be transformed according to certain rules \emph{without affecting the observable outcome}. This is known as the gauge symmetry invariance, although on the outset it has little to do with our intuition of visual symmetry since these symmetric similarities lie in the underlying mathematical formulation rather than in geometric shapes as such. One of the problems -- or properties -- of quantisation is that it can break some of these symmetries as ones moves from the classical to the quantum. The parts of the theory reflecting this are called anomalous and the effect itself an anomaly.

Hamiltonian anomaly is then an anomalous quantisation effect that is seen on the level of quantum operator dynamics. There are several ways to portray this formally, and \emph{one of the aims of this thesis is to put together many of these perspectives}. To this end we also introduce a handful of fundamental mathematical ideas which form the basis for quantisation anomalies. A central argument is that the \emph{symmetry anomalies have a topological origin}, or that they can be best approached by using topological tools such as K-theory and group cohomology. Many of these methods are not particularly new; what could be consider more contemporary is the use of higher structures in the sense of category theory. There is an undergoing process in mathematical and theoretical physics in general to figure out how gauge field theories could look from such a perspective. To this we add a small contribution concerning mostly categorical groups, motivated by the fact that it is in the group theory where one still finds the most appropriate mathematical realisations for symmetry.

Indeed, underlying the gauge symmetry are what are called \term{current groups}. Formally these can be formed of smooth maps $\map{M,G}$ from a smooth manifold $M$ to a Lie group $G$, defining infinite-dimensional Lie groups.%
  \footnote{Subject to appropriate conditions: in general, the proper context for spaces of Lie group -valued maps is the category of infinite-dimensional Fréchet Lie groups, see for instance \cite{NEEB2001}.}
  \marginnote{The first physical current algebra was constructed by Murray Gell-Mann (1929--2019) to describe hadrons.}
Their algebraic counterpart $\map{M,\lie{g}}$ consisting of continuous maps to the Lie algebra $\lie{g}$ of the symmetry group $G$ form the setting to what is called \term{current algebra}.
Historically, current algebras started as a formal means to understand the basics of nuclear interaction, leading to the theory of chromodynamics.

Using the current algebra one studies the commutator relations of physical currents. Corresponding to a gauge symmetry one can define current (densities) $j^a(x)$ as elements in the current algebra and then impose equal-time commutator relations akin to%
  \footnote{Here the indices $a,b,c$ refer to the generators of the algebra $\lie{g}$, and the arguments $x,y$ are points on the manifold $M$ of dimension $2k+1$. The currents are expressed in a simplified form without space-time indices; we will discuss the more general case in the main text.}
\[
  \comm{j^a(x)}{j^b(y)}_{t=0} = i \lambda^{ab}_c j^c(x) \delta^{2k+1}(x-y) .
\]
It is in these commutators that one also finds the earliest manifestation of the Hamiltonian anomaly: additional terms may appear on the right side of the above equation~\cite{GI1955,SCHWINGER1959}. We show how the appearance of the anomaly in such commutators is derived from the topology of the Dirac Hamiltonians in the case of massless fermion fields. The topological argument also makes it clear how the commutator anomaly is a result of the impossibility of defining a proper vacuum state in the space of observables.

\section{Outline}

The bulk of this thesis is naturally a distillation of already-established research. In Chapter~\ref{chap:gauge} we briefly recap the central ideas in gauge theory and its quantisation, focusing on the fermionic field theory. This is done largely just to set down common notation and to define the geometric backdrop.

In Chapter~\ref{chap:higher} we discuss the use of higher structures in gauge theory. We focus mostly on constructing categorical groups as in Section~\ref{sec:cat_groups}, in which the treatment of $3$-loop group extensions is largely based on an earlier publication~\cite{MN2019}.%
  \footnote{The joint publication \cite{MN2019} contains significant contributions by the author.}
Here we explain in more detail an alternative approach which previously was mentioned only in passing (\cite[Rem.~3]{MN2019}). The broader aim beyond these examples is to envision the concept of current groups in this higher framework. We will also give a brief introduction to bundle gerbes, which have been successfully used as a geometric handle to the cohomology classification of anomalies. 

The topological nature of quantum Dirac operators is discussed in Chapter~\ref{chap:ktheory}, in which we explore the rudiments of K-theory and operator index theorems. This forms the operator-theoretical foundations for the further treatment of the Hamiltonian anomaly.

Finally in Chapter~\ref{chap:hanomaly}, we review various perspectives on the Hamiltonian anomaly and try to clarify the connections in-between. What can be considered as new is Section~\ref{sec:general_comm}, in which we link the bundle gerbe cohomology to general anomalous commutator terms and derive them explicitly in a simplified example over the spatial manifold $M=S^3$. In this we rely on the topological arguments presented earlier in Chapter~\ref{chap:ktheory}.

Stylistically we aim to strike a balance between the basic treatments of quantum physics and more recent mathematical research. For instance, in Chapter~\ref{chap:gauge} we introduce the groups of gauge connections and their transformations in the standard textbook fashion while noting that the proper mathematical entity involved is not a pair of groups but a groupoid, and later come back to it with more details in Chapter~\ref{chap:higher}. However, while giving space to introduce abstract mathematical topics, we also try to keep in mind the goal of concrete realisations useful to physics. For the brave intent on diving deep into the modern mathematical treatment of gauge theories we refer \cite{SCHREIBER2017}.

% gauge symmetries and anomalies
%!TEX root = teesirunko-arxiv.tex
\chapter{Gauge theory and anomalies}\label{chap:gauge}

In this brief introductory chapter we go through some of the basic notions behind gauge field theory and its quantisation. We will pay special attention to the structure of \emph{fermionic field theories}. We will also briefly discuss the general idea of quantisation anomalies: these appear when the quantisation process breaks the invariance of quantum operators under the action of some physically relevant symmetry group.

Since this chapter is meant to be a quick introduction to the main ideas in gauge theory, for the most part we opt for a simplified approach in describing the mathematical structures. The subsequent chapters will lay out more details for some of the material, where relevant and needed.

Some general assumptions introduced in this chapter carry out through the text. They are the following:
\begin{assumptioN}[\bautoref{assu:m}]
	Let $M$ be the manifold description for the physical space. We assume that $M$ is a compact and connected spin manifold. The space-time structure can then be expressed as a globally hyperbolic product manifold $\M = M\times\R$.
\end{assumptioN}
\begin{assumptioN}[\bautoref{assu:g}]
	Let $G$ be the group of internal gauge symmetries. We assume that $G$ is a compact, connected, and semi-simple Lie group. We work mostly with matrix Lie groups; when in doubt, assume $G=SU(p)$ for a suitable $p\geq2$.
\end{assumptioN}
\begin{assumptioN}[\bautoref{assu:gg}]
	Let $\G$ be the set of vertical automorphisms on the space of gauge connections $\A$. We assume that as a group $\G$ is based: for a fixed $p\in M$, $g(p) = e$ for all $g \in \G$, where $e\in G$ is the identity element. 
\end{assumptioN}

\section{Classical gauge theory}

In this section we will discuss the basic geometric setting for fermionic gauge theory: to this end we introduce gauge connections, spinors and Dirac operators. Finally, we discuss the moduli space of gauge connections. Our take on gauge connections rests mostly on that of \cite{MS2000}, and the canonical source for the spin geometry and Dirac operators is \cite{LM1989}.

The described geometry is \term{classical}: that is, we define field entities and the relevant structure over smooth spaces without the notion of quantised energies. 
\idx{space-time}
Let $\M$ be a suitable finite-dimensional space-time manifold. In relativistic quantum field theories $\M$ is Min\-kowskian or, more generally, pseudo-Riemannian. However, for the most part we are content to work with Riemannian manifolds and regard time as a simple parameter: thus we will write $\M = M\times \R$, where $M$ accounts for the \emph{physical space} (possibly restricting to a time interval $I\subset \R$). In particular, we assume that $M$ is oriented, compact, connected, and smooth. We also assume a spin structure in order to construct the fermionic theory of spinor fields.

\begin{assumption}\label{assu:m}
	Let $M$ be the manifold description for the physical space. We assume that $M$ is a compact and connected spin manifold. The space-time structure can then be expressed as a globally hyperbolic product manifold $\M = M\times\R$.
\end{assumption}

\idx{connection!gauge}
\idxsee{gauge!connection}{connection}
\idxsee{gauge!potential}{connection}
The main actors of gauge field theory are \term{gauge potentials} or \term{connections} locally expressed as functions $A^{\mu}_a$, with $\mu$ the local space-time index and $a$ the internal symmetry index. This latter index runs through the dimensions of the Lie algebra of the \term{gauge or internal symmetry group} %
	\idx{gauge!symmetry group}
$G$ which represents the invariance of the theory -- the internal degrees of freedom of the connection. Let $\lie{g} \defeq \mathrm{Lie}(G)$ be the Lie algebra of $G$ with a matrix representation generated by anti-Hermitian matrices $\tau^a$ with relations%
	\marginnote{The \term{Einstein convention} of summing over repeated indices is used throughout the text.}
\[
	\comm{\tau^a}{\tau^b} = i\lambda_c^{ab}\tau^c \quad \text{and} \quad (\tau^a)^{\dagger} = -\tau^a .
\]
The numbers $\lambda_c^{ab}$ are the \term{structure constants} %
	\idx{structure constants}
of the algebra. 
In this framework one then writes the connections as %
	\idx{connection!$1$-form}
\[
	A \defeq A_a \tau^a = A_a^{\mu}(x)\ext x_{\mu} \tau^a ,
\]
with values in the Lie algebra $\lie{g}$, and $A_a \in \forms^1(M)$.%
	\footnote{Often for physical applications there is also a coupling constant $c$ in the expression $A = c A_a \tau^a $; this determines the strength of the field in relation to the other components of theory.}
Note that these expressions must be understood \emph{locally} unless one is working on a trivial principal $G$-bundle, as explained below.

\subsection{Gauge field geometry}
% principal bundle indexing starts
\idxalku{principal bundle}
% connection indexing starts
\idxalku{connection}

Let $G$ be a compact, connected, semi-simple Lie group, and let $P \to M$ be a smooth principal $G$-bundle over a Riemannian spin manifold $M$. We are mostly interested in the special unitary group $SU(p)$ as the structure group $G$,
and in general matrix groups will be enough for the physical applications considered.

\begin{assumption}\label{assu:g} %
	\idx{gauge!symmetry group}
	Let $G$ be the group of internal gauge symmetries. We assume that $G$ is a compact, connected, and semi-simple Lie group. We work mostly with matrix Lie groups; when in doubt, assume $G=SU(p)$ for a suitable $p\geq2$.
\end{assumption}

Any principal bundle $P$ has associated vector bundles $E \to M$ defined by a representation $\rho: G \to \aut(V)$, where $V$ is the standard (vector space) fibre of the bundle $E\defeq P \times_G V$ and we have a left action of $G$ on $E$. %
	\idx{associated vector bundle}
In particular there is the \term{adjoint bundle} %
	\idx{adjoint bundle}
$\ad(P) \defeq P \times_G \lie{g}$ where the standard fibre is the Lie algebra $\lie{g}$ of $G$, and the representation is the adjoint representation $\ad:G \to \aut(\lie{g})$.
In the context of gauge theory we will call the associated bundle $E\to M$ the \term{gauge bundle}. %
	\idx{gauge bundle}

If we take the quotient of the tangent bundle $TP$ and the vertical tangent bundle $VP$ by the group $G$, we get the vector bundles $T_GP \defeq TP/G$ and $V_GP \defeq VP/G$ over the base $M$. Now, sections of the bundle $T_GP \to M$ are $G$-invariant vector fields on $P$, and the sections of $V_GP \to M$ are $G$-invariant \emph{vertical} vector fields on $P$. This latter bundle we call the \term{gauge algebra bundle} $V_GP \to M$. The name is a reference to the fact that its fibres are isomorphic to the \emph{right} Lie algebra $\lie{g}_r$ of $G$ -- the tangent vectors at the identity invariant under the right action of the group $G$. The action of the group $G$ on the standard fibres is provided by the adjoint representation.

A \term{connection} $\conn$ on the principal $G$-bundle $P$ is a splitting of the following short exact sequence:%
	\marginnote{The exact sequence is often called the \term{Atiyah sequence} after its introduction in \cite{ATIYAH1957}.}
\[
  \exact{V_GP}{T_GP}{TM} ,
\]
where $TM \isom (P\times_M TM)/G$ (the quotient of the pullback bundle of $TM$ onto $P$ by the action of $G$). Thus a connection can be represented as a $G$-equivariant $T_GP$-valued field on $P$. It gives a distribution of vertical tangent subspaces over the space $P$, so that for each point $p\in P$ there is a unique decomposition $TP = T_hP \oplus T_vP$. We usually assume that this distribution is \emph{smooth}.

For a given point $p \in P$, there is a linear projection $T_{G,p}P \to V_{G,p}P$,
and thus we get a \term{connection $1$-form} %
	\idx{connection!$1$-form}
as a linear map $\omega_{\conn}: T_{G,p}P \to \lie{g}$. From the connection $1$-form we can construct a $\lie{g}$-valued differential $1$-form $A$ on $M$. We can always pull back the form $\omega_{\conn}$ with respect to some \emph{local} section $\Gamma_i$ of the principal bundle $P$, but generally we need to impose a compatibility condition with respect to the transition functions of a covering $\{\U_i\}$ of the base manifold $M$. If we take two such local representatives
$A_i = \Gamma_i^*(\omega_{\conn})$ and
$A_j = \Gamma_j^*(\omega_{\conn})$ defined on the open sets $\U_i$ and $\U_j$ with a nontrivial intersection, they must fulfill the following condition:
\[
	A_j = \ad(\inv{g}_{ij}) A_i + g_{ij}^*(\theta) ,
\]
where $\ad$ is the adjoint representation of the group $G$ on the algebra $\lie{g}$, and $\theta$ is the Maurer-Cartan form on $G$; the transition functions are maps $g_{ij}: \U_i \cap \U_j \to G$. The compatibility condition also serves as the motivation for introducing the group of gauge transformations in Section~\ref{sec:moduli}.

Locally on $M$ we then write the connection $1$-form as %
	\idx{connection!local}
\[
	A = A^{\mu}_a(x) \ext x_{\mu} \, \tau^a ,
\]
where $\tau^a$ are the Lie algebra generators in the adjoint representation, and $A^{\mu}_a$ is a function on the base manifold $M$. Such connections always exist and they form an affine fibre bundle $\A \to M$~\cite[Sec.~2.9]{MORGAN1998}. The sections of this bundle are referred to as \term{gauge connections} %
	\idx{connection!gauge}
or \term{gauge potentials}.
	\marginnote{The use of letter $A$ and the term \emph{potential} harks back to the classical theory of electromagnetic potentials.}
If the principal bundle is \emph{trivial}, we can extend the local potential over the whole base $M$. Since a nontrivial principal bundle does not have global sections, in general the connection $1$-form on $P$ descends to a \emph{family} of local potentials on $M$ (subject to the compatibility condition).
Lastly, we note that a connection on the principal bundle $P$ induces a connection on the associated bundle $E$.

To a connection $1$-form $A$ we can associate a \term{curvature $2$-form $F$} %
	\idx{curvature!$2$-form}
by defining
\[
	F \defeq \ext A + \frac{1}{2}\comm{A}{A} .
\]
	\idx{adjoint bundle}
The curvature is a form in the adjoint bundle $\ad(P)$ associated to the adjoint representation of $G$ on $\lie{g}$: it is a \emph{horizontal form} in the sense that it vanishes identically if one of its arguments is a vertical vector.%
	\footnote{Sometimes terms \emph{tensorial} and \emph{pseudotensorial} are used to characterise their local transformations: a connection form is pseudotensorial while its curvature form is tensorial, since for the latter the local term $g_{ij}^*(\theta)$ in the compatibility condition vanishes.~\cite[Sec.~2.1]{AI1995}}
A connection is said to be \term{flat} if its curvature $2$-form is trivial.

The central idea behind connections is that now we can move between the fibres of the principal bundle over different points in a consistent fashion, and respect the action of the structure group $G$ while doing so. This property enables us, among other things,
to define a covariant differentation on the associated bundle, which then motivates the use of the symbol $\conn$ to denote a connection.

Furthermore, to this principal bundle we associate the \term{group of gauge transformations} %
	\idx{gauge transformation!group}
$\G$ as its vertical automorphisms, a Lie group in itself. Its significance lies in the fact that it also provides a transformation group on the space of connections.~\cite{MV1981} We will take a more careful look at these transformations later in Section~\ref{sec:moduli}.

\begin{remark} %
		\idx{associated vector bundle}
	In many applications we can assume that the associated bundle $E\to M$ is trivial (for example by demanding that $P$ is trivial), so that the group of gauge transformations is isomorphic to the smooth $G$-equivariant maps $\map{M,G}$. The associated bundle is then the product $E = M \times G$, where $G$ acts adjointly on itself.
\end{remark}

\begin{example}[Electromagnetism] %
		\idx{electromagnetism}
	Useful simplifications of the electromagnetic field theory can be formulated in terms of the gauge connection $A$. For instance, consider a contractible Riemannian base manifold $M$ with $\dim(M) = 3$. The mesoscopic (vacuum) equations for \emph{magnetostatics} are
	\[
		\ext B = 0 \quad \text{ and } \ext \star B = J ,
	\]
	where $\star$ is the metric-dependent Hodge operator on $M$, and the $2$-forms $J$ and $B$ are the electric current (density) and the magnetic flux (density), respectively. These condense into one equation by using the identity $\ext A = B$ (guaranteed by the trivial topology of $M$):
	\[
		\ext \star \ext A = J .
	\]

	Note that since $\ext(A + \alpha) = B$ for any closed $1$-form $\alpha$, we have gauge freedom in choosing the connection. This is what is manifested by the symmetry group $G$ in terms of giving gauge equivalent connections.\idx{gauge!invariance}

	More generally, the relativistic electromagnetism in a $4$-dimensional space-time can be formulated in terms of the \term{Faraday-Maxwell tensor} $F$, derived from the gauge connection $F = \ext A$ under the structure group $U(1)$.
\end{example}

\begin{remark} %
		\idx{connection!local}
	The importance of global connections on a principal $G$-bundle must be stressed. If the theory was formulated simply in terms of local potentials $A$, the classical theory could not be properly quantised. For this reason the quantum theory of electromagnetism has to be formulated in terms of a principal $U(1)$-bundle; similarly for more general $G$-invariant gauge theories. Moreover, this brings in a quantisation condition: in the case of $U(1)$-bundles the first Chern class, and for $SU(p)$-bundles the second Chern class, is always an integer. See \cite{SCHREIBER2016} for a general outlook on this aspect of prequantum theories.
\end{remark}

% principal bundle indexing starts
\idxloppu{principal bundle}
% connection indexing starts
\idxloppu{connection}

\subsection{Spin structure and spinors}
\idx{spin structure}

In order to properly describe quantum particles, we need to introduce spinor fields. What we want is a spin bundle carrying an action of the spin group.

	\idx{spin group}
The spin group $\spin(n)$ is the double cover of the rotation group $SO(n)$, defined via short exact sequence
\[
	0 \to \Z_2 \to \spin(n) \to SO(n) \to 0 .
\]
Let $T_{SO(n)}(M)$ then be the oriented orthonormal frame bundle of the tangent space of $M$. By the exact sequence above, we lift this to a principal $\spin(n)$-bundle $P$. We then call the associated complex vector bundle $S\to M$ under a fixed spin representation of $\spin(n)$ the \term{spinor bundle on $M$}, %
	\idx{spinor!bundle}
and we call the sections of this bundle \term{spinor field}. %
	\idx{spinor!field}
	\index{spinor|seealso {fermion}}

\begin{remark}
	Earlier we assumed that the base manifold $M$ admits a spin structure, since there are topological obstructions for the existence of spin bundles. Namely, we must assume that the manifold $M$ is oriented and that its second Stiefel-Whitney class is trivial. See \cite[Ch.~II]{LM1989} for details.
\end{remark}

\subsection{Dirac operators}
% Dirac opetator indexing starts
\idxalku{Dirac operator}
\idxsee{operator!Dirac}{Dirac operator}

In relativistic quantum physics the dynamics of fermion fields $\psi$ are described by the \term{Dirac equation} %
	\idx{Dirac equation}\idx{fermion!field}
$\drc \psi = 0$, which takes the following local form on a Minkowskian space $M$ (using the \emph{natural units} with $\hslash = c = 1$):
\[
	(i \gamma^{\mu} \partial_{\mu} - m) \psi = 0 ,
\]
in which $m$ is a mass term and the coefficients $\gamma^{\mu}$ are the \term{gamma matrices} satisfying the relation $\acomm{\gamma^{\mu}}{\gamma^{\nu}} = 2g^{\mu\nu}I$ with respect to the metric $g$ on $M$. The Hamiltonian anomaly arises only when the mass term in the equation is zero: we call the related massless fields \term{Weyl fermions}. %
	\idx{fermion!Weyl}
In the following we consider mostly the massless case.

We can write this in more general terms with the connections of the bundle $S\otimes E \to M$: Let $\conn^E$ be a connection on $E$ and $\conn^S$ the Levi-Civita connection on $S$. We can then define the total connection $\conn^M$ on $S\otimes E$ so that for any section $\sigma \otimes e$ of the bundle $S\otimes E$, we write $\conn^M(\sigma \otimes e) \defeq \conn^S (\sigma)\otimes e + \sigma \otimes \conn^E(e)$. Then the massless Dirac equation has the following local expression:
\[
	i g^{\mu\nu} \gamma_{\mu} \conn^M_{\nu} \psi = 0 .
\]

For even-dimensional base manifolds one can introduce $\Z_2$-grading on the spinor bundle, and thus decompose it into \emph{left- and right-handed} (sometimes called \emph{positive and negative}) spinors. Let us then write $S = S^L \oplus S^R$. The Dirac operator $\drc$ now maps the left-handed spinors into right-handed spinors, and vice versa, so that the operator itself can be decomposed as
\[
	\drc = \begin{bmatrix} 0 & \drc^R \\ \drc^L & 0 \end{bmatrix} .
\]
We will come back to the properties of the Dirac operator later in Chapter~\ref{chap:ktheory}, where we discuss general operator algebras and the topology of families of Dirac-type operators.

% Dirac opetator indexing ends
\idxloppu{Dirac operator}

\subsection{Moduli space of connections}\label{sec:moduli}
\idxalku{connection!moduli space of}
\idxsee{moduli space}{connection}
\idx{connection!gauge}

A central concept in gauge theory is the invariance under gauge transformations. For the connections $A\in \A$ on the gauge bundle we can define the gauge transformations
\[
	A \mapsto A^g \defeq \inv{g}A g + \inv{g}\ext g ,
\]
where $g \in \G$ is a vertical automorphism on the bundle. (To justify this, see \cite{MS2000} for a standard treatment or \cite[Sec.~2.1]{SCHREIBER2016} for a cohomology argument). Note that locally this reverts back to the compatibility condition of different connection $1$-forms introduced earlier.

A \term{vertical automorphism} %
	\idx{vertical automorphism}
of a principal $G$-bundle $\pi:P \to M$ is a $G$-equivariant diffeomorphism $\varphi:P \to P$ covering the identity on the base, so that $\pi \circ \varphi = \pi$. The set of such automorphisms forms an infinite-dimensional Lie group $\G \defeq \aut_V(P)$ with respect to the usual composition of functions. It acts on the left on $P$ and commutes with the right action of $G$. %\hox{On the associated bundle ...}
Due to its importance in transforming gauge connections, the group $\G$ is called the \term{group of gauge transformations}.%
	\footnote{In literature, this group is often called the \emph{gauge group}, thus potentially mixing with the symmetry group $G$. We will avoid the use of this term altogether for the sake of precision.}
\idx{gauge transformation!group}

In general, the action of $\G$ on the space of connections $\A$ is not free. Therefore we will require that the group $\G$ is \term{based} so that we can define topology on the moduli space of connections. This means that for some fixed point $p \in M$, all transformations $g \in \G$ yield the identity of $G$, $g(p) = e$. Now the action is free and by the quotient manifold theorem we can define the orbit space $\A/\G$ as a manifold with topology compatible with these spaces.~\cite{SINGER1981} %
	\idx{gauge transformation!based}\index{orbit space}
\begin{assumption}\label{assu:gg}
	Let $\G$ be the set of vertical automorphisms on the space of gauge connections $\A$. We assume that as a group $\G$ is based: for a fixed $p\in M$, $g(p) = e$ for all $g \in \G$, where $e\in G$ is the identity element. 
\end{assumption}

\begin{definition} %
		\idx{gauge!orbit}
	The \termd{moduli space of gauge connections with respect to $G$} is the quotient $\A/\G$, where $\G$ is the group of based gauge transformations induced by the group $G$. The moduli space is also often called the \term{gauge orbit space}.
\end{definition}

The importance of the moduli space is that in a certain sense it gives the true configuration space for the physical theories: since the gauge connections are thought to be symmetric, all the relevant physics should manifest already with respect to the quotient space. However, while the topology of the space of connections $\A$ can be assumed to be trivial, there can be severe complications in the structure of the moduli space inherited from the group $\G$.
For this reason one often introduces a priori restrictions such as our demand of a fixed point on the base $M$.~\cite{ACM1989,MORGAN1998}

Under these conditions the space of connections admits the structure of a principal bundle:
\[
	\exact{\G}{\A}{\A/\G} .
\]
Now if this bundle was trivial, that is $\A = \G \times \A/\G$, the homotopy exact sequence of the fibration would induce an isomorphism
\[
	\pi_n(\A) \isom \pi_n(\G)\oplus \pi_n(\A/\G) .
\]
Since $\A$ is an affine space, this would mean that both the moduli space and $\G$ would be trivial, which in general does not hold.
From this we conclude that the topology of the moduli space depends on that of the space $\G$, and that the bundle in general cannot be trivial.%
	\marginnote{This is called the \term{Gribov ambiguity} after \cite{GRIBOV1978}.}
This has serious consequences for fixing a \emph{global} gauge; in other words, choosing a section of the bundle may not be possible.

\begin{remark} %
		\idx{groupoid!gauge}
		\idxsee{gauge!groupoid}{groupoid}
	A classical field theory without internal symmetries can be defined via the set of sections of a fibre bundle over the manifold $M$, for which the necessary properties can be deduced from the local description. In gauge field theories -- classical or quantum -- we need to discuss principal bundles over $M$ with a given symmetry structure group $G$. The gauge theory then hinges on the underlying \emph{groupoid} structure rather than on the simple set of sections. In this groupoid the objects are the connection $1$-forms, and the morphisms are the gauge transformations $\G$. We will come back to groupoids in Section~\ref{sec:cat_groups}, and explore the idea of a gauge groupoid more closely in Example~\ref{ex:gauge_groupoid}. The gist of this approach is that in the groupoid structure the details of the gauge transformations on gauge connections is not lost as easily as with the gauge orbit approach described above.~\cite{BSS2015}
\end{remark}

\idxloppu{connection!moduli space of}

\section{Gauge symmetries}
% gauge symmetry indexing starts
\idxalku{gauge!symmetry}

The group $G$ gives the space of gauge connections its underlying symmetry structure. It turns out that much of the essence of quantum field theories is encoded in this symmetry. The important question is whether observables are invariant under some proposed symmetry, that is, if they transform equivariantly under the suitable group action. From the field theory perspective one is mostly interested in having an invariant \emph{Lagrangian (density)} defining the action functional, which is then used to describe the dynamics of the physical system. On the other hand, an invariant Lagrangian leads to \emph{Noether's symmetry currents}, which can be analysed separately. Historically, these currents also provided the first mathematical description of symmetry anomalies in quantum field theories.

\subsection{Current groups and algebras}
\idxalku{current group}
\idxalku{current algebra}

We will now have the first look at one of the main actors in this thesis: current groups and algebras. Per an encyclopedia definition~\cite{GOLDIN2006}
\begin{quote}
	a \termd{current algebra} is an infinite-dimensional Lie algebra of current density operators augmented with specific commutation relations.
\end{quote}
The origin of the current algebras was in the algebraic formulation of physical currents describing hadrons and, after proven successful, the approach expanded to encompass various special cases. As with finite-dimensional Lie algebras, many (but not all) current algebras have their correspondence in Lie groups: these are then called \term{current groups}. We will tackle the group side first.

\idx{Lie group!of smooth maps}
Given a smooth manifold and a compact Lie group $G$, one can form a space of smooth maps $\map{M,G}$. Under suitable conditions, this space admits a smooth structure ({\eg~ if $M$ locally convex}) and a group structure ({\eg~ if $M$ compact}); see \cite{NEEB2006} for a review on how this leads to locally convex Lie groups. 
The group $\map{M,G}$ is then the Lie group of smooth maps
\[
	g: M \to G : p \mapsto g(p) ,
\]
for which the group structure follows from the pointwise multiplication at $p \in M$:
\begin{enumerate}
	\item $(gh)(p) = g(p)h(p)$,
	\item $\inv{g}(p) = \inv{(g(p))}$,
	\item $e(p) = e$.
\end{enumerate}

In the previous section we noted that the group of gauge transformation can be formulated as a current group if the gauge bundle is trivial. In general, the gauge transformations are always \emph{locally} expressible as current groups; that is, over open sets $\U \subset M$ we can define $\G_{|\U} \defeq \map{\U,G}$. For this reason we could also speak of the group of \term{local gauge transformations}. %
	\idx{gauge transformation!local}

The original idea behind current algebras was that one could describe the local symmetry conditions algebraically, even if the exact formulation of the field theory was not known. Locally, we can define \term{gauge currents} %
	\idx{gauge!current}
as continuous maps in $\map{M,\lie{g}_S}$ of the following form:
\[
	j_{\mu}^a(x) = \psi^{\dagger}(x)(\gamma_{\mu}\otimes\tau^a)\psi(x) ,
\]
where $\tau^a$ are the generators of the gauge algebra $\lie{g}$, and $\gamma_{\mu}$ are gamma matrices: together these generate the extended algebra $\lie{g}_S = \lie{spin}\otimes\lie{g}$. For matrix algebras this extension is straightforward, and it is easy to see how the Lie algebra structure is retained when the Lie bracket is the usual matrix commutator. More generally, a tensor product Lie algebra can be constructed following~\cite{ELLIS1991}.

The gauge currents can be derived from the invariant Lagrangian density (see \cite{JACKIW1972,JACKIW1985} for details) and, importantly, they correspond to a conserved symmetry. %
	\idx{theorem!Noether}
	\footnote{The background for this conservation is in the famous \term{Noether's theorem}: a continuous symmetry of the action corresponds to a current conservation~\cite{NOETHER1918}.}
The space $\map{M,\lie{g}_S}$ has the structure of an infinite-dimensional Lie algebra reminiscent to its group version but, in contrast to the finite-dimensional case, we do not always have such a direct correspondence between the group and the algebra.

In order to serve quantum physics, these spaces are further refined into current algebras by imposing equal-time commutation relations upon the currents. For example, an expected relation for the time-components $\mu = 0$ would be something akin to
\[
	\comm{j_0^a(x)}{j_0^b(y)} = i\lambda_c^{ab}j_0^c(x)\delta^{(2k+1)}(x-y) ,
\]
where $(2k+1)$ is the dimension of the physical space. However, it is well-known that often these commutators need to be amended with additional terms arising as an effect of the quantisation. This is also the origin for the study of Hamiltonian anomalies. %
	\idx{anomaly!commutator}

Important examples include loop groups and their algebras, which we discuss more closely in Chapter~\ref{chap:higher}. Cohomology theories will play a significant role in the application of current groups and algebras, and we will summarise some of the relevant results in Appendix~\ref{app:cohomology}. Along the way we will gain a decent assortment of tools for the analysis of anomalies, as we will see in Chapter~\ref{chap:hanomaly}.

\idxloppu{current group}
\idxloppu{current algebra}

\subsection{Yang-Mills theory}
\idxalku{Yang-Mills theory}

At its core, the Standard Model of particle physics is a \term{Yang-Mills theory}: a gauge theory based on non-Abelian symmetry groups. Let $\ip{\cdot,\cdot}$ be a symmetric bilinear form on the Lie algebra $\lie{g}$ of the gauge symmetry group $G$; assuming adjoint representation, we can write this form as a scalar multiple of the \term{trace operator} $\tr:\lie{g}\to \R$. As before, the gauge potential and its curvature take their values in the Lie algebra $\lie{g}$.
We take the adjoint representation as generated by elements $\tau^a$ subject to the normalisation
\[
	\tr (\tau^a \tau^b) = \frac{1}{2}\delta^{ab} .
\]

The Yang-Mills theory is then characterised by its Lagrangian%
\[
	\mathcal{L} = -\frac{1}{2}\tr (F \wedge \star F) ,
\]
where $F$ is the \term{field strength}, locally the curvature $2$-form derived from \emph{some} gauge potential $A$.
The action functional of the field can be given as
\[
	c_1 \int_M \tr ( F \wedge \star F) + c_2 \int_M \tr (F \wedge F) ,
\]
where $c_i$ are internal constants of the theory: the \emph{coupling constant} and the \emph{theta angle}, respectively. 
All gauge fields in the Standard Model can be modelled as Yang-Mills fields which describe the essential dynamics of the system. For the matter fields one introduces spinorial structure relative to the underlying Yang-Mills theory together with the associated vector bundles; see for instance \cite{JACKIW1985} for a standard treatment.

\idxsee{gauge!invariance}{gauge symmetry}
Regardless of the exact presentation of the field theory, the Lagrangian must be invariant under the local symmetry transformations given by the group $G$. Since the Lagrangian is composed from the curvature $2$-form $F$ rather than the potential $A$, any transformation of the potential which leaves its curvature intact can in principle be accepted. For instance, in electrodynamics the symmetry group $U(1)$ creates a \emph{phase transformation} when applied to the potential and the fermion fields; this, however, will not affect the field strength and hence the Lagrangian is invariant.

In the Standard Model, the symmetry group constituents are $U(1)$, $SU(2)$ and $SU(3)$, roughly corresponding to the three fundamental interactions, and in addition the fermions must obey the gravitational symmetries subject to the Lie group $\spin(1,3)$. It is an open question if this combination of symmetries could be found inside a larger group in such a way that the overall structure could be simplified.

\begin{remark} %
		\idx{mass gap}
	While the classical Yang-Mills theory is well-defined and the quantum version of the theory forms the basis of the Standard Model, there are still open questions regarding its quantisation. The so-called \emph{mass gap problem} is one of the seven Millennium Problems posed by the Clay Mathematics Institute, as it is not yet fully understood how all the observed phenomena relating to Yang-Mills fields coupled to fermion fields would emerge from the mathematical structure. This is especially true for the \emph{color confinement} in quantum chromodynamics describing the interactions within the atomic nuclei.~\cite{JW2006}
\end{remark}

% gauge symmetry indexing starts
\idxloppu{gauge!symmetry}
\idxloppu{Yang-Mills theory}

\section{Quantisation and anomalies}

In physics, the \term{classical observable} %
	\idx{observable}
is defined as a time-dependent smooth function on a given phase space -- the state of the system. The \term{quantum observable} on the other hand is defined as a Hermitian operator on the Hilbert space of the quantum states.%
\footnote{This leads to the distinguishing of the \emph{expectation values} and \emph{observed values} in a quantum system, and the related measurement problem; but we will not need to discuss this further in the given context.}
This idea of an observable is further complicated in quantum field theories, since they should adhere to the relativistic space-time with a causal structure. The proper (nonperturbative) algebraic formulation for these operators is provided by $C^*$-algebras introduced later in Section~\ref{sec:cstar}. %
	\idx{C@$C^*$-algebra}

This process of transforming geometry of fields into algebra of operators -- loosely speaking, since there is no formal procedure but rather guidelines to follow -- is what is called \term{canonical quantisation}, %
	\idx{quantisation!canonical}
since it follows in line with the earlier constructions of the simpler quantum mechanics of particle systems. There are several different attempts to describe quantisation mathematically, to a varying degree of specification; we will not try to summarise these ideas here.

\subsection{Fermionic quantisation}
% fermionic quantisation indexing starts
\idxalku{quantisation!fermionic}

Consider a Dirac operator $D_A$ coupled to an external field via principal $G$-connection $A$.
For a given spinor field configuration $S\otimes E \to M$ we can define \term{one-particle Hilbert space $\Hs$} as the space of square-integrable sections of this product bundle: the space $\Hs$ inherits its inner product from the metrics on the bundles $S$ and $E$. We call \term{fermion fields} %
	\idx{fermion!field}
those sections which fulfill the Dirac equation.

The Hilbert space admits a \term{spectral decomposition} %
	\idx{spectral decomposition}
$\Hs = \Hs^{-}\oplus \Hs^+$. This is derived by fixing a \emph{non-eigenvalue} $\lambda$ of the Dirac operator $D_A$, and then defining $\Hs^{+}$ as the subspace of eigenvectors corresponding to eigenvalues larger than $\lambda$ and taking $\Hs^{-}$ as its orthogonal complement. %\hox{The Fock space should be now written in terms of the decomposition}
Now, for each potential $A$ in the space of connections $\A$ we can form the corresponding \term{fermionic Fock space}\idx{fermionic Fock space} as the direct sum %(\hox{completion?})
	\marginnote{We denote by $\N_0$ the set of non-negative integers.}%
 \[
 	\F_A(\Hs) = \F_A(\Hs^+) \otimes \F_A(\overline{\Hs^-})
		= \bigoplus_{p,q \in \N_0} \forms^p (\Hs^+) \otimes \forms^q(\overline{\Hs^-}), %\quad \text{with } \forms^0 (\Hs^{\pm}) = \C ,
 \]
\emph{antisymmetric} with respect to the tensor product.%
	\footnote{Antisymmetry ensures the \emph{Pauli exclusion principle} of denying two fermions to be in the same quantum state.}

These Fock spaces should carry an irreducible representation of the algebra of \term{canonical anticommutation relations (CAR)}. %
	\idx{canonical anticommutation relations (CAR)}
That is, there are operators $\creat$ and $\annih$ on $\F_A(\Hs)$ such that for all vectors $u,v\in \Hs$ we have
\[
	\acomm{\creat(u)}{\annih(v)} = \ip{u,v}_{\Hs}\idoperator \quad \acomm{\creat(u)}{\creat(v)} = \acomm{\annih(u)}{\annih(v)} = 0 .
\]
This representation is fixed (made \term{quasi-free} %
	\idx{quasi-free representation})
by choosing a \term{vacuum vector} %
	\idx{vacuum vector}
$\psi_0 \in \F_A(\Hs)$ fulfilling the condition
\[
	\creat(u)\psi_0 = \annih(v)\psi_0 = 0 \quad \text{for all } u \in \Hs^- \text{ and for all } v \in \Hs^+
\]
with respect to the spectral decomposition.

These operators $\creat$ and $\annih$ are called \term{creation} and \term{annihilation} operators, respectively. The appearance of spectral subspaces relative to a chosen vacuum energy is a reminiscent of the Dirac sea of \laina{negative} and \laina{positive} energy states as originally envisioned in \cite{DIRAC1928}.
Note that it is exactly this introduction of the anticommutation over a Fock space which turns the fermion fields into operators.

A central question regarding the Hamiltonian anomaly is whether such an assignment of Fock spaces can be made continuously over the moduli space of connections, which is taken as the physical configuration space. That is, \emph{is there a mathematically sound definition for a Fock bundle?} The physical side of this hinges on the possibility of defining a consistent vacuum state. We return to this question in Chapter~\ref{chap:hanomaly}. %
	\idx{Fock bundle}

\begin{remark}
	The Hamiltonian anomaly is described as a quantisation problem with respect to an external classical field -- that is, the gauge field itself is not quantised. There are a couple of reasons for this. First of all, such a model might very well be a realistic simplification of a physical system in which one is interested in the particle dynamics with respect to some externally applied electromagnetic field. External fields are also often useful for computational reasons.~\cite[Ch.~16]{WEINBERG1996}. But most importantly, in the mathematical formulation the anomaly itself appears already with quantised Dirac fields.
	For a discussion on the mathematical difficulties in Dirac quantisation under an external field, see \cite{DM2016}.
\end{remark}

% fermionic quantisation indexing ends
\idxloppu{quantisation!fermionic}

\subsection{Symmetry anomalies} %
\idx{anomaly}

The switch from fields to operators in quantisation may also have geometric consequences. In particular, the proposed gauge symmetry can be broken in the process. For a canonical example, we can consider the \term{chiral anomaly} %
	\idx{anomaly!chiral}
in quantum electrodynamics. The fermion fields in the massless Dirac equation have a chiral symmetry, which can be written as follows:
\[
	\psi(x) \mapsto e^{i\alpha\Gamma}\psi(x) ,
\]
where $\Gamma$ is the chirality operator which anticommutes with the gamma matrices $\gamma_{\mu}$, and $\alpha$ is some constant. For this symmetry we can describe the associated Noether current as %
	\idx{theorem!Noether}
\[
	j_{\mu}(x) = \psi^{\dagger}\Gamma\gamma_{\mu}\psi ,
\]
and, in the massless case, we would expect that this current has zero divergence:%
	\footnote{More generally, there would a non-zero term coming from the mass term in the Lagrangian -- nevertheless, the kinetic part is expected to vanish.}
\[
	\partial^{\mu} j_{\mu}(x) = 0.
\]

However, after the dust has settled from quantisation, we will find that the divergence yields an anomalous term proportional to $F \wedge F$, where $F$ is the Maxwell-Faraday tensor. Historically this was first derived from divergent $1$-loop Feynman diagrams \cite{ADLER1969,BJ1969}, and later on various perspectives have been proposed to explain the anomaly. Perhaps the most interesting mathematically is the functional integral approach in \cite{FUJIKAWA1979}, in which the anomaly is explained as a topological effect and computed using the Atiyah-Singer index theorem. %
	\idx{index theorem}

It can be also shown that the anomaly will appear regardless of the used renormalisation scheme. 
Moreover, the chiral anomaly has actual physical consequences, so it is indeed a crucial component of the physical theory rather than a problem to be solved.~\cite{JACKIW1972,JACKIW1985}

This is not the case with all symmetry anomalies. The chiral anomaly does not break the all-important gauge invariance; the term \emph{external symmetry} is sometimes used to indicate that the chiral symmetry is not a fundamental symmetry of the theory.%
	\footnote{For different flavours of anomalies, see \cite{BERTLMANN2000}.}
On the other hand, the gauge symmetry may be in peril as well: this is the case with the Hamiltonian anomaly, which comes about when quantising massless fermion fields under non-Abelian gauge symmetries. It turns out that the Hamiltonian anomaly destroys the gauge invariance and thus the theory cannot be properly quantised. %
	\idx{anomaly!Hamiltonian}
%

% higher structures in gauge theory
%!TEX root = teesirunko-arxiv.tex
\chapter{Higher structures}\label{chap:higher}

One of the two main conceptual questions in this thesis is what happens to mathematical -- mostly geometric -- structures under the process of quantisation. The other question could be stated as follows: \emph{what are the most suitable structures to begin with?} In this chapter we explore a handful of ideas to this end. In particular, we introduce concepts that involve a \emph{higher} or \emph{categorical} perspective. There is a growing interest in mathematical physics to study such objects in hopes that they would help make the known phenomena conceptually simpler, and perhaps lead to new physics as well. Ultimately one hopes to get a better insight into the mysteries of quantisation, so that at the end of the day these two conceptual questions combine to just one. While even attempting to answer this Great Question properly is well beyond the scope of this thesis, we do try to provide a few glimpses to this direction.

Thus we do not strive to give an in-depth treatment of category theory or higher geometry. Rather, we are motivated by two examples with nontrivial connections to quantum anomalies: namely, categorical groups and bundle gerbes. The first of these gives a context to the earlier results published in \cite{MN2019}, and the second provides a neat geometric handle to the group cohomology aspect of the Hamiltonian anomaly. These two examples are not entirely unrelated, since in both a prominent feature is the Mickelsson-Faddeev extension of current groups on three-dimensional manifolds -- and indeed one of our main motivations in this chapter is to look for a categorical perspective to current groups.

%% Categorification
\section{Categorification and higher objects} %
\idxalku{categorification}

The classical gauge field theory as introduced in the previous chapter is geometrically a theory of connections on principal bundles over topological spaces, describing the (gauge) dynamics of point particles. However, in theories such as string theory or loop quantum gravity one would need a dynamic description for objects that are not point particles but rather $1$-dimensional paths or loops in a space -- or, going further, $k$-dimensional surfaces (or branes, as they are often called). For the gauge theory formalism to extend to such applications it is necessary to extend the basic geometric concepts beyond the point-space foundations. This is one of the motivations for defining what could be called \emph{higher geometry} through the process of \emph{categorification}. Historically, the study of higher categorical objects has also been closely linked to topological quantum field theory~\cite{BAEZ1997}. More recently, higher symmetry groups have been a subject of growing interest not only among gauge theorists, but also in condensed matter physics. In particular, higher groups appear to be relevant in the study of topological phases of matter.~\cite{GKSW2015,SHARPE2015}

One can roughly divide categorification into two flavours: \term{vertical} and \term{horizontal categorification}.\footnote{Horizontal categorification could also be called \term{oidification} since it often leads to objects like groupoids, algebroids, ringoids, and such.} The following diagram illustrates the general idea: in vertical categorification, one extends objects and functions -- algebras, groups and the like -- to categories and functors of such, and further to higher categories; on the horizontal side, one can realise the same objects with a category of a specific type, and then study the general structure behind this instance.
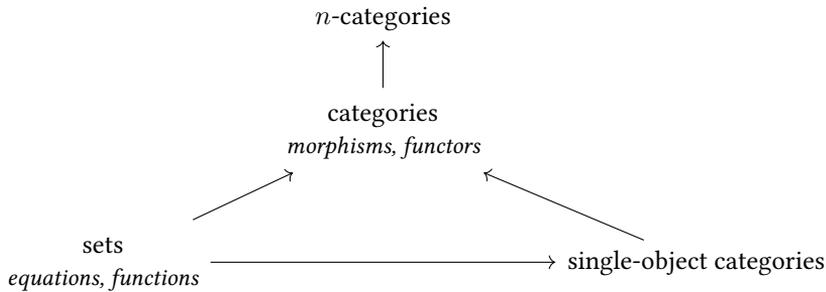
\[
\begin{tikzcd}
  & \text{$n$-categories} & \\
  & \mlnode{categories\\{\itshape{\small morphisms, functors}}}  \arrow[u] & \\
  \mlnode{sets\\{\itshape{\small equations, functions}}} \arrow[rr] \arrow[ur] & & \text{single-object categories} \arrow[ul]
\end{tikzcd}
\]
Put together, one goes not only higher but also \emph{wider}. A good example is the categorification of groups. The corresponding diagram could look like this:
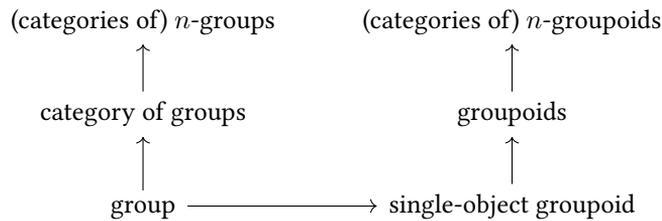
\[
\begin{tikzcd}
  \text{(categories of) $n$-groups} &  \text{(categories of) $n$-groupoids} \\
  \text{category of groups} \arrow[u] & \text{groupoids} \arrow[u] \\
  \text{group} \arrow[r] \arrow[u] & \text{single-object groupoid} \arrow[u]
\end{tikzcd}
\]
  \idx{groupoid}
  \idx{$2$-group}
While groups, groupoids, and their $n$-categorified ancestors are not entirely alien concepts to a working mathematician, this process is not always straightforward nor is it easy to see as to what exactly would the appropriate definitions for the extended entities be like. Already in the categorification process for Lie algebras one runs into rather complicated questions, see for instance \cite{BC2004} and \cite[Sec.~6.5.2]{SCHREIBER2017}.%
  \footnote{Besides specific examples, there seems to be little literature available on the categorification as such. Here we have mostly followed \cite{BD1998} and the more informal ideas found in \cite{NLAB:CAT}.}

Since the horizon of categorification is rather vast, we do not dwell on the general ideas too long. A bit more space will be dedicated to categorical groups from which one may gain a higher perspective to symmetry groups and related anomalies. 
However, before the discussion on the categorical groups, we argue that there is also some sense in speaking of higher objects and structures without explicitly diving into category theory as such.

\subsection{What is it like to be a higher object?} %
\idx{higher object}

Categorification described above is not necessarily the only way to catch a glimpse of topology or geometry that is in some sense higher.

One characteristics of the generalisation of the point spaces is that of a more complex homotopy on the base. From this perspective it makes sense to call higher objects those reflecting a higher homotopy type of the space. Recall that a space $X$ of a homotopy type $n$ is characterised by the groups $\pi_k(X)$ with $k\leq n$, in addition to the demand that all higher homotopy groups are trivial. For instance, the $2$-group encountered below is a group structure on a homotopy $1$-type objects -- thus the generalised group structure equivalences must preserve isomorphisms of the homotopy groups $\pi_0$ and $\pi_1$, which is not a property of the standard definition of groups. Similarly, the bundle gerbe is an instance of a Lie groupoid coming either from a principal $2$-bundle or a groupoid extension -- in the former case, this is a principal $U(1)$-bundle structure over a $2$-group (homotopy $1$-type). %\hox{And the latter?} \hox{Cite Schreiber}

In Chapter~\ref{chap:gauge} we implicitly introduced a categorical object: the space of gauge connections and their transformations. This is an example of a groupoid -- and since the gauge transformations themselves are homotopies of other geometric objects, it is natural to consider them in the categorical framework.

We will later see in our examples that higher objects can be well defined without explicitly speaking about categories or such. This indeed is one of the aims: to \emph{concretely} build objects that fulfill the necessary properties without making the process itself unduly abstract. It is the application of category-theoretical concepts that drives us, not so much the machinery itself.

\idxloppu{categorification}
%

%% Categorical groups
\section{Categorical groups and symmetries}\label{sec:cat_groups}
\idxalku{categorical groups}

We introduce two main components in this section: groupoids and $2$-groups as categorification of groups. Groupoids have already appeared in various flavours in mathematics -- also in gauge theory, as argued above -- while $2$-groups have gained more recent interest in mathematical physics. We discuss the necessity of the $2$-group structure when it comes to string theory, and in the next section construct two examples. We refer to \cite{SHARPE2015} for a more general exposition to categorical groups in quantum field theory.

\subsection{Categorical groups}

We begin the categorification of the group structure in the horizontal direction by introducing groupoids.
\begin{definition} %
    \idx{groupoid}
  A \termd{groupoid} is a small category in which all morphisms are isomorphisms.
\end{definition}
Recall that in a \term{small category} %
  \idx{category!small}
both the objects and the morphisms form proper sets. A single-object groupoid would then take us back to the usual definition of a group, with morphisms serving as the group elements. The definition naturally extends to include smooth structures, and so we get \term{Lie groupoids} %
  \idx{groupoid!Lie (smooth)}
(and Lie algebroids, see \cite{MACKENZIE2005}).

The canonical example of a groupoid is the fundamental groupoid of a topological space: the fundamental group $\pi_1(X,p)$ of a space $X$ based at $p\in X$ is the automorphism group of $p$ in the fundamental groupoid $\Pi_1(X)$ -- thus the fundamental groupoid can encode the information provided by fundamental groups without a singled-out base point. The objects of the fundamental groupoid are points in $X$, and its morphisms are homotopy classes of paths between two points. The groupoid axioms are fulfilled by defining composition as concatenation of paths; the rest follows from the properties of the homotopy equivalence.

We have already mentioned that the space of gauge connections forms a groupoid. What we mean is the following.
\begin{example}[Gauge groupoid]\label{ex:gauge_groupoid} %
    \idx{groupoid!gauge}
  Let $\A$ be the space of connections on a gauge bundle over a manifold $M$, and let $\G$ be the group of gauge transformations. The transformations $A \mapsto A^g$ given by a $g\in \G$ are morphisms in the Lie groupoid of local connections $A \in \A$. In other words, for an open set $\U \subset M$, the objects of the \term{gauge groupoid} are the smooth Lie algebra valued $1$-forms $A \in \map{\U,\lie{g}}$, and the morphisms are derived from the smooth $G$-valued functions $g \in \map{\U,G}$ by imposing the gauge principle %
    \idx{gauge transformation}
  \[
    A^g \defeq \ad_g(A) - g^*(\theta_G) ,
  \]
  where $\theta_G$ is the Maurer-Cartan form on $G$. Often in the context of matrix groups the pullback is written as $g^*(\theta_G) = \inv{g}\ext g \in \map{\U,\lie{g}}$.
\end{example}

Another good example, which we will also use later on, is that of \term{action groupoid}. %
  \idx{groupoid!action}
Let $G$ be a group acting on a set $X$. The objects in the action groupoid $\actg{X}{G}$ are the elements of $X$, and the morphisms between two given elements $a$ and $b$ in $X$ are the group elements $g\in G$ such that $g.a = b$; the composition of morphisms is given by the multiplication in $G$. The notation $\actg{X}{G}$ follows from the idea that the action groupoid is in some sense a \emph{weak quotient}: in the proper quotient $X/G$ the elements in the same orbit subset of $X$ are considered equal, whereas in the action groupoid they are considered merely isomorphic.

\idxalku{$2$-group}
In the vertical direction the group structure is generalised by introducing $2$-groups as categories with specified functors analogous to the group operations. One of the many ways to think about $2$-groups is to view them as a group object in the category of groupoids.%
  \footnote{See \cite{FORRESTERBARKER2002} for a review on various equivalent formulations of strict $2$-groups.}
\begin{definition}
  A \termd{strict $2$-group} is a monoidal category in which all morphisms are isomorphisms, and every object has an inverse.
\end{definition}
A \term{monoidal category} %
  \idx{category!monoidal}
is a category $\Ca$ with a \term{product functor} $\otimes:\Ca\times \Ca \to \Ca$ such that it contains
\begin{enumerate}
  \item a unit object $1 \in \Ca$,
  \item an associator isomorphism $\alpha: (x\otimes y) \otimes z \to x \otimes (y\otimes z)$,
  \item left/right unitor isomorphisms $l : 1 \otimes x \to x$ and $r: x \otimes 1 \to x$,
\end{enumerate}
together with the \term{triangle} and \term{pentagon identities} given by the following commuting diagrams:~\cite{BL2004}
% diagram begins
\[
   \begin{tikzcd}[row sep=large]
     (x \otimes 1 ) \otimes y \arrow{dr}[swap]{r\otimes 1} \arrow{rr}{\alpha} & & x \otimes (1  \otimes y ) \arrow{dl}{1\otimes l} \\
        & x \otimes y &
   \end{tikzcd}
\]
% diagram ends
and
% diagram begins
\[
   \begin{tikzcd}[column sep=-2em, row sep=huge, inner sep=0pt]
     & & (w\otimes x)\otimes(y\otimes z) \arrow{drr}{\alpha} & & \\
     ((w \otimes x ) \otimes y ) \otimes z \arrow{urr}{\alpha} \arrow{dr}[swap]{\alpha \otimes 1} & & & &  w \otimes (x \otimes (y \otimes z)) \\
     &  (w \otimes (x \otimes y)) \otimes z \arrow{rr}{\alpha} & &  w \otimes ((x \otimes y) \otimes z) \arrow{ur}[swap]{1 \otimes \alpha}&
   \end{tikzcd}
\]
% diagram ends

Hence a strict $2$-group is a groupoid equipped with a product respecting conditions similar to those of an ordinary group multiplication -- the product is associative, there is a unit object, and objects have inverses under the product. For a method to construct a $2$-group from ordinary groups we refer to the next section in which we discuss the equivalent formulation as crossed modules and give concrete examples of strict $2$-groups.

\idx{$2$-group!Lie (smooth)}
Again, introducing further smoothness conditions gives Lie $2$-groups: we then require that both the objects and the morphisms in the $2$-group carry the Lie group structure. There is a categorical counterpart to Lie algebras as well, and these are called Lie $2$-algebras in the case of Lie $2$-groups. In contrast to ordinary Lie groups and algebras, the relationship between the two is not that straightforward; in particular, given a Lie group $G$ there is a whole family of Lie $2$-algebras arising from the categorification of the Lie algebra $\lie{g}$ of $G$, and only one of them has a natural Lie $2$-group counterpart~\cite{BSCS2007}.
We will not discuss Lie $2$-algebras further, since we do not have any immediate use for them in what follows.

\begin{remark} %
    \idx{$2$-group!coherent}
    \idx{$2$-group!weak}
  For the strict $2$-group, the inverse operations hold as equalities of functors, and the group laws hold strictly as equations. Since the construction is categorical, the axioms can be defined also up to coherent isomorphisms. Thus the definition is \emph{strict} as opposed to \emph{weak} -- in the latter case, given an object $g$ there is a weak inverse $h$ such that $g\otimes h \isom 1 \isom h\otimes g$. Such categories are called weak $2$-groups. Another interesting category is formed by \emph{coherent $2$-groups}, which are weak $2$-groups in which all objects have a special weak inverse subject to certain coherence laws. These concepts are naturally more general, and may well turn out be the right choice of group categorification for physical applications as well.~\cite{BL2004}
\end{remark}

\begin{remark}
  These definitions generalise to $n$-groups and $n$-groupoids. If a $2$-group is a category object in the category of groups, then a $3$-group would be a $2$-category object in the category of $2$-groups, and so on. Analogously, an $n$-groupoid is an $n$-category with all morphisms equivalent. Note that this \emph{equivalence} can mean different things depending on the level of $n$ -- with $1$-groupoids as above it is simply an isomorphism.
\end{remark}

\idxloppu{$2$-group}

The usefulness of categorical groups may be illustrated by considering a basic concept from theoretical physics: the string group.

\subsection{String group and the need for categorical groups} %
  \idx{string group}
  \idx{spin group}
  \idx{spin structure}

Let $M$ be a manifold. The general structure group for the frame bundle of the manifold $M$ is $O(n)$, where $n = \dim M$. We can refine this structure as follows:
\begin{enumerate}
  \item If $M$ is orientable -- the first Stieffel-Whitney class vanishes -- the structure group is $SO(n)$;
  \item If the second Stieffel-Whitney class vanishes, $M$ has a spin structure and the structure group is $\spin(n)$;
  \item If the first fractional Pontryagin class vanishes,
   $M$ has a string structure and the structure group is $\strng(n)$.
\end{enumerate}

The spin group $\spin(n)$ is well-known as the double cover of the rotation group $SO(n)$, defined via the short exact sequence
\[
	0 \to \Z_2 \to \spin(n) \to SO(n) \to 0 .
\]
The string group comes by a similar construction. Consider the short exact sequence
\[
	0 \to K(\Z,2) \to \strng(n) \to \spin(n) \to 0 ,
\]
where $K(\Z,2)$ is an Eilenberg-MacLane space.\footnote{See \cite[Ch.~9]{HJJS2007} for details on the Eilenberg-Maclane spaces.}
The string group builds from the orthogonal group $O(n)$ via a sequence given by the relevant Whitehead tower:
\[
	\dots \to \strng(n) \to \spin(n) \to \spin(n) \to SO(n) \to O(n)  .
\]
If we take the orthogonal group and kill its zero homotopy, we get the special orthogonal group. Likewise, killing the first homotopy of the special orthogonal group we gain the spin group. If we extend this to the third homotopy group (the second being trivial already), the result is the string group.%
  \marginnote{Higher up in the tower is the \term{five-brane group} obtained by eliminating the seventh homotopy group; this structure is significant in the M-theory.}
Hence as the spin group is the $1$-connected cover of the special orthogonal group, so is the string group the $3$-connected cover of the spin group.~\cite{SP2011}

As the Whitehead tower illustrates, 
the string group as an abstract (homotopy) object cannot be realised as a finite-dimensional Lie group since by definition it should have a closed subgroup with non-trivial second homotopy, which by Cartan's theorem is not possible for finite-dimensional Lie groups. Various approaches to model the string group more concretely seem to orbit around the idea that whatever the realisation, it is a specific representation of a higher object. Hence a $2$-group object makes a good candidate for being a concrete model. An example of this will be seen in the next section.

\idx{loop space}
The string structure is closely connected to \emph{loop spaces}, which are topological spaces formed from maps $S^1 \to M$ (we will explore the related loop groups below). While loop spaces play a central role in the study of string theories, there are definite shortcomings in trying to describe (string) gauge theory with principal bundles over loop spaces -- for instance, coupling of strings with gauge fields cannot be fully realised without introducing extra structure. This, too, points towards higher geometry as a more appropriate context.~\cite{WALDORF2015}

\idxloppu{categorical groups}

\section[Strict $2$-groups and $n$-loop group extensions]{Strict \mots{{2}}-groups and \mots{{n}}-loop group extensions} %
\idx{$2$-group}

Having motivated the introduction of categorical groups, let us now concretely build two examples. The aim is to show how $2$-group objects emerge from extensions of loop groups and justify why this is a useful approach. This section is partly based on the article~\cite{MN2019}.

Before jumping into categorical groups proper, we discuss loop groups to some detail -- these groups will be the basis for the presented examples and play an important role in many gauge theory applications.

\subsection[Loop groups and $n$-loop groups]{Loop groups and \mots{\boldsymbol{n}}-loop groups} %
\idxalku{loop group}
\idxalku{current group}
\idx{loop group!$n$-loop}
\idx{loop group!based}
%

%Loop groups and $2$-group models.
Let $G$ be a compact Lie group. The \term{free loop group $\loopg{G}$ of $G$} is the group of smooth maps
\[
	\gamma : S^1 \to G ,
\]
in which the group multiplication is defined pointwise. A \term{based loop group $\bloopg{G}$} is a loop group with a fixed base point: $\gamma(1) = e \in G$ for all $\gamma \in \bloopg{G}$. Geometrically the free loop group $\loopg{G}$ is a principal $\bloopg{G}$-bundle over the space $G$. An important property of the free loop group (and inherited by the based group $\bloopg{G}$) is the existence of nontrivial central extensions by the circle group $S^1 \isom U(1)$:
\[
  \exact{S^1}{\gext{\loopg{G}}}{\loopg{G}} .
\]
As a topological space the based loop group is dual to the suspension $\Sigma G$ of the space $G$, when assuming compact-open topology. The standard result in homotopy theory then states that in the connected components
\[
  \pi_{k+1}(G) = [S^{k+1},G] \isom [\Sigma S^k,G] \isom [S^k,\bloopg{G}] = \pi_k(\bloopg{G}) ,
\]
following from the sphere suspension map $S^{k+1} \isom \Sigma S^k$. Furthermore, if the group $G$ is simply connected, so is $\loopg{G}$.~\cite[Ch.~4]{PS1986}

The loop groups generalise to groups of smooth maps $S^n \to G$ of arbitrary dimensions $n\in \N$. We call the analogous groups formed this way the \term{free $n$-loop group $S^{n}{G}$ of $G$} and the \term{based $n$-loop group $\bnloopg{n}{G}$}. There is furthermore an analogous $n$-fold suspension map such that the following holds as a homotopy classification~\cite[Ch.~5]{MAY1972}:
\[
  [\Sigma^nS^{k},G] \isom [S^k,\bnloopg{n}{G}] .
\]
From this it follows that $\pi_{k+n}(G) \isom \pi_{k}(\bnloopg{n}{G})$ in the connected components.%

\begin{remark}
  As noted in Chapter~\ref{chap:gauge}, the mapping groups $\map{M,G}$ are meaningful as gauge transformation groups (under the assumption of a trivial gauge bundle) or current groups, 
  and one often considers spherical spaces $M \defeq S^n$ as the base manifold. The $1$-loop group in particular is important in the study of low-dimensional quantum field theories as well as in string theory, and there exists a sizeable body of mathematical theory for the $1$-loop group -- see \cite{PS1986} for the canonical reference. However, not much is known about the mapping groups for general $M$.
\end{remark}

\idxloppu{loop group}
\idxloppu{current group}

\subsection{Crossed modules}
\idxalku{crossed module}

The central point in this section is to construct practical examples of $2$-groups. We do this by considering objects called \emph{crossed modules}, which are algebraic counterparts to strict $2$-groups. The usefulness of this lies in the fact that crossed modules can be built rather concretely from suitable ordinary groups and homomorphisms between them.

We recall the basic definition:
\begin{definition}
  Let $G$ and $H$ be groups, and consider morphisms
  \[
  	\delta : H \to G \quad \textrm{and} \quad \alpha : G \to \aut(H) .
  \]
  We say that  $[\delta:H \to G]$ is a \termd{crossed module} if the following two diagrams commute.
  % diagrams begin
  \[
  \begin{tikzcd}
  H \times H \arrow[rd,"\ad"] \arrow[r,"\delta\times\id"] & G \times H \arrow[d,"\alpha"]\\
  	& H
  \end{tikzcd}
  \qquad
  \begin{tikzcd}
  G \times H \arrow[d,"\id\times\delta"] \arrow[r,"\alpha"] & H \arrow[d,"\delta"] \\
  G \times G \arrow[r,"\ad"] & G
  \end{tikzcd}
  \]
  % diagrams end
  Equivalently, if we denote by $h^g$ the element-wise action of $G$ on $H$, the diagrams correspond to the equations
  \[
  	h^{\delta(h')} = \inv{h'} h h'
  \]
  and
  \[
  	\delta(h^g) = \inv{g} \delta(h) g
  \]
  for all $h,h' \in H$ and $g \in G$.
\end{definition}
The following theorem then gives a way to associate strict $2$-groups with crossed modules~\cite{BS1976}:
\begin{theorem}
  The categories of crossed modules and of strict $2$-groups are equivalent. This extends to the level of homotopy, that is, also the $2$-categories of crossed modules and of strict $2$-groups are equivalent.
\end{theorem}
\begin{remark}
  In \cite{BS1976} strict $2$-groups are treated as $G$-groupoids, that is, group objects in the category of groupoids. This is equivalent to our earlier definition.
\end{remark}

Note that by the definition any central extension of groups
\[
  \exact{A}{H}{G}
\]
with action $\alpha:G \to \aut(H)$ gives a crossed module through the epimorphism $\delta:H \to G$. We will use this property below in the examples of constructing $2$-groups from loop groups.

Since we consider mostly Lie groups, it is also necessary to verify that crossed modules can retain the smooth structure.
\begin{definition} %
    \idx{crossed module!Lie (smooth)}
  If the groups $G$ and $H$ in a crossed module $[\delta:H\to G]$ are Lie groups and the action defined by the morphism $\alpha$ is smooth, the crossed module is called a \termd{Lie crossed module}, or a \termd{smooth crossed module}.
\end{definition}
In what we follows we work with \emph{smooth group} and \emph{continuous algebra} cohomology, see Appendix~\ref{app:cohomology}.

\idxloppu{crossed module}
%

% example
\begin{example}[Loop group extension and quasi-periodic automorphisms~\cite{MRW2017}]

Let $G$ be simple, compact and simply connected. Let $\bloopg{G}$ be the group of based loops, and $\flatloopg{G}$ a \term{flattened subgroup} defined by imposing the condition of vanishing derivatives at the base point. There are well-known non-trivial central extensions~\cite[Ch.~4]{PS1986} %
  \idx{Kac-Moody!group}
  \marginnote{The group $\gext{\bloopg{G}}$ is often called the \term{affine Kac-Moody group}.}
\[
  \exact{S^1}{\gext{\bloopg{G}}}{\bloopg{G}}
\]
which naturally restrict to central extensions of the flattened subgroup $\flatloopg{G}$. This restriction can be done on the Lie algebra level as well, and the corresponding $2$-cocycle $\kappa_{\flat} : \flatloopg{\lie{g}} \times \flatloopg{\lie{g}} \to \R$ is simply the restriction of the Kac-Moody cocycle %
  \idx{Kac-Moody!cocycle}
\[
  \kappa(u,v) \defeq \frac{1}{2\pi}\int_0^{2\pi} \ip{u(t), v'(t)} \ext t .
\]
Here $\ip{\cdot,\cdot}$ is an invariant bilinear symmetric form on $\lie{g}$,%
  \footnote{A form $c$ on a Lie algebra is \emph{invariant} if it fulfills the condition $c(\comm{u}{v},w) = c(u,\comm{v}{w})$.}
which can be taken as a scalar multiple of the Killing form if $\lie{g}$ is simple. In a matrix algebra we can write $\ip{u,v} = k \tr (u v)$, and hence
\[
  \kappa(u,v) = k \int_{S^1} \tr u \ext v ,
\]
where a suitable normalisation by $k \in \C$ is assumed.

While the flattened loop group $\flatloopg{G}$ is a subgroup of the based loop group and the free loop group, all of them are \emph{subspaces} of a larger space of smooth maps, that of \term{quasi-periodic paths}:
\[
  \qloopg{G} \defeq \{ \gamma \in \map{\R,G} :~ \gamma(t+1) \cdot \inv{\gamma(t)} \text{ is constant} \}  .
\]
Note that $\qloopg{G}$ is not a group with respect to the point-wise multiplication, since the quasi-periodicity is not a property that is preserved under such an operation; moreover, there is no Lie group structure that would be compatible with the standard structure of the free loop group, and it cannot be realised as a Lie group extension of $G$ with the fibre $\loopg{G}$.

However, we can impose the same flattening condition to the quasi-periodic paths and thus gain an honest group:
\[
  \flatqloopg{G} \defeq \{ \gamma \in \qloopg{G} :~ \gamma^{(k)}(0) = \gamma^{(k)}(1) = 0 \quad \forall k \in \N \} .
\]
Since any quasi-periodic path is uniquely determined by considering the interval $t \in [0,1]$, this flattening at the end points is sufficient to ensure the group structure by point-wise multiplication. Then it holds that $\flatloopg{G} = \flatqloopg{G} \cap \bloopg{G}$. Finally, there is a subgroup of based paths:
\[
  \flatqbloopg{G} = \{ \gamma \in \flatqloopg{G} :~ \gamma(0) = e \} .
\]

We can now construct a crossed module based on the groups $\gext{\flatloopg{G}}$ and $\flatqbloopg{G}$. The group of quasi-periodic flattened paths $\flatqbloopg{G}$ acts on $\flatloopg{G}$ pointwise -- this is the right adjoint action $u.\gamma$ for $\gamma \in \flatqbloopg{G}$ and $u \in \flatloopg{G}$. The effect of this action on the Lie algebra cocycle can be written as
\[
  \kappa_{\flat}(u.\gamma,v.\gamma) = \kappa_{\flat}(u,v) + \int_{S^1} \tr \comm{u}{v} \ext \gamma \inv{\gamma} .
\]
Then by the lifting theorem~\cite[Theorem~V.9]{MN2003}
the right conjugation action of $\flatqbloopg{G}$ on $\flatloopg{G}$ lifts to a unique action on the extension $\gext{\flatloopg{G}}$ which is trivial on the centre $S^1$. Since there is a natural inclusion $\flatloopg{G} \hookrightarrow \flatqbloopg{G}$, we also gain a morphism $\delta: \gext{\flatloopg{G}} \to \flatqbloopg{G}$. This together with the action
$\alpha: \flatqbloopg{G} \to \aut(\gext{\flatloopg{G}})$ gives a smooth crossed module.
\end{example}

\begin{remark} %
  \idx{loop space}
Note that the inclusion $\flatloopg{G} \hookrightarrow \flatqbloopg{G}$ together with the corresponding conjugation action would already yield a smooth crossed module. The interest in the central extensions comes from the idea that the spin group given by the extension
\[
  \exact{\Z_2}{\spin(n)}{SO(n)}
\]
should have a string-theoretic analogue. A crucial component in this is the Kac-Moody extension for the free loop group: one begins with a principal $\loopg{G}$-bundle over the loop space $\map{S^1,M}$ of the manifold $M$ and wishes to lift this to a principal $\gext{\loopg{G}}$-bundle. The group $\gext{\loopg{G}}$ can be here interpreted as the group $\spin(n)$, and this lifting can be used as the definition of the string structure.~\cite{KILLINGBACK1987}
\end{remark}

\idx{string group}
Much of the discussion in \cite{MRW2017} is devoted to how this approach gives a coherent (as opposed to strict) $2$-group model for the string group. The value of this particular construction is that -- since it is built upon the smooth free loop group -- it naturally admits the action of the circle group $S^1 \isom SO(2)$. This is important for the definition of $S^1$-equivariant operators on the loop space (which itself is a difficult problem, see \cite{WITTEN1988} for a discussion on the Dirac operator). However, there is an obstruction to using the free loop group as the basis for a strict $2$-group model; one way out of this is to use the group of flattened paths as in the example above. This gives a strict $2$-group which has an underlying action groupoid \emph{equivalent} to the one constructed from the free loop group. This action groupoid is given by the weak quotient $\actg{\qloopg{G}}{\gext{\loopg{G}}}$.
From this equivalence it follows that there is an essentially unique $2$-group structure that fulfills the wanted properties of the string group.

\subsection[$3$-loop group extensions]{\mots{\boldsymbol{3}}-loop group extensions} %
\idxalku{loop group!$3$-loop}

The $2$-group construction from the previous example raises the question whether similar construction would work on general $n$-loop groups. A generalisation to the $3$-loop group $\threeloop$ can be built using a similar method as for the $1$-loop case. This subsection is mostly a summary of the article~\cite{MN2019}, although in places we have expanded the treatment (most importantly in Proposition~\ref{prop:3loop-homotopy}).

The central extension of the loop group
\[
	\exact{S^1}{\gext{\bloopg{G}}}{\bloopg{G}}
\]
can be initially replaced with an Abelian extension
\[
	\exact{\map{\threeloop, S^1}}{\gext{\threeloop}}{\threeloop} .
\]
On the level of Lie algebra this extension is characterised by the \term{Mickelsson-Faddeev cocycle}: %
  \idx{Mickelsson-Faddeev cocycle}
\[
	\theta(A;x,y) = c_2 \int_{S^3} \tr A ~[\ext x , \ext y ] ,
\]
where $c_2 \in \C$ is the normalisation coefficient. %
In the gauge-theoretic framework as in \cite{MICKELSSON1987}, the function $A$ is a gauge connection on the gauge bundle of the symmetry group $G$, a $\lie{g}$-valued $1$-form on the domain $S^3$. As a useful identification we can set $A \defeq \inv{f}\ext f$ for a given $f\in \threeloop$, that is, $A$ is the (pullback of the) left-invariant Maurer-Cartan form.
This is based on the homotopy equivalence of smooth paths in $S^2G$ and the moduli space of connections over $S^3$, see \cite{SINGER1981}; we use the same approach later in Section~\ref{sec:general_comm} when discussing the Hamiltonian anomaly.

Now the idea is essentially the same as in \cite{MRW2017}: a strict $2$-group can be obtained by forming a smooth crossed module of relevant Lie groups. However, the Mickelsson-Faddeev extension is not central and this complicates the process. Consequently the definition of the crossed module will be more involved, since it would naturally entail a central extension and will not work out of the box with an Abelian extension.

Let us assume initially that we use a similar flattened group of maps on the $3$-ball $B^3$ and define it such that all the radial derivatives vanish:
\[
  \threeball \defeq \{ \gamma \in \map{B^3,G} :~ \partial_r^{k} \gamma (\bdr B^3) = 0 \quad \forall k\in \N \} .
\]
There is a natural automorphic action $\threeball \to \aut(\threeloop)$, but there are two issues: first, the lift to the extension will not be possible, and second, there is no group morphism $\gext{\threeloop} \to \threeball$ that would fulfill the requirements of the crossed module.

Let us redefine the $3$-loop group as follows:
\[
  \threeloop \defeq \{ \gamma \in \threeball :~ \gamma \textrm{ extends smoothly to } S^3 \textrm{ and } \gamma(\bdr B^3) = e \} .
\]
This redefinition does not change the nature of the Abelian extension above, and the corresponding Lie algebra extension is characterised by the Mickelsson-Faddeev cocycle without modifications. Note that if we now wish to identify $A \defeq \inv{f}\ext f$, we must take $f\in \threeball$.

Following Proposition~\ref{proposition:aut_lift} in Appendix~\ref{app:cohomology}, in order to lift the action of $\threeball$ to the extension $\gext{\threeloop}$ we would want the $2$-cocycle of the extension to be cohomologically invariant under such actions. It is straightforward, if a bit tedious, to check that the \emph{Lie algebra} $2$-cocycle $\theta(A;x,y)$ is cohomologically invariant under the right adjoint action of $\threeball$ on $\threeloop$. Conjugating with an arbitrary element $g\in\threeball$ one gets
\[
  \theta^g(A;x,y) = \theta(A;x,y) - \delta(\lambda)(A;x,y)
\]
where $\lambda$ is the $1$-cochain
\[
  \lambda(A;z) = c_2 \int_{S^3} \tr \left( A \omega \comm{\omega}{z} + \comm{\omega}{A} \ext z + \omega^3 z \right) ,
\]
with $\omega \defeq \ext g \inv{g}$ the right-invariant Maurer-Cartan form.

The same cohomology invariance holds (smoothly) on the group level \emph{under the assumption that the first homotopy group $\pi_1(\threeloop) \isom\pi_4(G)$ is trivial} -- see Theorem~\ref{theorem:cohom_hom} in Appendix~\ref{app:cohomology}. In any case, the invariance holds in the connected components of the $3$-loop group $\threeloop$. We can state this preliminary result as follows:
\begin{lemma}
  The Lie algebra $2$-cocycle $\theta \in \coh^2\left(\threeloopalg,\map{\threeball,i\R}\right)$ given by
  \[
    \theta(A;x,y) = c_2 \int_{S^3} \tr A \comm{\ext x}{\ext y}
  \]
  is cohomologically invariant under the right adjoint action of $\threeball$. If the first homotopy group $\pi_1(\threeloop)$ is trivial, the cohomology invariance lifts to the group $\threeloop$ with respect to the automorphic action $\threeball \to \aut(\threeloop)$.
\end{lemma}

This action does not lift to the extension, however. There is a natural action on both the base and the fibre, but the map $\threeball \to \aut(\gext{\threeloop})$ obtained by combining these is not a group homomorphism. It is therefore not enough that the defining cocycle is smoothly cohomologically invariant; compare to Proposition~\ref{proposition:aut_lift} in Appendix~\ref{app:cohomology}. Hence we need to revise the group extension in order to get a crossed module out of it.

\begin{remark}
  The homotopy group $\pi_4(G)$ vanishes for the groups $G = SU(p)$ when $p > 2$. In the case of $SU(2)$ we have a non-trivial extension of $\threeloop$ by $\Z_2$ (the Mickelsson-Faddeev extension would be trivial), and the right adjoint action of $\threeball$ lifts to this extension without complications.
\end{remark}

\subsubsection{Revised extension and the crossed module}

Let us first tackle the problem of the proposed morphism $\delta: \gext{\threeloop} \to \threeball$. We can extend the acting group $\threeball$ as follows:
\[
  \exact{\threemaps/S^1}{\gext{\threeball}}{\threeball} .
\]
This is a topologically trivial Lie group extension, and on the level of Lie algebra it is essentially characterised by the same Mickelsson-Faddeev cocycle $\theta$, with the following modification. Since now the integration domain for the $2$-cocycle is $B^3$ instead of $S^3$, we need to take the fibre modulo $S^1$. Then the coboundary $\cob \theta$ of the cocycle has a constant boundary term:
\[
  c_2 \int_{\bdr B^3} \tr \left(x\comm{\ext y}{\ext z} - y \comm{\ext z}{\ext x} + z\comm{\ext x}{\ext y} \right) ,
\]
which vanishes by this construction. Thus $\theta(A;x,y)$ fulfills the cocycle condition in this case as well.

Furthermore, there is a central extension
\[
  \exact{S^1}{\gext{\threeloop}}{\gext{\threeloop}/S^1} ,
\]
where the group $\gext{\threeball}$ acts trivially on the centre $S^1$. Indeed, $\gext{\threeloop}$ is a normal subgroup of the group $\gext{\threeball}$: the group $\threeloop$ is a normal subgroup of $\threeball$, and the fibre $\map{\threeball,S^1}$ of the acting extension is mapped onto itself under the conjugation action.

With these extended groups we have a natural morphism $\delta: \gext{\threeloop} \to \gext{\threeball}$ given by the projection
\[
  \gext{\threeloop} \to \gext{\threeloop}/S^1
\]
and the inclusion
\[
  \gext{\threeloop}/S^1 \to \gext{\threeball} .
\]

Next we need the smooth action morphism: does the action of $\gext{\threeball}$ on the Lie algebra $\gext{\threeloopalg}$ lift to a smooth action on $\gext{\threeloop}$? The answer is yes, \emph{if the group $\gext{\threeloop}$ is simply connected}. Given any simply-connected Lie group $H$, there is a group isomorphism~\cite[Cor.~9.5.11, p.~341]{HN2012}
\[
  \aut(H) \isom \aut(\lie{h}) .
\]
Furthermore, the action of $\aut(H)$ on $H$ is smooth if the automorphism group is endowed with a compatible smooth structure and the action on the Lie algebra level is smooth.
These conditions are easily fulfilled in our example when we restrict to the connected component of the identity.
\begin{proposition}\label{prop:3loop-homotopy}
  Let $\pi_5(G) = \Z$. Then the identity component of the group $\gext{\threeloop}$ is simply connected for the basic extension.
\end{proposition}
\begin{proof}
  This is a simple generalisation of the similar result for the $1$-loop group~\cite[Prop.~4.4.6, p.~50]{PS1986}. The fibration
  \[
    \exact{\threemaps}{\gext{\threeloop}}{\threeloop}
  \]
  induces the homotopy exact sequence
  \begin{align*}
    \dots &\to \pi_2(\threemaps) \to \pi_2(\gext{\threeloop}) \to \pi_2(\threeloop) \\
          &\to \pi_1(\threemaps) \to \pi_1(\gext{\threeloop}) \to 0 .
  \end{align*}
  Here $\pi_1(\threeloop) \isom \pi_4(G) = 0$ and, since (the identity component of) $\threeball$ was assumed to be simply connected earlier, $\pi_n(\threemaps) \isom \pi_n(S^1)$. Moreover, under the basic extension we can identify $\pi_2(\threeloop)$ with $\pi_1(S^1) \isom \Z$.
  %\hox{(Wess-Zumino-Witten?)}.
  Thus we have
  \[
    0 \to \pi_2(\gext{\threeloop}) \to \Z \overset{i}{\to} \Z \to \pi_1(\gext{\threeloop}) \to 0 ,
  \]
  where the map $i:\Z \to \Z$ is an isomorphism. Hence for the identity component of the extension $\pi_1(\gext{\threeloop}) = \pi_2(\gext{\threeloop}) = 0$.
\end{proof}
\begin{remark}
  The condition $\pi_5(G) = \Z$ is justified again for groups $G=SU(p)$ with $p>2$.
\end{remark}

Let us summarise the construction. Originally the aim was to build a $3$-loop crossed module on the Abelian extension
\[
	\exact{\map{\threeloop, S^1}}{\gext{\threeloop}}{\threeloop}
\]
by using the action of the group $\threeball$ on $\threeloop$. However, this action does not lift to the extension and there is no natural morphism $\gext{\threeloop} \to \threeball$ as required by the crossed module. We can then proceed by the following: first, extend the acting group by
\[
  \exact{\threemaps/S^1}{\gext{\threeball}}{\threeball} ,
\]
and impose the action of the extended group $\gext{\threeball}$ to another extension given by
\[
  \exact{S^1}{\gext{\threeloop}}{\gext{\threeloop}/S^1} .
\]
Thus we get an action groupoid $\actg{\gext{\threeloop}}{\gext{\threeball}}$ which fulfills the requirements of a smooth crossed module by the morphisms $\alpha: \gext{\threeball} \to \aut(\gext{\threeloop}) $ and $\delta: \gext{\threeloop} \to \gext{\threeball}$. This is then the $3$-loop group analogue to the $1$-loop group crossed module introduced in \cite{MRW2017}.

\begin{remark}
  There is a slightly different route to the action morphism $\alpha$
  which at the outset does not require simply-connectedness for the group $\gext{\threeloop}$. Instead we assume that the cohomology group $\coh^1(\threeloop,\threemaps)$ is trivial, which is again true for $G=SU(p)$ with $p > 2$.
  We will briefly outline this approach; see \cite{MN2019} and \cite{MW2016} for further details.

  First recall from Appendix~\ref{app:cohomology} that $A^N$ denotes the group of $N$-invariant elements of $A$, where $N$ is a normal subgroup of a Lie group $H$, and $A$ is a smooth $H$-module. Let $\coc_s^1(N,A)$ be the group of smooth $1$-cocycles on $N$ with values in $A$: under the condition of a trivial cohomology group $\coh^1(N,A)$, there is an isomorphism $\coc_s^1(N,A) \isom A/A^N$.

  We can then apply this to $A = \threemaps$, $N=\threeloop$ and $H=\threeball$, and construct the following commuting diagram:~\cite[Sec.~4]{MW2016}
  % diagrams begins
  \[
  \begin{tikzcd}
  	0 \arrow[r] & \threemaps \arrow[d] \arrow[r] & \gext{\threeloop} \arrow[d] \arrow[r] & \threeloop \arrow[d] \arrow[r] & 0 \\
  	0 \arrow[r] & \threemaps/\map{S^2G,S^1} \arrow[r] & \Gamma \arrow[r] & \threeball \arrow[r] & 0
  \end{tikzcd}
  \]
  % diagram ends
  in which the group $S^2G$ consists of all smooth maps from the $2$-sphere to the group $G$, as introduced earlier at the beginning of this section. We can then identify $\threemaps^{\threeloop}$ with $\map{S^2G,S^1}$, and the group $\Gamma$ is the Lie group extension
  \[
    \exact{\threemaps/\map{S^2G,S^1}}{\Gamma}{\threeball} .
  \]
  Furthermore, there is a smooth action
  \[
    \Gamma \to \aut(\gext{\threeloop}) .
  \]
  By the smooth homomorphism
  \[
    \threemaps/S^1 \to \threemaps/\map{S^2G,S^1}
  \]
  and the trivial action of the group $\map{S^2G,S^1}$ on $\gext{\threeloop}$, we can lift the action $\Gamma \to \aut(\gext{\threeloop})$ to a smooth homomorphism
  \[
    \alpha: \gext{\threeball} \to \aut(\gext{\threeloop}) .
  \]
\end{remark}

\begin{remark}
  One of the virtues of the $1$-loop group crossed module was that it retained the possibility to a physically meaningful $S^1$-action of the group $SO(2)$. In the $3$-loop case one would naturally want the corresponding rotations of the $3$-sphere, but this is not possible in the above construction. This is due to the fact that the definition of the loops in $B^3$ requires a contracted boundary, and hence a choice of a distinguished point which cannot be made equivariant under the $S^3$-action of the group $SO(4)$.
\end{remark}

\idxloppu{loop group!$3$-loop}
%

%% Gerbes
\section{Bundle gerbes}\label{sec:gerbes} %
\idxalku{bundle gerbe}

As principal bundles are central to the mathematical tapestry of gauge theory, it is natural to wish for higher versions thereof. Here we focus on one particular case, that of (bundle) gerbes, which can be thought of as a categorification of (Hermitian) line bundles with connection. A bundle gerbe also provides a natural geometric realisation for the third cohomology group which has a decisive role in understanding Hamiltonian anomalies -- we study this aspect more closely in Chapters~\ref{chap:ktheory}~and~\ref{chap:hanomaly}, and concentrate here on the definition and basic properties. We follow mostly the introductory works by Murray~\cite{MURRAY1996,MURRAY2010}.

Let us begin fresh with the definition. In the following we assume that the spaces we work with are smooth manifolds.
\begin{definition}\label{def:bundle_gerbe}
  Consider maps $\pi: Y \to X$ and $p: P \to Y^{[2]}$, where $Y^{[2]}$ is the fibre product of $Y$ with itself. The pair $(P,Y)$ is called a \term{bundle gerbe over $X$} if the following holds:
  \begin{enumerate}
  	\item $\pi : Y \to X$ is a surjective submersion.
  	\item $p: P \to Y^{[2]}$ is a $U(1)$-bundle.
  	\item Define smooth maps $\pi_i : Y^{[k]} \to Y^{[k-1]}$ by omitting the $i$th element from the fibre product. The multiplication
  		\[
  			\mu: \pi_3^*(P) \otimes \pi_1^*(P) \to \pi_2^*(P)
  		\]
  		defines a smooth isomorphism of $U(1)$-bundles over $Y^{[3]}$ that is associative. In other words, if we denote by $P_{(y_i,y_j)}$ the fibre of $P$ over $(y_i,y_j)$, the following diagram commutes for all $(y_1,y_2,y_3,y_4) \in Y^{[4]}$.
  		\[
  			\begin{tikzcd}
  				P_{(y_1,y_2)} \otimes P_{(y_2,y_3)} \otimes P_{(y_3,y_4)} \arrow{d} \arrow{r} &
  					P_{(y_1,y_3)} \otimes P_{(y_3,y_4)} \arrow{d} \\
  			    P_{(y_1,y_2)} \otimes P_{(y_2,y_4)} \arrow{r} &
  			    		P_{(y_1,y_4)}
  			\end{tikzcd}
  		\]
  \end{enumerate}
\end{definition}

\begin{remark}
  Originally in \cite{MURRAY1996} the projection $\pi: Y \to X$ was assumed to be a fibration. However, it is enough if the map $\pi$ has only local sections, so that for all points $x\in X$ we have an open neighbourhood $\U$ and a local section $\Gamma: \U \to Y$. Such maps are called \term{locally split}, and in the context of smooth spaces they are surjective submersions~\cite{MuS2000}.
  For the anomalous bundle gerbe introduced in Chapter~\ref{chap:hanomaly} the map $\pi$ is indeed not a fibration.
\end{remark}

A bundle gerbe is closely related to the more general notion of \term{gerbe}, which can be formally defined either as a stack of groupoids~\cite[Ch.~III Déf.~2.1.1, p.~129]{GIRAUD1971} (see for instance \cite{MOERDIJK2002} for an introduction in English) or as a Lie group -banded sheaf of groupoids~\cite[Def.~5.2.4, p.~196]{BRYLINSKI1993}, both armed with certain local properties.\footnote{Brylinski traces his definition back to Giraud's, 
though it takes a few moments of pondering to see how the two are linked.} While these terms are often mixed rather liberally in the literature, the definition of a gerbe is quite a bit more general -- in this light, bundle gerbes as defined above are smooth versions of the Abelian gerbe: the band here is the group of smooth maps $\map{X,U(1)}$~\cite{MURRAY2010}.
In any case, bundle gerbes can be thought of as fibrations with groupoids as fibres. As such, they can represent obstructions for the existence of global bundles with particular properties, the details depending on the setting. In Chapter~\ref{chap:hanomaly}, we deal with such a bundle gerbe obstruction arising in quantum theory of massless fermion fields. 

In order to open up the definition we need to consider the fibre products more carefully. A $k$-fold fibre product $Y^{[k]}$ is a submanifold of the Cartesian product manifold $Y^k$:
\[
  Y^{[k]} \defeq \{(y_1,\dots,y_k) :~ \pi(y_1) = \dots = \pi(y_k) \} \subset Y^k .
\]
The manifold structure is retained since $\pi$ is a submersion. One can then naturally consider differential $p$-forms $\Omega^p(Y^{[k]})$ on the fibre products; moreover, there is a boundary map
\[
  \cob: \Omega^p(Y^{[k-1]}) \to \Omega^p(Y^{[k]}) :~ \cob\omega \defeq \sum_{i=1}^{k} (-1)^{k-1} \pi_i^{*} (\omega) .
\]
From these maps one forms the \term{fundamental complex}
\[
\begin{tikzcd}
  0 \arrow[r] & \Omega^p(X) \arrow[r,"\pi^*"] & \Omega^p(Y) \arrow[r,"\cob"] & \Omega^p(Y^{[2]}) \arrow[r,"\cob"] & \dots
\end{tikzcd}
\]
which is exact for all forms of degree $p \geq 0$. Furthermore, if we have a map $g: Y^{[k-1]} \to A$ to an Abelian group $A$, we define $\cob g : Y^{[k]} \to A$ as an alternating sum
\[
  \cob g = \sum_{i=1}^k (-1)^{i-1}(g \circ \pi_i) .
\]
Finally, given an $A$-bundle $P \to Y^{[k-1]}$ there is another $A$-bundle $\cob P \to Y^{[k]}$ fulfilling
\[
  \cob(P) = \pi_1^*(P) \otimes (\pi_2^*(P))^* \otimes \pi_3^*(P) \otimes \cdots
\]
From these definitions it follows that $\cob^2 g = 1$ and that $\cob^2 P $ is trivial. Note also that the boundary map $\cob$ commutes with the de Rham derivative of forms.

The fibre product $Y^{[2]}$ in Definition~\ref{def:bundle_gerbe} has a natural groupoid structure, and the multiplication property of the bundle gerbe is constructed to be compatible with the in-built groupoid multiplication on the base $Y^{[2]}$.
By multiplication we can take a section $\sigma$ of $\cob P \to Y^{[3]}$, and furthermore a section $\cob\sigma$ of $\cob^2 P \to Y^{[4]}$; since the latter bundle is trivial by the definition above, and following the associativity of the multiplication, we settle for the natural condition $\cob \sigma = 1$.

\begin{definition} %
    \idx{bundle gerbe!trivial}
  A bundle gerbe $(P,Y)$ is \termd{trivial} if there is a $U(1)$-bundle $T \to Y$ such that $(P,Y) \isom (\delta(T),Y)$. If such an isomorphism can be found, we call it a \termd{trivialisation} of $(P,Y)$.
\end{definition}

Recall that for an Abelian $G$ any principal $G$-bundle $P$ has a dual bundle denoted by $P^*$, and that the tensor product of two principal $G$-bundles is a principal $G$-bundle. These concepts extend naturally to bundle gerbes.
\begin{definition}
  Given a bundle gerbe $(P,Y)$ there is a \termd{dual bundle gerbe} $(P,Y)^* \defeq (P^*,Y)$.
\end{definition}
\begin{definition}
  Two bundle gerbes $(P,Y)$ and $(Q,Z)$ form a \termd{product bundle gerbe} $(P,Y) \otimes (Q,Z) \defeq (P\otimes Q, Y \times_X Z)$ where
  \begin{enumerate}
    \item $\pi: Y \times_X Z \to X$ is a surjective submersion.
    \item $P \otimes Q  \to (Y \times_X Z)^{[2]}$ is a $U(1)$-bundle.
  \end{enumerate}
  Here, $(Y \times_X Z)^{[2]} = Y^{[2]} \times_X Z^{[2]}$ and fibrewise
  \[
    (P\otimes Q)_{((y_1,z_1),(y_2,z_2))} = P_{(y_1,y_2)} \otimes Q_{(z_1,z_2)} .
  \]
\end{definition}

From these definitions we conveniently arrive at a suitable concept of \emph{equivalence} for bundle gerbes. Specifically and in contrast to principal bundles, not all trivial bundle gerbes are isomorphic; hence the right equivalence is the stable isomorphism, reminiscent of the stable isomorphism of vector bundles in topological K-theory (see Section~\ref{sec:ktheory}).
\begin{definition} %
    \idx{bundle gerbe!equivalence of}
  Two bundle gerbes $(P,Y)$ and $(Q,Z)$ are \termd{stably equivalent} or simply \termd{equivalent} if the product $(P^*,Y) \otimes (Q,Z)$ is isomorphic to a trivial bundle gerbe.
\end{definition}

Stable equivalence is an equivalence relation of bundle gerbes~ \cite{MuS2000}. Given such an equivalence class one can define a local bundle gerbe in terms of a \emph{family} of $U(1)$-bundles over intersections $\U_i \cap \U_j$, where $\U_i, \U_j$ belong to a good cover of $X$. As such bundle gerbes can work as a categorification of transition functions of fibre bundles as follows. Recall that fibre bundles $Y\to X$ can be characterised up to an isomorphism by local transition functions; these are (complex) functions $Y^{[2]} \to \C$ which over local triple intersections satisfy the cocycle condition. If instead of complex numbers we characterise these cocycles by {$U(1)$-torsors} and replace the equations with coherent isomorphisms, we are back at the definition of the bundle gerbe.

Lastly, we note the following convenient way to illustrate a bundle gerbe $(P,Y)$ as a single diagram:
\[
  \gerbe{P}{Y}{X}
\]

\begin{remark}
  Bundle gerbes of Definition~\ref{def:bundle_gerbe} are sometimes called (principal) $U(1)$-bundle gerbes, or Abelian bundle gerbes. One can work out the theory of more general (principal) $G$-bundle gerbes as well, see \cite{ACJ2005} for the definition and applications in higher Yang-Mills theory.
\end{remark}

\subsection{Differential geometry}

A bundle gerbe is an intrinsically differential geometric object. It is therefore natural to define a connection and its curvature analogous to those of a fibre bundle -- furthermore, there are additional entities related to the bundle gerbes: the $3$-curvature and the curving.

Recall that any principal $U(1)$-bundle $P$ over $Y$ has a connection. This extends naturally to their fibre products, and thus we can formulate a connection on a bundle gerbe.
\begin{definition} %
    \idx{connection!bundle gerbe}
  Let $\conn$ be a connection on $P \to Y^{[2]}$. It extends to \termd{bundle gerbe connection} if it is compatible with the bundle gerbe multiplication. In other words, if for the section $\sigma$ of $\cob P \to Y^{[3]}$ it holds that
  \[
    \sigma^*(\cob\conn) = 0 .
  \]
  The \termd{curvature form} %
    \idx{curvature!$2$-form}
  $F_{A} \in \Omega^2(Y^{[2]})$ of a bundle gerbe connection is defined as $F_{A} \defeq \ext A$ for a $1$-form $A$ of the connection $\conn$. The curvature satisfies $\cob F_A = 0$.
\end{definition}
Bundle gerbe connections always exist: this follows from exactness of the fundamental complex and the fact that we can always define a connection on a principal $U(1)$-bundle.

Since $\cob F_A = 0$ and the fundamental complex is exact, it follows that there must be a $2$-form $f$ on $Y$ such that $\cob f = F_A$. Taking the de Rham derivative we get
\[
  \cob \ext f = \ext \cob f = \ext F_A = 0 ,
\]
and hence for some $3$-form $\omega$ on $X$ we can write $\ext f = \pull{\omega}$. Moreover, $\omega$ is closed, %\hox{in vertical directions?},
as $\pull{\ext \omega} = \ext \pull{\omega} = \ext^2 f = 0$.
\begin{definition}
  The $2$-form $f \in \Omega^2(Y)$ defined by $\cob f = F_A$ is called the \termd{curving} of the bundle gerbe connection $\conn$, and the $3$-form $\omega \in \Omega^3(X)$ defined by $\pull{\omega} = \ext f$ is called the \termd{$3$-curvature} of the bundle gerbe connection $\conn$.
\end{definition}

The $3$-curvature can be used to represent the third cohomology of the space $X$, and hence bundle gerbes can be used as a geometric realisation of these classes.

\subsection{Geometric realisation of cohomology classes}

Giraud's motivation in \cite{GIRAUD1971} for introducing gerbes was to describe non-Abelian cohomology. Specifically, a gerbe with a Lie group band $G$ over a manifold $X$ characterises cohomology classes in $\coh^2(X,\sheaf{G})$ (in sheaf cohomology, so cochains take values in the sheaf of maps to $G$ -- see Appendix~\ref{app:cohomology}). If the band is an Abelian group one gets an isomorphism between the equivalence classes of gerbes and the second (Abelian) cohomology~\cite[Thm.~5.2.8, p.~201]{BRYLINSKI1993}. By a \emph{transgression} this gives a handle on the third integral cohomology group $\coh^3(X,\Z)$,
hence bundle gerbes offer a geometric realisation of the third cohomology. %
  \idx{transgression}

An apt analogue can be given in regards to principal $U(1)$-bundles. A principal $U(1)$-bundle or a complex line bundle over a manifold $X$ can be characterised in terms of the second integral cohomology: it can be locally represented as a cocycle $c_{ij}: \U_i \cap \U_j \to U(1)$, where $\U_i, \U_j$ are elements of an open cover of $X$. The associated characteristic class of the bundle is the first Chern class in $\coh^2(X,\Z)$, and the map $c$ fulfills the familiar first order cocycle condition~\cite[Sec.~23.1]{HJJS2007}
\[
  c_{ij}c_{jk}c_{ki} = 1 \quad \textrm{ on } \U_i \cap \U_j \cap \U_k .
\]

Similarly, to a bundle gerbe we can locally attach a cocycle $c_{ijk}: \U_i \cap \U_j \cap \U_k \to U(1)$ (fulfilling the second-order cocycle condition). We then associate a class in $\coh^3(X,\Z)$ to the bundle gerbe; this is called the Dixmier-Douady class. 

\begin{definition} %
    \idx{Dixmier-Douady class}
  \termd{Dixmier-Douady class} as a characteristic class is a homotopy class of maps
  \[
    X \to K(\Z,3) ,
  \]
  where $X$ is a compact space. It is an element in $\coh^3(X,\Z)$ by the isomorphism
  \[
    \coh^3(X,\Z) \isom [X,K(\Z,3)] .
  \]
\end{definition}
This is a higher degree version of the first Chern class, which is a homotopy class of maps  $X \to K(\Z,2)$.~\cite[Sec.~10.1]{HJJS2007}

% \begin{remark}
%   The Dixmier-Douady invariant needs to be defined modulo the action of the homeomorphism group of $X$ on the group $\coh^3(X,\Z)$, see \cite{SCHOCHET2009}. The algebraic definition relates to projective bundles and $C^*$-algebras which we consider more closely in Chapter~\ref{chap:ktheory}.
% \end{remark}

The Dixmier-Douady class can then be represented as a closed $3$-form with integral periods. Since any closed $3$-form appears as the $3$-curvature of a bundle gerbe, we can form an equivalence between the Dixmier-Douady classes and bundle gerbes.
\begin{proposition} %
    \idx{bundle gerbe!equivalence of}
  Two bundle gerbes are equivalent (stably isomorphic) if and only if their Dixmier-Douady classes match. In particular, a bundle gerbe is trivial if and only if its Dixmier-Douady class vanishes. Hence there is an isomorphism
  \[
    \text{Stable isomorphism classes of bundle gerbes on $X$} \isom \coh^3(X,\Z) .
  \]
\end{proposition}
\begin{proof}
  Central component is the transgression
  \[
    \coh^2(X,\sheaf{U(1)}) \isom \coh^3(X,\Z) .
  \]
  For details, see \cite{MURRAY1996}.
\end{proof}

\begin{example}[Lifting bundle gerbe~{\cite[Sec.~6.1]{MURRAY2010}}] %
    \idx{bundle gerbe!lifting}
  The Hamiltonian anomaly manifests as a projective bundle of Fock spaces, and this is tied to the problem of lifting the group of gauge transformations over the moduli space to its extension by an Abelian group. Anticipating this, we can as an example consider the lifting bundle gerbe related to such prolongation problems of group actions. The gist of this construction is as follows. Consider a central extension of a Lie group $G$
  \[
    \exact{U(1)}{\gext{G}}{G} ,
  \]
  and a principal $G$-bundle $Y \to X$. Does this bundle lift to a principal $\gext{G}$-bundle $\gext{Y} \to X$?
  The answer lies in topology. Transition functions of the bundle $Y \to X$
  induce a $U(1)$-valued cocycle as above, and the lift is possible if and only if the corresponding cohomology class in $\coh^2(X,\sheaf{U(1)}) \isom \coh^3(X,\Z))$ is trivial.

  Associated to the bundle $Y \to X$ we can define a map $Y^{[2]} \to G$, and the $U(1)$-bundle $\gext{G} \to G$ can be pulled back to a $U(1)$-bundle $Q \to Y^{[2]}$. One can show that there is a sound definition for the fibre multiplication, and thus a well-defined bundle gerbe $(Q,Y)$ called the \term{lifting bundle gerbe}. Its Dixmier-Douady class is by definition the obstruction to the desired lift, and thus the lifting bundle gerbe is trivial if and only if the $G$-bundle $Y\to X$ lifts to a $\gext{G}$-bundle. This result can be extended to hold for more general Lie group extensions besides the central extension -- in particular, it applies to the Abelian extension of the Hamiltonian anomaly~\cite{HMSV2013}.
\end{example}

\begin{remark} %
    \idx{Kac-Moody!group}
  There is a natural connection to the affine Kac-Moody groups in terms of the canonical generator $\coh^3(G,\Z) \isom \Z$ which underlies the theory of central extensions of $\loopg{G}$ for a compact Lie group $G$~\cite[Ch.~4]{PS1986}. This can be characterised by bundle gerbes over $G$, the cohomology class of which by the transgression map gives a central extension of the loop group $\loopg{G}$. Thus a $U(1)$-bundle over the loop group can be realised as an element in $\coh^3(G,\Z)$ with the generator
  \[
    \omega = k \tr \left( \inv{g}\ext g \wedge \comm{\inv{g}\ext g}{\inv{g} \ext g} \right),
  \]
  where $k$ is the normalisation parameter and $\inv{g}\ext g$ is the Maurer-Cartan form on an adjoint representation of $G$. %
  This gives the representation of the Dixmier-Douady class of the related bundle gerbe. See \cite[Ch.~6]{BRYLINSKI1993} for details and more examples regarding applications to loop spaces.
\end{remark}

\begin{remark}
  In the spirit of continuing categorification one can also introduce higher versions of the (bundle) gerbes and thus construct realisations of higher cohomology classes. Note in particular the bundle $2$-gerbe with useful applications as outlined in \cite{STEVENSON2004}.%
    \footnote{See also \cite{WALDORF2007} and \cite{BUNK2017} for the $2$-category structures defined on bundle gerbes.}
  Similarly to the bundle gerbe, the $2$-gerbes provide a geometric handle to the cohomology classes of degree $4$. Whereas the cohomology classification of bundle gerbes has a direct link to the anomalous gauge algebra $2$-cocycles in quantum physics, bundle $2$-gerbes may illuminate the failing of the Jacobi identity described by gauge algebra $3$-cocycles~\cite{CMW1997,JACKIW1985}.
\end{remark}

\idxloppu{bundle gerbe}
%

% k-theory and dirac operators
%!TEX root = teesirunko-arxiv.tex
\chapter{K-theory and Dirac operators}\label{chap:ktheory}

In order to get a proper handle on the topological aspects of gauge symmetries we need to discuss the topology of Dirac operators. To this end we introduce the rudiments of K-theory as a means to deal with homotopic families of Dirac operators, and show how this connects to the cohomology classification of bundle gerbes.

K-theory and operator topology is a vast landscape in its own right with many applications. One of special importance in gauge theory is the celebrated Atiyah-Singer index theorem: that the index of a Fredholm operator $D$ defined as a map
\[
	\ind D = \dim( \ker D) - \dim( \cker D) \in \Z
\]
in fact gives a map to a K-theory group, and thus is a topological invariant which can be related to suitable characteristic classes. We will briefly discuss the index theorem as applied to Dirac operators. In the next chapter we will see how this relates to the gauge anomalies; namely, how the families index theorem yields the symmetry-breaking cocycle terms in the Hamiltonian anomaly. This chapter is meant largely to give the general background for this particular application.

The origin of K-theory in quantisation anomalies is quite natural. The Dirac %Hamiltonian
operators are usually by construction self-adjoint and have polarised essential spectrum. Homotopy classes of such operators can be used as a definition for the elements in the group $K^1(X)$, where $X$ is the parameter space of the operator family. The characterisation of these elements -- ignoring possible torsion -- can be done via cohomology, where we also find the Dixmier-Douady class of the associated bundle gerbe. This then is also the topological source for the anomalous terms in the current algebra commutators we will explore in Chapter~\ref{chap:hanomaly}.

%% operator theory
\section{Some fundamentals from operator theory}

Let us begin by recalling some of the basic definitions from operator theory. Let $\Hs$ be a complex Hilbert space and consider a linear operator $D:\Hs \to \Hs$. By the phrase \laina{operator $D$ on $\Hs$} we mean that the operator $D$ is defined on some (dense) subset of $\Hs$. For composite operators we expect that the intersection of their domains will also be dense.

\begin{definition}
	An operator $D$ on $\Hs$ is \termd{closed} if its graph is closed in $\Hs\oplus \Hs$. By \termd{graph} we mean the set
	\[
		\G_D = \{ (v,w) ~\mid~ Dv = w \} \subset \Hs\oplus \Hs .
	\]
	The \termd{closure} $\bar{D}$ of $D$ is its smallest closed extension: a closed operator $\bar{D}$ such that $\bar{D} = D$ on $\dom(D)$ and $\dom(D) \subset \dom(\bar{D})$.
\end{definition}

\begin{definition} %
		\idx{operator!bounded}
	An operator $D$ on $\Hs$ is \termd{bounded} with respect to a given norm in $\Hs$ if there is some constant $m\geq 0$ such that
	\[
  	\norm{Dv} \leq m\norm{v} \quad \forall v \in \Hs \, .
	\]
	The least upper bound of $\norm{Dv}$ is called the \termd{norm of the operator $D$}.
\end{definition}
The \term{set of bounded operators on $\Hs$} is denoted by $\bnd{\Hs}$. If $\Hs$ is a Banach space (complete in the metric induced by the norm), the algebra generated by $\bnd{\Hs}$ with respect to the operator composition is a \term{Banach algebra}. The Banach algebra $\bnd{\Hs}$ is \term{unital}\idx{algebra!unital Banach} if the identity operator $\idoperator$ on $\Hs$ has the norm of $1$.

\begin{definition}
	An operator $D$ on $\Hs$ is \termd{compact} if there exists a neighbourhood $\U$ of the origin $0 \in \Hs$ and a compact subset $\V \subset \Hs$ such that $D(\U) \subset \V$.
\end{definition}

\begin{definition} %
		\idx{operator!self-adjoint}
	An operator $D$ on $\Hs$ is \termd{self-adjoint} with respect to the inner product on $\Hs$ if $\dom(D) = \dom(D^*)$ and
	\[
		\ip{Du,v} = \ip{u,Dv} \quad \forall u,v \in \Hs \, .
	\]
	The operator is \termd{essentially self-adjoint} if its closure is self-adjoint.
\end{definition}

Let $D$ be a closed operator, and either bounded or unbounded but self-adjoint. The \term{resolvent set} $\rho(D) \subset \C$ constitutes of complex numbers $\alpha$ for which the operator $\alpha \idoperator - D$ is a bijection of the domain $\dom(D)$ and has a bounded inverse on $\Hs$. This inverse is then called the \term{resolvent} of the operator $D$. Note that the \term{spectrum} of the operator $D$ can now be defined as the complement of the resolvent set: $\spec(D) \defeq \C \setminus \rho(D)$. From the definitions it follows that the spectrum of a self-adjoint operator is always a real subset $\spec(D) \subset \R$.%
\footnote{Note that the spectrum in this definition is not necessarily exactly the set of eigenvalues of $D$, but can be larger.}

Is the resolvent itself a compact operator? For \term{positive} operators, that is if $\ip{Dv,v} \geq 0$ for all $v\in \Hs$, the compactness can be solved by studying the operator $\inv{(\idoperator + D)}$.\footnote{The compactness of the resolvent does not depend on the choice of the number $\alpha$, and so in the case of positive operators we can use $\alpha=-1$, see \cite[p.~273]{GVF2001}.} Failing the positivity condition, we can consider the \term{Laplacian} $D^2$ and give the following definition.
\begin{definition}
	An operator $D$ on $\Hs$ has a \termd{compact resolvent} if $\inv{(\idoperator + D^2)}$ is compact.
\end{definition}

\begin{remark}
	The idea of the resolvent is linked to the \term{spectral theorem}. In essence, this means that self-adjoint operators can be realised as multiplication operators through unitary equivalence; to discuss this theorem properly we would have to introduce spectral measures of the operator and so we will not follow this further (see \cite{RS1980} for an introduction to the topic). Nevertheless, as a result, the underlying Hilbert space admits a spectral decomposition.
		\idx{observable}
	The usefulness of this in physics comes from the idea of representing quantum mechanical observables as self-adjoint operators on a Hilbert space, the state space of the system. One can describe the physical dynamics with unitary operators depending continuously on some external parameter (often the time variable). The self-adjoint generator of this group of unitary operators is called the \term{Hamiltonian} of the system, and changes in the system state can be linked to the spectrum of this operator.
	However, it is not always clear if this setting can be built in a fashion that is mathematically sound, see for instance \cite[Sec.~VII.11]{RS1980} for a discussion on possible problems.
\end{remark}

\idx{operator!Fredholm}
There are two important classes of operators we are interested in: Fredholm operators and elliptic operators.
\begin{definition}
	A bounded operator $D$ on $\Hs$ is \termd{Fredholm} if its kernel and cokernel are finite-dimensional, and its range is closed.
\end{definition}
Fredholm operators have a well-defined \term{analytical index}: %
	\idx{index!of Fredholm operator}
	\idx{index!analytical}
	\index{index|seealso {index theorem}}
\[
	\ind D =  \dim( \ker D) - \dim( \cker D) .
\]
Note that these definitions extend easily to operators for which the domain and codomain are different Hilbert spaces.

Let $D: \sct{E} \to \sct{F}$ be a \term{differential operator} %
	\idx{operator!differential}
on vector bundles $E,F \to X$ over a manifold $X$. In local coordinates we can write
\[
	D = \sum_{|\alpha| \leq m} A_{\alpha}(x) D^{|\alpha|} ,
\]
where $\alpha$ is the multi-index $(\alpha_1,\dots,\alpha_n)$ with $|\alpha| = \sum_k \alpha_k$, and
\[
	D^{|\alpha|} = \partial^{|\alpha|}/\partial x_1^{\alpha_1} \cdots \partial x_n^{\alpha_n} .
\]
The coefficients $A_{\alpha}(x)$ are matrices of smooth functions on $X$ fulfilling $A_{\alpha} \neq 0$ for some $\alpha$ with $|\alpha| = m$. The integer $m$ is then called the \term{order} of the operator. In the following we will abbreviate the local expression by $D = A_{\alpha}(x) D^{\alpha}$.

Let $\xi$ be a cotangent vector at a point $x\in X$ such that locally $\xi = \xi_i \ext x^i$ for some scalar coefficients $\xi_i$. The \term{principal symbol} %
	\idx{symbol!differential operator}
of the operator $D$ is a map $\sigma_{\xi}(D):E_x \to F_x$ defined locally by
\[
	\sigma_{\xi}(D) = i^m \sum_{|\alpha| = m} A_{\alpha}(x) \xi^{\alpha} .
\]
The principal symbol is a section of the bundle $(\odot^mTX)\otimes\Hom(E,F)$, represented by the coefficients $\{i^m A_{\alpha}\}_{|\alpha|=m}$.
\begin{definition} %
		\idx{operator!elliptic}
	A differential operator $D$ is \termd{elliptic} if its principal symbol $\sigma_{\xi}(D)$ is an isomorphism for all non-zero cotangent vectors $\xi \in T^{*}X$ .
\end{definition}

Any elliptic operator can be extended to a Fredholm operator with a matching analytical index. Moreover, on compact spaces $X$ all elliptic operators are Fredholm.~\cite[Sec.~III.5]{LM1989} The Hilbert spaces for these Fredholm extensions are the Sobolev space completions of the sections $\sct{E}$ and $\sct{F}$. %
	\idx{operator!Fredholm extension of}

\subsection[{$C^*$-algebras}]{\mots{\boldsymbol{C^*}}-algebras}\label{sec:cstar}
% C^* : indexing start
\idxalku{C@$C^*$-algebra}

\idx{algebra!$^*$-}
We introduce one additional operator in order to define $^*$-algebras: Let $\A$ be an associative algebra over $\C$. The operator $(\cdot)^*$ on $\A$ is called an \term{involution} if it satisfies the following conditions for all $a,b \in \A$:
\begin{enumerate}
	\item $(ab)^* = b^*a^*$.
	\item $(a^*)^* = a$.
\end{enumerate}
If furthermore the conditions $(a+b)^* = a^* + b^*$ and $(\lambda a)^* = \bar{\lambda}a^*$ are satisfied for all $a,b \in \A$ and $\lambda \in \C$, we call the algebra $(\A,^*)$ a \term{$*$-algebra}.
\begin{definition}
	A Banach algebra $\A$ over $\C$ is a \termd{$\boldsymbol{C^*}$-algebra} if it is a $^*$-algebra $\A$ and for all $a\in \A$ the following condition holds:
	\[
		\norm{a^*a} = \norm{a}\norm{a^*} = \norm{a}^2 .
	\]
\end{definition}

\begin{example}[Algebra of continuous functions]
	Let $X$ be a compact Hausdorff topological space. The set of all continuous complex functions $C(X)$ on $X$ has a natural Abelian $C^*$-algebra structure. Its unit is the constant function $f_1: x \mapsto 1$, and the Banach-algebra structure is given by the sup-norm
	\[
		\norm{f} \defeq \sup_{x\in X} \abs{f(x)} , \quad f \in C(X) .
	\]
	Involution can be defined as the complex conjugate $f^*(x) \defeq \overline{f(x)}$.
\end{example}

The algebra of bounded operators $\bnd{\Hs}$ is a $C^*$-algebra, when the involution is defined to be the operator adjoint. Indeed, by \emph{Gelfand-Naimark theorem},%
	\idx{theorem!Gelfand-Naimark}
any given $C^*$-algebra is isometrically $^*$-isomorphic to a $C^*$-algebra $\bnd{\Hs}$ of some Hilbert space $\Hs$; that is, the norm and the involution are preserved under the isomorphism.

\begin{remark}
	The Gelfand-Naimark equivalence of $C^*$-algebras with operators on a Hilbert spaces demands that the operators are bounded. This is somewhat problematic in quantum physics, where the relevant operators are usually unbounded. Various refinements to remedy this shortcoming have been proposed; for a \emph{higher perspective}, see \cite{MITCHENER2018} in which the problem is addressed by introducing an algebroid structure.
\end{remark}

\subsection*{CAR-algebra}\idx{canonical anticommutation relations (CAR)!algebra}

Recall that in Chapter~\ref{chap:gauge} we sketched the basic idea for transforming classical fields into quantum operators. For quantum field theories on Minkowskian manifolds this leads to $C^*$-algebras generated by operators subject to the \emph{Haag-Kastler axioms}~\cite{HK1964}. The $C^*$-algebra structure is necessary in order to include proper causal relations into the algebra of observables. For fermionic fields this can be extended to include globally hyperbolic Lorentzian manifolds~\cite{DIMOCK1982} and, as the end result, one recovers a functorial isomorphism linking the space-time structures and the algebraic $C^*$-structures.

Let $(V,\ip{\cdot,\cdot})$ be a topological vector space endowed with a nondegenerate symmetric bilinear form with values in $\R$.
The elements of $V$ generate a (unital) $^*$-algebra by the introduction of a bilinear multiplication and a linear involution $^*: V \to V$.%
	\footnote{The identity map always yields a trivial involution; for the multiplication we demand \emph{associativity} in light of the preceding definition.}
Using the symmetric form we can then define the CAR-algebra subject to the following conditions:
\begin{enumerate}
	\item $\acomm{f}{g} = \ip{f,g} \idoperator $.
	\item $f^* = f$.
\end{enumerate}
A standard example is given by a suitable algebra of $n\times n$-matrices with (conjugate) transpose as the involution, such as the subspace of symmetric or Hermitian matrices. More generally one introduces a linear map $\phi:V\to \bnd{\Hs}$ to bounded self-adjoint operators on a Hilbert space $\Hs$, and demands the relation
\[
	\acomm{\phi(f)}{\phi(g)} = \ip{f,g} \idoperator , \quad \text{for } f,g \in V .
\]
The map $\phi$ is then called the \term{representation of the CAR over $V$ in $\Hs$}; sometimes this is also called the Clifford system or Clifford relation on $\Hs$, reflecting the history of these relations in Clifford algebras.

For operators on a Hilbert space $\Hs$ we can define the analogous \term{$\boldsymbol{C^*}$ CAR-algebra} as the completion of the above unital $^*$-algebra. This algebra is generated by the symbols $\{B(f),B^*(f) : f \in \Hs \}$ and fulfills
\begin{enumerate}
	\item $\acomm{B(f)}{B^*(g)} = \ip{f,g} \idoperator $.
	\item $\acomm{B(f)}{B(g)} = \acomm{B^*(f)}{B^*(g)} = 0$.
	\item $(B(f))^* = B^*(f)$.
\end{enumerate}
Recall from Chapter~\ref{chap:gauge} that these are the relations for the creation and annihilation operators on the particle states. A representation of the CAR-algebra can be then built on the fermionic Fock space; for details, see \cite{ARAKI1987}.

\begin{remark}
	For the bosonic theory we would consider canonical \emph{commutation} relations, for which the anticommutator in the above conditions is changed into a commutator. Mixtures of these two then lead to free-particle states of mixed particle theory (and by perturbation, to interacting systems).
\end{remark}

% C^* : indexing end
\idxloppu{C@$C^*$-algebra}

%% operator K-theory
\section{Rudiments of K-theory}\label{sec:ktheory}

In this section we will briefly introduce the K-theories of topological vector bundles and of operators on a Hilbert space. These two concepts are closely linked, and together they will give us useful tools for studying the topological properties of families of Dirac operators, such as the operator index theorem introduced in a later section. We rely mostly on the introductory materials provided by \cite{ATIYAH1967,HATCHER2017} (topological K-theory) and \cite{BLACKADAR1986,GVF2001} (operator K-theory).

\subsection{Topological K-theory} %
\idxalku{K-theory!topological}

Let $X$ be a compact topological manifold. We consider complex vector bundles for which globally the fibre dimension may vary, as long as it is locally constant. Now for any two topological vector bundles $E \to X$, $F \to X$ we can take the direct sum $E\oplus F \to X$, itself a topological vector bundle of (local) dimension $\dim(E) + \dim(F)$. Thus the sum respects isomorphism classes of vector bundles so that
\[
	[E] \oplus [F] = [E \oplus F] ,
\]
and these classes form a monoid $V(X)$ under the sum operation, with the $0$-dimensional trivial vector bundle as the unit called the \term{zero bundle}. We can extend this to a group structure as follows.

First we define vector bundles $E$ and $F$ to be \term{stably isomorphic} if there are trivial vector bundles such that
\[
	E \oplus I^k \isom F \oplus I^l .
\]
We denote the stable isomorphism by $E \sim_S F$. This is an equivalence relation in $V(X)$, and the quotient set $V(X)/{\sim_S}$ is called the \term{reduced $K$-group of $X$} denoted by $\redK(X)$: the group structure is natural, since for every vector bundle we can find a complement bundle with which the direct sum is stably isomorphic to the zero bundle.

We can further extend the reduced $K$-group to a \term{$K$-group of $X$} %
	\idx{K-group!topological}
denoted by $K(X)$ by defining elements as virtual differences
\[
	E - F \in K(X) ,
\]
with respect to the equivalence relation
\[
	E_1 - F_1 = E_2 - F_2 \quad \Leftrightarrow \quad E_1 \oplus F_2 \sim_S E_2 \oplus F_1 .
\]
The group addition in $K(X)$ is defined by the taking the direct sum of the minuends and subtrahends, and the unit element is the class given by $E - E$. The group $K(X)$ is obviously Abelian. We sometimes also use the term \term{K-theory} of $X$ for the group, if the context is clear.

This procedure of constructing groups from monoids is general and such groups go by the name of Grothendieck groups (of commutative monoids). In the case of vector bundles, the contravariant functor assigning Abelian groups to topological spaces is called the \term{K-theory functor} or the Grothendieck functor.%
	\idx{K-theory!functor}
	\idx{Groethendieck functor}

\begin{remark}
	The fibrewise tensor product of vector bundles introduces a multiplication on $K(X)$, and thus the topological K-theory is endowed with a natural \emph{ring structure} as well. In general, this property does not extend to noncommutative operator K-theory.
\end{remark}

\subsubsection[{Groups $K^n$}]{Groups \mots{\boldsymbol{K^n}}}\label{sec:kn-groups} %
\idx{K-group!$K^n$}

The elementary topological K-theory group introduced above is often denoted by $K^0(X)$ -- the degree may be omitted if the context is clear. The degree is in relation to general $K^n(X)$ groups defined as (reduced) suspensions of the base.%
	\footnote{Reduced suspension is analogous to suspension in \emph{pointed spaces}. For $X$ compact Hausdorff, both suspension and reduced suspension will yield the same K-theory, see~\cite[Ch.~I]{BLACKADAR1986}.}
We define
\[
	K^{-n}(X) \defeq K(\Sigma^n X) .
\]

\begin{remark}
	The use of \emph{negative} indexing is purely a convention in reference to the cohomology aspect of sequences of $K$-groups. The idea is that there is a coboundary-like operator mapping from $K^{n}$ to $K^{n+1}$, and that the $K$-groups represent a (periodic) cohomology theory in this sense; see Section~\ref{sec:bott} for an example of the Bott periodicity. In the complex case the period is $2$, so we usually use just $K^0$ and $K^1$, since all the other groups are isomorphic to these two.
\end{remark}

\idxloppu{K-theory!topological}

\subsection{Families of operators and K-theory}

Topological K-theory can be generalised to a homotopy classification of operators acting on vector bundles. Thus we have a K-theory functor not only on topological spaces, but also on $C^*$-algebras. 

To begin with, we look at the index map of Fredholm operators $\F(\Hs_1,\Hs_2)$ between Hilbert spaces $\Hs_1$ and $\Hs_2$: %
	\idx{index!of Fredholm operator}
\[
	\ind: \F \to \Z : D \mapsto \dim(\ker D) - \dim(\cker D) .
\]
The index map is continuous and thus locally constant with respect to the operator norm topology~\cite[Sec.~III.7]{AS1968, LM1989}. In fact, the index defines a bijection between integers and the connected components of the operator space. If we have Fredholm operators $\F(\Hs)$ on a Hilbert space $\Hs$, we can use the composition of operators to give the space of connected components $\pi_0 \F$ a semigroup structure. The index map then becomes a group isomorphism
\[
	\ind : \pi_0 \F(\Hs) \to \Z .
\]
In general, the index of an elliptic operator is defined as the index of its Fredholm extension~\cite[Cor.~III.5.3]{LM1989}.

\idx{operator!elliptic}
Let $E$ and $F$ be vector bundles over $X$ and consider a continuous family of elliptic operators $D_t : \sct{E} \to \sct{F}$, where the parametrisation is over the unit interval, $t \in I$; a family $D_t$ is \term{continuous} if the operators
\[
	D_t = A^{\alpha}(x,t)D^{\alpha}
\]
are continuous both in $x \in X$ and $t \in I$. This implies that the order of $P_t$ is constant and that the Fredholm extension can be made continuous.%
	\footnote{We can use the norm topology for bounded operators. In the unbounded case one needs to apply additional conditions, for instance by introducing the Riesz map $D \mapsto D/(\norm{D}+1)$ to bounded operators. See \cite[Sec.~2.6]{BB2013} for a more detailed account.}
Hence $\ind P_t$ is constant in $I$ as the index map of Fredholm operators.

\idxalku{operator!homotopy class of}
Now, two elliptic operators $P_1,P_2$ over $(X;E,F)$ are \term{homotopic} if they can be joined by a continuous family $P_t$ of elliptic operators. Thus we get that the index of the elliptic operator on a compact manifold depends only on its homotopy class, which in turn is isomorphic to the symbol class. Moreover, since the principal symbol of the operator defines an element of the K-theory $K(TX)$ (where $TX$ is the tangent bundle) and the index depends only on the symbol class of the operator, the index map induces a homomorphism
\[
	K(TX) \to \Z .
\]
The index theorem states that the analytical and topological indices coincide as such homomorphisms.~\cite[Sect. III.13]{LM1989}

\idxalku{operator family}
The homotopy parametrisation of operators gives us a simple first instance of families of operators: that of a product family, where for the continuous family $P_t$ of operators the parameter $t$ lies not in the unit interval $I$ but in a more general Hausdorff space $A$. If $X$ is a compact manifold and constant in $A$, we consider the product space $X \times A$.
This way we get an {index bundle} for Fredholm operators.

More generally, letting the parameter space $A$ be homotopically nontrivial, we need to consider twisting of the manifold and the operators over $A$. Instead of a product space $X \times A$, we have a bundle of manifolds $X \to A$ over the parameter space. The analytical index is then a map
\[
	K(TX) \to K(A) ,
\]
where $TX$ is the tangent bundle along the fibres, and the symbol class of the operator is an element of the K-theory $K(TX)$.~\cite{AS1971}.

Let $A \to \F(\Hs_1,\Hs_2)$ be a continuous family of Fredholm operators $T:\Hs_1 \to \Hs_2$, where $\Hs_i$ are Hilbert spaces and the parameter space $A$ is a compact topological space. We assume that the Hilbert spaces are fixed over $A$.%
	\footnote{Cf. Remark 3.31 in \cite[p. 85]{BB2013}.}
If $A$ is connected, we have
\[
	\ind T_a = \ind T_{a'}	\quad \textrm{for all } a,a' \in A ,
\]
and as above, each Fredholm operator is assigned an integer and for the connected parameter space there is an isomorphism
\[
	\ind: [A,\F] \to \Z
\]
on the homotopy set $[A,\F]$~\cite[Sec.~3.7]{BB2013}.

For a continuous family with \emph{constant kernel dimension}, there are well-defined vector bundles $\ker T$ and $\cker T$ over $A$, and we can consider the index via topological K-theory.
For each continuous family $A \to \F$ of Fredholm operators in a Hilbert space $\Hs$, we can then assign an \term{index bundle} %
	\idx{index!of operator family}
\[
	\ind T \coloneqq [\ker T ] - [\cker T] \in K(A) .
\]

Note that this definition relies on the fact that the dimensions of the nullspaces do not vary. However, we emphasise the following: it can be shown that the K-theoretic classification has sufficient stability \emph{even in the general case where the dimensions of the operator nullspaces may vary with a topologically nontrivial parameter space $A$}~\cites[Sec.~III.8]{LM1989}[p. 84]{BB2013}. In this case the K-theory does not map to integers, but the definition is nevertheless formally meaningful.

In the case of $\Hs_1 = \Hs_2$, the homotopy invariance generalises to the statement that the space of Fredholm operators $\F(\Hs)$ actually constitute a classifying space for the K-theory functor: there is a natural isomorphism 
\[
	\ind : [A,\F] \to K(A)
\]
such that for any continuous $f:A' \to A$, where $A'$ and $A$ are compact Hausdorff,
\[
	\ind \circ \, \pull{f} = \pull{f} \circ \ind ,
\]
where $\pi:\F \to A$ is the projection.

\idxloppu{operator family}

\subsubsection[{Group $K^1$}]{Group \mots{\boldsymbol{K^1}}} %
\idx{K-theory!operator}
\idx{K-group!$K^1$}

The index map of homotopies is a degree $0$ K-theory:
\[
	[A,\F] \isom K^0(A) .
\]
This isomorphism gives a complete classification of the homotopy type of the space of Fredholm operators on a Hilbert space $\Hs$. In anticipation of families of Dirac operators (Section~\ref{sec:dirac_families}), we present the refinement into the group $K^1(A)$ as follows:
\begin{theorem}\label{theorem:k1-group}
	Let $\Hs$ be a separable Hilbert space, and let $\F^*$ be the space of self-adjoint Fredholm operators on $\Hs$ with polarised essential spectrum. Then for any compact Hausdorff space $A$, the space $\F^*$ is a classifying space for the $K^1$-functor, that is
	\[
	  K^1(A) \isom [A,\F^*] .
	\]
\end{theorem}
\begin{proof}
For the original result, see \cite{AS1969}.
\end{proof}

The importance of this classification is that it enables us to work with different families of Dirac operators up to homotopy. Later in Section \ref{sec:families_cohomology} we connect the homotopy classes to the third cohomology of the parameter space, which will be useful in the analysis of the Hamiltonian anomaly.

\idxloppu{operator!homotopy class of}

\subsection{K-theory of operator algebras} %
\idx{K-theory!of operator algebra}
\idx{C@$C^*$-algebra}

Let $C(X)$ be the $C^*$-algebra of a topological space $X$ defined by continuous complex functions. By the Gelfand-Naimark theorem %
\idx{theorem!Gelfand-Naimark}
there is an identification
\[
	\begin{tikzcd}
		\mlnode{unital Abelian $C^*$-algebra $\A \defeq C(X)$} \arrow{r} & \arrow{l} \mlnode{compact Hausdorff space $X$} ,
	\end{tikzcd}
\]
where from the algebra-perspective the space $X$ is the space of characters of $\A$, and from the topology-perspective the algebra consists of evaluation maps at points of $X$.

Consider a vector bundle $E\to X$, and define another vector bundle $F\to X$ with respect to $E$ by setting $E\oplus F \isom X \times \C^n$; such a bundle always exists~\cite[Prop.~2.2,~p.~52]{GVF2001}. Then besides the algebra $\A = C(X)$ we can also define modules $M = C(X,E)$ and $N = C(X,F)$ so that $M \oplus N \isom \A^n$. In other words, the vector bundle $E\to X$ defines a canonical finitely generated projective module $M$ over $\A$. Conversely, any such module over $\A$ is the space of all continuous sections of some vector bundle over $X$.
All in all, we have the \term{Serre-Swan theorem}: the category of vector bundles over the topological space $X$ is equivalent to the category of finitely generated projective modules over the algebra $\A$.~\cite[Ch.~2]{GVF2001}

Therefore we can define $K$-groups of $C^*$-algebras similarly to the topological K-theory. The set of isomorphism classes of finitely generated projective modules over $\A$ is an Abelian semigroup with respect to the direct sum. By the Grothendieck construction %
	\idx{Groethendieck functor}
we get a group $K_0(\A)$ by defining the equivalence relation
\[
	(M,N) \sim (M',N') \Leftrightarrow M \oplus N' \oplus B \isom M' \oplus N \oplus B ,
\]
for some projective module $B$ over $\A$. Then~\cite[Sec.~3.2]{GVF2001}
\[
	K_0(C(X)) \isom K^0(X) .
\]

% \begin{remark}
% 	For a nonunital $C^*$-algebra $\A$ the above construction can be applied when extending to $\A^+ = \A \oplus \C$. One can show that for a unital $\A$ there is an isomorphism $K_1(\A) \isom K_1(\A^+)$; although this does not extend to the $K$-groups of degree $0$.\cite[???]{BLACKADAR1986}
% \end{remark}

\subsection{Bott periodicity}\label{sec:bott} %
\idxalku{Bott periodicity}

Lastly, we make note of the (complex) K-theoretic formulation of the \term{Bott periodicity}. Its application to group cohomology is often useful in gauge theory; the \emph{real} K-theory version on the other hand has applications in condensed matter physics, where it is used to classify topological phases of matter for free fermionic systems~\cite{THIANG2016}.

Let $X$ be a locally compact Hausdorff space, $Y$ a closed subspace of $X$, and define $Z=X\setminus Y$ as the complement. Let $q: X^+ \to Z^+$ be the map between one-point compactification of the spaces $X$ and $Z$ such that $q(x) = x$ for $x\in Z^+$ and $q: x \mapsto \{+\}$ for $x\in X^+\setminus Z$. This induces an isomorphism between the relative group $K(X,Y)$ and the group $K(Z)$. Hence the following sequence is exact:
\[
	K(Z) \overset{q^*}{\longrightarrow} K(X) \overset{\iota^*}{\longrightarrow} K(Y) ,
\]
where $q^*$ and $\iota^*$ are pullbacks to the commutative monoid of vector bundles over the spaces. Here we have identified the groups $K(X)$ and $K(X^+)$ with each other, as usual.~\cite[Ch.~I]{BLACKADAR1986}

\begin{theorem}[Bott periodicity]
	The diagram of $K$-groups
	% diagrams begins
	\[
	\begin{tikzcd}
		K^0(Z) \arrow[r,"q^*"] & K^0(X) \arrow[r,"\iota^*"] & K^0(Y) \arrow[d,"\bdr"] \\
		\arrow[u,"\bdr\,"] K^{-1}(Y) & \arrow[l,"\iota^*"] K^{-1}(X) & \arrow[l,"q^*"] K^{-1}(Z)
	\end{tikzcd}
	\]
	% diagram ends
	is a cyclic exact sequence. For real K-theory there is a similar cycle with the period of $8$ instead of $2$.
\end{theorem}
\begin{proof}
See \cite[Ch.~IV]{BLACKADAR1986} or \cite[Sec.~I.9]{LM1989}.
\end{proof}

\begin{remark}\label{rem:bottcohomology}
	The Bott periodicity is quite a general phenomenon. It applies to the groups of unitaries $U(p)$ and $SU(p)$ as follows:~\cite{BOTT1957}
	\[
		\pi_k(U(p)) \isom  \pi_k(SU(p)) = \begin{cases} 0, & k \text{ even;}\\ \Z, & k \text{ odd.} \end{cases}
	\]
	assuming $k>1$ and $p \geq (k+1)/2$. 
	From this one also gains generators of the integral cohomology groups $\coh^{2m-1}(SU(p))$ in the odd degrees for $m\leq p$; by the Bott periodicity they are one-dimensional and can thus be generated by a single element. This will be useful when computing the anomalous Dixmier-Douady class in Section~\ref{sec:general_comm}.
\end{remark}

\idxloppu{Bott periodicity}
%

%% index theorem basics
\section{Atiyah-Singer index theorem} %
\idxalku{index theorem}
\idxsee{theorem!Atiyah-Singer}{index theorem}

The Atiyah-Singer index theorem states that the analytical index of an elliptic operator $P$ on a compact Riemannian manifold $X$ %
	\idx{index!analytical}
\[
	\ind_a P = \ker P - \cker P
\]
coincides with the \emph{topological index} %
	\idx{index!topological}
$\ind_t P$ as a map
\[
	K(TX) \to \Z .
\]
The gist of the theorem is in connecting the functional analysis of an elliptic operator and the homotopy theory related to the topological properties of the operator symbol. It is then clear that the operator index is a \emph{topological invariant} and, by cohomology formulas for the topological index, it can be expressed with characteristic classes for which there are well-known computation methods.

\subsection{The topological index} %
\idxalku{index!topological}

Our aim is to sketch the construction of the topological index as a mapping
\[
	K(TX) \to \Z .
\]
As usual, on a Riemannian manifold we can canonically identify the tangent and cotangent bundles, and in the following we do not bother with being too pedantic about the distinction.

A crucial component in defining the right mapping for the topological index is the \term{Thom isomorphism}. Let $\pi : E \to X$ be a vector bundle of rank $n$ over a compact manifold $X$. For any integer $k \geq 0$, there is an isomorphism
\[
	\psi: \coh^k(X,\Z_2) \to \tilde{\coh}^{k+n}(T(E),\Z_2) ,
\]
with $\tilde{\coh}$ referring to reduced cohomology and $T(E)$ is the \emph{Thom space} of the bundle. The K-theoretic variant of this map can be given as follows.~\cite[App. C]{LM1989}
\begin{proposition}[Thom isomorphism in K-theory]
	The mapping
	\[
		\psi: K(X) \to K(E) : x \mapsto (\pull{x}) \lambda_E
	\]
	is an isomorphism.
\end{proposition}
Here $\lambda_E$ is the canonical \emph{difference element} in the K-theory $K(E)$; see \cite[Sec.~I.9]{LM1989} or \cite[Sec.~12.2]{BB2013} for details.
% The relevant Thom space arises from a sphere bundle constructed from the cotangent spaces. Assume $X$ is Riemannian, so that we can consider unit cotangent vectors: these form a sphere bundle over $X$, denoted by $SX \to X$.%
% \footnote{More generally, without metric structure, we can construct such spheres via one-point compactifications.}
% Related to the sphere bundle is the unit ball bundle $BX \to X$. The Thom space $T(E)$ is then defined as the quotient $BX/SX$.
%
Consider then the \term{principal symbol}
\[
	\sigma(D): T^*X \to \Hom(p^* E, p^* F)
\]
of an elliptic operator $D : \sct{E} \to \sct{F}$, where $p:T^*X \to X$. For an elliptic operator, the symbol defines an isomorphism away from the zero section. Hence we have a \term{symbol class}%~\cite[p.~169]{LM1989}
\[
	[\sigma(P)] = [p^*E,p^*F;\sigma(P)] \in K(T^*X) .%\isom K(E).
\]
% for $K(BX,SX)$ the relative $K$-group $\tilde{K}(BX/SX)$; here $\tilde{K}(X)$ is the kernel of the projection $K(X) \to K(x) \isom \Z$ for $X$ with a distinguished base point $x$. Since $(V^n)^+ \isom S^n$ for any real vector space, $BX/SX$ is homeomorphic to the one point compactification of $T^*X$, and we can identify $K(BX,SX)$ with $K(T^*X)$:~\cite[p.~60--66]{LM1989}
% \[
% 	K(BX,SX) = \tilde{K}(BX/SX) \isom \tilde{K}((T^*X)^+) = K(T^*X) .
% \]

Take a smooth embedding $f: X \embedd Y$, with $Y = \R^k$,\footnote{By the Whitney embedding theorem, this is always possible.} and push forward the tangent bundle $f_*: TX \embedd TY$. Denote by $N$ the normal bundle of $X$ in $Y$. We also have the normal bundle $TN$ of $TX$ in $TY$, which we gain by pulling $N\oplus N$ back to $TX$. Thus we have a fibre bundle $TN \to TX$ with pointwise components in manifold and fibre directions in $Y$. Now the Thom isomorphism gives us a map
\[
	\psi_{TN \to TX} : K(TX) \to K(TN) .
\]

Consider now the extension homomorphism $h:K(TN) \to K(TY)$ given by the natural embedding $TN \embedd TY$. Composing this with the Thom isomorphism we obtain a homomorphism
\[
	f_{!} : K(TX) \to K(TY) .
\]

We defined $Y = \R^k$, so $TY = \R^{2k}$. Let $i$ denote the inclusion of the origin in $Y$, and we get the induced map
\[
	i_{!}: K(T\{0\}) \to K(\R^{2k}) ,
\]
which in fact is a Thom isomorphism. Now, obviously, $K(T\{0\}) \isom \Z$, so that
\[
	\inv{i_{!}} : K(TY) \to \Z .
\]
Let us wrap this all in the final definition:
\begin{definition}[Topological index] The \termd{topological index of an elliptic operator $D$} is the map
\[
	\ind_t(D) : K(TX) \to \Z : D \mapsto \inv{i_{!}} f_{!} ~(\sigma(D)) .
\]
\end{definition}

The index is independent of the choice of the embedding $f$: First take a linear inclusion $j: \R^k \embedd \R^{k+n}$ for some $n \in \N$, and consider the composition $\tilde{f} = j \circ f$. The induced $j_! : K(T\R^k) \to K(T\R^{k+n})$ is a Thom isomorphism, and thus $ \inv{i_{!}} f_{!} = \inv{\tilde{i}_{!}} \tilde{f}_{!}$. Let then
$g$ be another embedding $X \embedd \R^{n}$, and consider the inclusion $j_g:\R^{n} \embedd \R^{n+k}$. The composite embeddings $j \circ f: X \embedd \R^{n+k}$ and $j_g \circ g : X \embedd \R^{n+k}$ are isotopic, and in fact we have a smooth family of such embeddings. Since K-theory is homotopy invariant, it follows that the topological index is independent of the choice of embedding.

% As with the analytical index, the topological index and the resulting index theorem can be modified to equivariant, families and Real variants, each giving a similar mapping with respect to the relevant K-theory. These generalisations are not always straightforward, however, and we will not try to sketch them here.

%
\idxloppu{index!topological}

\subsection{The index theorem and cohomology formulas}

We can now state the index theorem, although we will not attempt to give a proof here.
\begin{theorem}[Atiyah-Singer]
	Let $D$ be an elliptic operator on a compact Riemannian manifold $M$. Then
	\[
		\ind_a D = \ind_t D .
	\]
\end{theorem}
\begin{proof}
	Good expositions can be found in \cite[Ch.~III]{LM1989} and \cite[Pt.~III]{BB2013}. Consult \cite{AS1968} for the original proof relying on \emph{equivariant} K-theory.
\end{proof}

The K-theoretical index map can be realised also as a map between cohomologies, giving the common formulation for the index. The topological index can be written as
\[
	\ind_t D = \int_M \chern(E) \wedge \ahatgenus(M) ,
\]
where $\chern(E)$ is the Chern character of the vector bundle $E\to M$, and $\ahatgenus(M)$ is the $\ahatgenus$-genus of the base manifold. These invariants can be formulated in local terms: on one hand we have the Chern forms $\tr F^k$, where $F$ is the curvature $2$-form of \emph{some} connection on the bundle $E\to M$; and on the other hand we can write the $\ahatgenus$-genus as in terms of the curvature form $R$ on the tangent space of $M$.

\begin{remark}
	Not all topological information is preserved by the cohomology formula; K-theory does detect torsion, but this is lost when mapped to Chern characteristic class.~\cite[Sec.~III.17]{LM1989}
\end{remark}

\idxloppu{index theorem}
%

%% Dirac families
\section{Families of Dirac operators}\label{sec:dirac_families} %
\idxalku{operator family}
\idxalku{Dirac operator}

Let $X$ be a closed smooth manifold, and consider a family $\{D_x\}$ of self-adjoint unbounded Fredholm operators with compact resolvent. Such a family defines an element in $K^1(X)$, as explained in Section~\ref{sec:ktheory}.% -- in fact, this can be taken as the definition of the group $K^1(X)$.

Denote by $\F_R^*$ the \term{space of regular self-adjoint unbounded Fredholm operators}. This means that if $F \in \F_R^*$, then
\begin{enumerate}
	\item $F$ is closed and self-adjoint,
	\item the resolvent $\inv{(\idoperator + F^2)}$ is compact, and
	\item $F$ has both positive and negative spectrum, and both are infinite.
\end{enumerate}

A \term{family of such operators} $D_X \defeq \{D_x\}$ parametrised by a closed smooth manifold $X$ is a mapping $X \to \F_R^*$. Given a family $\{D_x\} \subset X \times \F_R^*$, the mapping $\xi(D_x) = D_x(\idoperator+D_x^2)^{-1/2}$ makes each operator $D_x$ bounded and we can define the bounded family
\[
  D_X^b \defeq \{\xi(D_x)\} .
\]
Applying the mapping $\xi$ does not change the homotopy of the original family, and thus to each family $D_X$ we can pair an element in $K^1(X)$ through the homotopy class of its bounded equivalent.~\cite{DK2010}

Observe that the common definition for the Dirac operator on a spin manifold fulfills these requirements. Let $M$ be a compact Riemannian spin manifold with an associated spinor bundle $S$ and let $E \to M$ be a complex vector bundle with Hermitian metric and a unitary connection $\conn^E$. 
\begin{definition}\label{def:dirac_op}
	The \termd{Dirac operator} $D$ is the composition of maps
	\[
		\sct{S\otimes E} \overset{\conn^E}{\longrightarrow} \sct{T^*M \otimes S\otimes E} \overset{\gamma}{\longrightarrow} \sct{S\otimes E} ,
	\]
	so that locally we have the formula
	\[
		D = -i \gamma(\ext x^k) \conn^E_k ,
	\]
	where $\gamma$ is the Clifford multiplication in the fibre, and $k$ is the local coordinate index. %
\end{definition}

\begin{remark}
	The factor $-i$ in the local formulation is due to the requirement for self-adjointness; see Proposition~\ref{prop:dirac_properties} below.
\end{remark}

\begin{remark}
	On even-dimensional manifolds the spinor bundle has a $\Z_2$-grading and so one gains Dirac operators $D^{\pm}$ mapping \emph{positive} spinors to \emph{negative}, and vice versa. This gives the often useful \term{Weyl spinor polarisation}, and the index of the Dirac operator can be defined as the difference of the nullspaces:
	\[
		\ind D = \dim(\ker D^+) - \dim(\ker D^-) .
	\]
\end{remark}

\begin{proposition}\label{prop:dirac_properties}
	The Dirac operator is regular, essentially self-adjoint, and Fredholm.
\end{proposition}
\begin{proof}
	The Dirac operator is \emph{elliptic}, and hence Fredholm on compact manifolds $M$.
	For regularity and self-adjointness, see for instance \cite[Ch.~9]{GVF2001}.
\end{proof}
In the following, we will call an operator \term{Dirac-like} if it fulfills these three requirements.

\begin{remark} %
	\idx{Dirac operator!pseudo-Riemannian}
	\idx{space-time!pseudo-Riemannian}
	The operator in Definition~\ref{def:dirac_op} could be called \emph{mathematician's Dirac operator}. It differs somewhat from what Paul Dirac originally considered in~\cite{DIRAC1928}, most importantly in the requirement of a Riemannian metric. One can define Dirac operators in a pseudo-Riemannian metric too, only then the operator will not be elliptic but hyperbolic. Proper take on the pseudo-Riemannian Dirac operators leads one to consider essentially self-adjoint operators on \emph{Krein spaces} instead of Hilbert spaces~\cite{STROHMAIER2006}. Some properties of the elliptic case are nevertheless retained; for instance, there is a sound definition for the Fredholm index and a related index theorem on globally hyperbolic Lorentzian manifolds with boundary~\cite{BS2017}. It is still an open question if this pseudo-Riemannian index theorem extends to other variants such as the families index theorem.
\end{remark}

\idxloppu{Dirac operator}

\subsection{Families index theorem} %
\idxalku{index theorem!of operator family}

The Atiyah-Singer index theorem has a natural extension to families of Fredholm operators. Recall from Section~\ref{sec:ktheory} that a family of Fredholm operators parametrised by a Hausdorff space $A$ induces an isomorphism
\[
	[A,\F] \to K(A) ,
\]
which we call the \term{index of the family of operators}. %
	\idx{index!of operator family}
Given a family $D_A$ of operators $D_a$, with $a\in A$, we can define the analytical index as a formal difference
\[
	\ind D_A = \{ \ker D_a \} - \{ \cker D_a \} ,
\]
where on the right-hand side we have families of vector spaces. If the dimension of the null-space $\ker D_a$ is \emph{constant} over $A$, these spaces are true vector bundles over $A$, and the definition of their difference as an element in $K(A)$ is natural.

More generally, we can have elliptic operators over a family of compact manifolds $\M \defeq \{M_a\}$, parametrised by a compact space $A$ in such a way that we have a fibre bundle $\M \to A$ with $M$ as the typical fibre and the structure group is the diffeomorphisms of $M$.

The analytical index of the family is then a map
\[
	K(T\M) \to K(A) .
\]
The families index theorem then states the following:~\cite{AS1971}
\begin{theorem}[Families index theorem]
	Let $D_A$ be a family of elliptic operators on a compact Riemannian manifold $M$ parametrised by a compact Hausdorff space $A$. Then
	\[
		\ind_a D_A = \ind_t D_A .
	\]
\end{theorem}
The cohomology formula of the index map can be written similarly to the normal index. For our purposes it is not necessary for the manifold $M$ to vary with the parameter space, and then the index formula is simply
\[
	\ind D_A = \int_M \chern(E)\wedge \ahatgenus(M) ,
\]
where $E$ is now a vector bundle over the \emph{product space} $M\times A$.

\begin{remark}\label{rem:aps_index} %
		\idx{index theorem!Atiyah-Patodi-Singer}
		\idx{$\eta$-invariant}
	In Chapter~\ref{chap:hanomaly} we discuss the index of a family over a manifold with boundary. In this case there is a boundary term given by the spectral $\eta$-invariant. Now boundary conditions apply: in the Atiyah-Patodi-Singer index theorem~\cite{APS1975}, one requires that the boundary values should lie in the subspace spanned by certain eigenfunctions under an orthogonal projection. %..\hox{check all this}
	The $\eta$-terms in the index formula are not local, since they reflect the spectral properties of the operators on the boundary. However, when applying the index theorem to the Hamiltonian anomaly, we see that these boundary parts do not actually contribute due to the gauge invariance.
\end{remark}

\idxloppu{index theorem!of operator family}
\idxloppu{operator family}
%

%% operator cohomology and gerbes
\section{Cohomology and gerbes}\label{sec:families_cohomology}

Let $D_X$ be a family of Dirac-like operators, parametrised over a closed smooth manifold $X$. We can define the \term{spectral graph} of the family as the closed subset
\[
	\mathcal{S}(D_X) = \{ (x,\lambda) ~\mid~ \lambda \in \spec(D_x) \} \subset X \times \R .
\]
If the spectral graph is disconnected, the family represents a trivial element of the K-theory $K^1(X)$. There are then topological obstructions to deforming a given family to a one which has a disconnected spectral graph, and these obstructions can be represented as the components of the related Chern character.~\cite{DK2010}

\idx{Dixmier-Douady class}
The first component of the Chern character can be realised as a $1$-cocycle and it defines the spectral flow of the family. The second component manifests as a $2$-cocycle and we have the following theorem connecting the K-theoretic classification of families and the Dixmier-Douady class of the bundle gerbe~\cite[Thm.~2.3]{DK2016}:
\begin{theorem}\label{theorem:family_kclass}
	Let $X$ be a closed smooth manifold and let $\{D_x\}_{x\in X}$ be a family of self-adjoint unbounded Fredholm operators with compact resolvent. Then the cohomology class $[D_X] \in \coh^3(X,\Z)$ equals the Dixmier-Douady class of the bundle gerbe of the family.
\end{theorem}

The cohomology classes are constructed as follows.
Let $\{\U_i\}$ be an open cover 
of the space $X$, and assume that its element sets and their finite intersections are contractible. Given an open set $\U_i$ in the cover, there is a \emph{non-eigenvalue} $\lambda_i \in \R$ such that $\lambda_i \notin \spec(D_x)$, for any $x\in \U_i$.
Such non-eigenvalues indeed exist for any given $x\in X$ since the Dirac operator coupled to a potential $x\in X$ always has a spectral gap, and in particular on compact manifolds the spectrum is discrete~\cite[Thm.~III.5.8, p.~196]{LM1989}. This extends over the set $\U_i$ by Propositions~2.4--2.5 in \cite{DK2010}.
Recall from Chapter~\ref{chap:gauge} that then the eigenspace $\Hs$ of $D_x$ can be decomposed into a subspace $\Hs^+$ for eigenvalues greater than $\lambda_i$
and its orthogonal complement $\Hs^-$: this is the spectral decomposition $\Hs = \Hs^+ \oplus \Hs^-$, locally defining the Dirac sea of \laina{positive} and \laina{negative} particles with respect to the spectrum of the Dirac operator.

\idx{canonical anticommutation relations (CAR)}
Let then $\pi^+(x)$ be the local projection onto the space $\Hs^+$. It induces a quasi-free representation of the CAR-algebra $\Ca(\Hs)$ by the map
\[
	\alpha_i(x) : \Ca(\Hs) \to \bnd{\F_{i}(\Hs)} .
\]
Here $\F_{i}(\Hs)$ is the fermionic Fock space over a point $x\in X$ defined by %
	\idx{fermionic Fock space}
\[
	\F_{i}(\Hs) = \F(\pi^+(x)\Hs)\otimes \F(\overline{(\idoperator - \pi^+(x))\Hs}) ,
\]
recall Chapter~\ref{chap:gauge} and see \cite{ARAKI1987} for the general idea.
By \term{quasi-free} %
	\idx{quasi-free representation}
we mean that in each of the Fock spaces there is a vacuum vector $v_0$ satisfying %
	\idx{vacuum vector}
\[
	\annih(v) v_0 = \creat(u) v_0 = 0
\]
for all $v\in \Hs^+$ and $u\in \Hs^-$. In general, the choice of a vacuum vector cannot always be made continuously, and the topology of the vacuum line bundle is physically significant. We will discuss this more closely in Chapter~\ref{chap:hanomaly}.

Now for $x$ in the intersection $\U_i \cap \U_j$, we can define two different but \emph{equivalent} representations $\alpha_i(x)$ and $\alpha_j(x)$. Between these there is an intertwining operator %
\[
	S_{j,i}(x): \F_j(\Hs) \to \F_{i}(\Hs)
\]
such that
\[
	S_{j,i}(x) \alpha_i(x) S_{j,i}^*(x) = \alpha_j(x) .
\]
Let us then consider triple intersections: let $x\in\U_i\cap\U_j\cap\U_k$. We can define a circle-valued function $g(\U_i,\U_j,\U_k) : \U_i\cap\U_j\cap\U_k \to S^1$ induced by the composition
\[
	S_{j,i}(x) S_{i,k}(x) S_{k,j}(x) = g(\U_i,\U_j,\U_k)(x) .
\]
By Schur's lemma, this defines a scalar in $S^1$ since the composition is an intertwining operator of irreducible representations of the algebra $\Ca(\Hs)$.

A crucial property of the map $g$ is that it is a continuous function over $X$, since the intertwining operators can be made continuous; see \cite{DK2016} for details. Therefore the \emph{families of maps} $\{g:\U_i\cap\U_j\cap\U_k \to S^1\}$ define a $2$-cocycle over the cover $\{\U_i\}$ with values in the sheaf of $S^1$-valued functions, and thus we have a  cohomology class
\[
	[g] \in \coh^2(\{\U_i\},\sheaf{S^1}) .
\]

Furthermore, with a suitable choice of the cover,
this class maps to the second \v{C}ech cohomology by the canonical isomorphism (see Appendix~\ref{app:cohomology}).
For smooth and paracompact $X$ the groups in \v{C}ech cohomology and de Rham cohomology are equivalent~\cite[Sec.~1.4]{BRYLINSKI1993}, and we have an isomorphism with the third integral de Rham cohomology
\[
	\coh^2(X,\sheaf{S^1})  \isom \coh^3(X,\Z)
\]
induced by the exponential exact sequence~\cites[Ch.~23]{HJJS2007}[Sec.~5.2]{BRYLINSKI1993}
\[
	\exact{\Z}{\C}{S^1} .
	% \begin{tikzcd}
	% 	\Z \arrow{r} & \R \arrow{r} & S^1 .
	% \end{tikzcd}
\]

The cohomology class defined by the map $g$ depends on the family $D_X$ only up to homotopy: hence we can link the homotopy class of the operator family to the third de Rham cohomology. That is, by construction, we have the following useful result:
\begin{lemma}
	The cohomology class $[D_X] \in \coh^3(X,\Z)$ depends on the family up to homotopy.
\end{lemma}
This forms the basis for computing the topological invariants related to the anomalous current commutators, as one can apply a homotopy transformation to the moduli space and yet derive the Schwinger terms from the same de Rham representations of the cohomology class. We will come back to an example of this in Section~\ref{sec:general_comm}.

\idx{Dixmier-Douady class}
Furthermore, these classes are the Dixmier-Douady classes defining equivalent bundle gerbes related to the operator family.
The gerbe class used in \cite{DK2016} is that of an \emph{index gerbe} defined in \cite{LOTT2002}. This is a variant of the bundle gerbe we discussed in \ref{sec:gerbes}; see \cite{CW2006} for further details.

\begin{remark}
	An interesting question is when the element of $K^1(X)$ given by the family is trivial -- in \cite{DK2010,DK2016} it is shown that the obstruction to trivialisation can be expressed as relevant Chern classes, the first two relating to the \emph{spectral flow} and the \emph{index gerbe}. Vanishing Dixmier-Douady class leads to a trivial element under the assumption that the operator family has a spectral multiplicity bounded by $2$.
	For families with \emph{constant} spectral multiplicity, it is enough to trivialise the spectral flow~\cite{DK2010}.
	These results do not in general show up in the physics applications though, since the restriction to the spectral multiplicity is often too severe.
\end{remark}

% hamiltonian anomaly
%!TEX root = teesirunko-arxiv.tex
\chapter{Hamiltonian anomaly}\label{chap:hanomaly}
\idxalku{anomaly!Hamiltonian}

In this chapter we consider chiral fermions coupled to an external gauge field and the breaking of their internal symmetry structure. 
The coupling comes from a parametrisation by external fields, and we see that the gauge invariance cannot always be maintained when moving in this parameter space.

From the Lagrangian perspective the anomalous symmetry results in an effective action functional which is not gauge invariant -- geometrically this is seen from the curvature of a complex line bundle called the determinant bundle, defined over the moduli space of gauge connections. We stress that this situation arises when considering Weyl fermions -- that is, \emph{massless} Dirac fields. For the general Dirac fields with a nonzero mass, the effective action functional has no such deficiency.

When we look at this from the Hamiltonian perspective, we see that the equal time commutator relations of the gauge current algebra are modified by anomalous terms called the Schwinger terms. Given a trivial vector bundle modelling the external field coupling, this current algebra is formed from the Lie algebra of functions $\map{M,\lie{g}}$, where the manifold $M$ is the physical space and the Lie algebra $\lie{g}$ is that of the gauge symmetry group $G$. The anomalous terms are derived from the cohomology of the related current group extension, and they are related to a topological obstruction to prolong the action of the gauge symmetry group on the projective Fock bundle to an action on a proper Hilbert bundle. In physics terms, this means that we cannot fix a coherent vacuum state.

The origin of the anomaly can be said to be topological. We strive to connect all the different perspectives in this light: that the cohomology classes given by the current algebra cocycles are computed as the curvature of the determinant bundle, which is a topological invariant,
and that these same classes are coming down from the anomalous bundle gerbe. Thus the defining topological invariant behind the anomaly is the Dixmier-Douady class which classifies the bundle gerbes. For computations one can apply the index theorem for families of Dirac operators over the moduli space of gauge connections.

The novel contribution in this Chapter is the generalisation of the current commutators and the computation of the Schwinger terms in Section~\ref{sec:general_comm}. Traditionally the Hamiltonian anomaly is concerned with the time-components of the gauge currents and their commutators. For the remaining current components one needs to consider a modified space of connections living on the tensor product of the spinor bundle and the gauge bundle. We show that under a fixed spin connection this modification is in fact topologically equivalent to the standard case, and that therefore we can apply similar methods to compute the generalised commutators.
We conclude with example computations in the case of $M=S^3$ illustrating the effect of these modifications.

\paragraph{Background geometry}

\idx{space-time}
Let us first recall the background geometry for the most of what follows in this chapter. Let $M$ be a compact and connected smooth spin manifold without boundary. Fix the dimension $n \defeq \dim(M) = 2k + 1$ for some $k\in\N_0$. Write $S$ for the spinor bundle over $M$, and let $E$ be a trivial vector bundle over $M$, so that we can define square-integrable sections $\psi \in \Hs$ of the vector bundle $S\otimes E \to M$.%
	\footnote{The measure in $\Hs$ is defined by the fixed metrics in $M$ and $E$.}
\idx{spinor!bundle}
\idx{fermion!field}

Let $G$ be a compact Lie group of internal gauge symmetries with a unitary representation in the fibre of the bundle $E\to M$, following the outline in Chapter~\ref{chap:gauge}. Associated to the symmetry group we locally define the gauge connections $\A$ as $1$-forms on $E$ with values in the Lie algebra $\lie{g}$ of $G$.
We will call the bundle $E\to M$ the \term{gauge bundle}.

\idx{gauge transformation!group of}
Furthermore, since $E$ is trivial, we can define $\G$ as the space of gauge transformations as smooth \emph{based} maps in $\map{M,G}$ -- for all $g \in \G$ and a fixed point $p \in M$, set $g(p) = e$, the identity in $G$.
\idx{connection!moduli space of}
The quotient $\A/\G$ is the moduli space of gauge connections. By the quotient manifold theorem the moduli space is a smooth manifold if the action of $\G$ on $\A$ is free, which is guaranteed by fixing the base point as above.

We call the collection $\left(M,S\otimes E,G \right)$ the \term{odd-dimensional fermionic geometry}. The \term{Dirac operator $D_A$} is then a first-order differential operator on $\Hs$, the sections of the bundle.

%% Physics
\section{Physics behind the anomaly}

The central statement of the Hamiltonian anomaly is the impossibility to define a coherent vacuum state on the bundle of Fock spaces. While the axial chiral anomaly we discussed in Chapter~\ref{chap:gauge} has laboratory-observable physical consequences~\cite{XIONG2015} -- only \emph{external} symmetry is broken and thus the anomaly term is not at odds with the theory itself -- the non-Abelian anomaly in general signifies the breaking of quantum theory since the gauge symmetry is not conserved.
In this case one loses the gauge equivalence between the quantised state spaces.

Let us quickly review the Dirac quantisation. Consider a family of (classical) Dirac operators $D_A$ defined with respect to the space of gauge connections $\A$. These operators transform covariantly under the adjoint action of the gauge transformations $g\in \G$:
\[
	\ad_{g}(D_A) = D_{A^g} ,
\]
where $A \to A^g$ is the gauge transformation of the connection. A self-adjoint operator $D_A$ defines a polarisation of the particle state space $\Hs$ into negative and positive energy subspaces with respect to some chosen non-eigenvalues of $D_A$. This, in turn, gives an irreducible representation of the CAR-algebra in the Fock space
\[
	\F(\Hs) = \F({\Hs^+}\oplus\Hs^-) = \F(\Hs^+) \otimes \F(\overline{\Hs^-}),
\]
see \cite{ARAKI1987} for a more in-depth review. For this representation to be uniquely defined (up to an equivalence) and physically meaningful, one requires the existence of a \term{vacuum state vector} $\psi_0$ such that
\[
	\creat(v)\psi_0 = \annih(w)\psi_0 = 0
\]
for all pairs $(v,w) \in \Hs^- \times \Hs^+$. Here, $\creat$ and $\annih$ are the creation and annihilation operators in $\Hs$, as defined earlier in Chapter~\ref{chap:gauge}.

The Fock bundle is defined over the parameter space, $\F \to \A$. The quantised Dirac operator acts in the fibres of this bundle, but it is not obvious if the action of the symmetry group on the base can be lifted to the fibre; in other words, if there is a symmetry covariance for the quantised Dirac operator $\hat{D}_A$ such that the vacuum state is preserved. For massless Dirac operators and smooth vector potentials, this may not be possible. The bulk of this chapter is dedicated to studying the means of assessing when such a situation arises, and to the different topological and geometric perspectives behind it.

\begin{remark}
There is a mass-gap in the general Dirac equation which ensures the existence of a vacuum state. For the massive fermion fields there is then a continuous interval of non-eigenvalues of the Dirac operator, and we can define the spectral polarisation of the Hilbert space continuously over the parameter space.
\end{remark}

\begin{remark}
Massless fermion fields do not sound very physical at the outset. However, Weyl fermions serve as important building blocks in the standard model and in general the Dirac spinors can be defined as combinations of massless left- and right-handed fermions. Moreover, while they have not been observed as elementary particles as such, Weyl fermions may emerge as quasi-particles in condensed matter systems. In particular, there is an important class of topological materials called \emph{Weyl semimetals} which retain many properties of the massless system.
\end{remark}

\subsection{Current commutators} %
\idxalku{anomaly!commutator}

Consider the mapping group $\map{M,\lie{g}}$ and the related gauge currents $j_{\mu}^a(x)$, where $\mu$ is the space-time index and $a$ the internal gauge index which runs through the generators of the gauge algebra; see Chapter~\ref{chap:gauge} for details.
	\idxalku{current algebra}
Anomalous symmetries were first observed as a divergence term coming from equal time commutator relations~\cite{GI1955}. The commutators for the \term{time components} of the current terms of the form
\[
	j_{0}^a(x) = \psi^{\dagger}(x)\tau^a\psi(x) ,
\]
where $\psi(x)$ is a fermion field, follow from the known relation
\[
	\comm{j_0^a(x)}{j_0^b(y)} = i\lambda_c^{ab}j_0^c(x)\delta^{(2n+1)}(x-y) + \alpha^{ab}(x,y) .
\]
Here $\delta^{(2n+1)}$ is the $(2n+1)$-dimensional Dirac delta distribution, and the coefficients $\lambda_c^{ab}$ are the structure constants for the generators $\tau^a$ of some given matrix representation of the Lie algebra $\lie{g}$ of $G$ so that
\[
	\comm{\tau^a}{\tau^b} = i\lambda_c^{ab} \tau^c .
\]

What we are interested in here are the \emph{anomalous terms} $\alpha^{ab}(x,y)$ that appear as an extra on the right-hand side.  %
	\marginnote{Julian Schwinger (1918--1994) introduced an additional term to the commutator equation as a remedy to the observed inconsistency~\cite{SCHWINGER1959}.} %
	\idx{Schwinger terms}
These \term{Schwinger terms} reflect the divergent gauge currents and hence introduce a symmetry anomaly.
Formally they are local expressions involving the Dirac delta distribution and its derivatives, and in general they can be computed as cocycles in group cohomology~\cite{JACKIW1972,JACKIW1985}. Depending on the model, these terms may be trivial -- and when they are not, they signify the presence of an anomaly. In three dimensions the Schwinger terms are dictated by the \term{Mickelsson-Faddeev cocycle} %
	\idx{Mickelsson-Faddeev cocycle}
we already touched upon in Chapter~\ref{chap:higher}~\cite{MICKELSSON1985,FADDEEV1984}:
\[
	\alpha(A;u,v) = \frac{i}{24\pi^2} \int_{M^3} \tr A ~[\ext u , \ext v ] ,
\]
so that for instance, following the notation in \cite[Ch.~10]{FADDEEV1984,AI1995}, the cocycle amends the commutator relation as follows:
\begin{align}\label{eq:faddeev_schwinger}
	\comm{j_0^a(x)}{j_0^b(y)} =& i\lambda_c^{ab}j_0^c(x)\delta^{(3)}(x-y) \nonumber \\
		&~ + \frac{i}{24\pi^2} \tr (\acomm{\tau^a}{\tau^b}\tau^c) \epsilon_{\mu\nu\eta} \partial_{\mu} A_{{\nu}}^c(x) \partial_{\eta} \delta^{(3)} (x-y) .
\end{align}

While these computations with the time-components are well-known, it is not quite as obvious as to what happens with the commutators for the space components. One can again write down the commutators formally as
\[
	\comm{j_0^a(x)}{j_{\mu}^b(y)} = i\lambda_c^{ab}j_{\mu}^c(x)\delta^{(2n+1)}(x-y) + \text{cocycle terms} ,
\]
but more care has to be heeded when describing the exact symmetry relations.
In general, the currents feature linear combinations of the basis $\sigma_{\mu} \otimes \tau^a$ built from Pauli matrices and the generators of the current algebra -- we will call the components of this tensor product the \term{spin components} and the \term{Yang-Mills components}, respectively. 
Now the general current components can be written as tensor products of the form
\[
	j_{\mu}^a(x) = \psi^{\dagger}(x)\left(\sigma_{\mu}\otimes \tau^a \right) \psi(x) ,
\]
and one can compute the naive commutator relations using the basic properties of matrix tensor products. To obtain the anomalous cocycle terms, the same basis extension needs to be done for the differential forms $A$, $u$ and $v$.

Geometrically, the gauge connections live now on the tensor product bundle $S\otimes E$, and the group of gauge transformations is not induced just from the internal symmetry group $G$ but from the \term{extended symmetry group} $G_S$ which accounts for both $\spin$ and $G$. However, since the \emph{moduli space} remains unaffected by this modification, one nevertheless expects similar cocycle terms to appear as with the time-components -- and this can be justified by topological arguments, as we do in Section~\ref{sec:general_comm}.

\begin{remark}
How does the commutator anomaly break the gauge symmetry? One can see this by considering the so-called \term{Gauss law constraints} in gauge theory which follow directly from the symmetry invariance~\cite{FADDEEV1984}. These constraints as quantised operators must follow the commutator relations so that the gauge algebra can be recovered: in non-anomalous theories the Gauss law operator generates the gauge transformations~\cite[Sec.~3.2]{JACKIW1985}. Anomalous terms in the commutators prevent this, and thus the constraints cannot be placed on the physical states; see \cite[Ch.~10]{AI1995} for details. 
The other side of the story lies in the topological origin of the anomaly terms, which we discuss more closely in the following.
\end{remark}

\idxloppu{current algebra}
\idxloppu{anomaly!commutator}
%

%% Group extensions
\section{Extensions of current groups} %
\idxalku{current group}

In quantum mechanics, the state space is modelled as rays on a Hilbert space $\Hs$, thus defining a projective space $P(\Hs)$. 
On the other hand, the projective unitary group $PU(\Hs) = U(\Hs)/S^1$ is a connected component of the unity of all unitary operators.
Given a connected topological group of symmetries $G$ of the system, there is then a continuous representation in the group $PU(\Hs)$ (or in the space $P(\Hs)$).

A question then arises whether it is possible to lift this to a unitary representation on $U(\Hs)$. There are two obstructions creating this \term{lifting problem}: the topological obstruction when $G$ is not simply connected, and the \emph{cohomological obstruction}. If $G$ is a Lie group, the latter can be computed in Lie algebra cohomology, where the relevant class is given by $\coh^2(\lie{g},\R)$.

A similar issue is behind the Hamiltonian anomaly in quantum field theory when considering massless fermions coupled to an external potential. Here one has a \emph{projective} bundle $\mathcal{P}$ of Fock spaces on the space of connections $\A$. Given a gauge symmetry group $G$ and the induced group of based gauge transformations \,$\G$, we would want to have a Hilbert bundle $\Hs$ over $\A/\G$ which has a projective bundle isomorphic to $\mathcal{P}/\G$. For $\A$, an affine space and simply connected, the question is trivial. But taking the connections modulo gauge transformations, we need to lift the group action of $\G$ to $\Hs$ and to this there is again a cohomology obstruction, this time as a class in $\coh^3(\A/\G, \Z)$.

The core issue in both of these questions is that the lift of the group action does not necessarily work with the desired group $G$, but \emph{only with a particular extension of the group.} Such extensions can be characterised by the cohomology theory of Lie groups and algebras, and the related cocycles then become the source for the physical anomaly terms. In particular, if the group extension of $G$ is by an Abelian group $A$, it is classified (up to an isomorphism) by the cohomology group $\coh^2(G,A)$. There is also an important connection to the third cohomology group through transgression, leading to a geometric characterisation through bundle gerbes.

\subsection{Cohomology theory} %
\idxalku{Lie group!cohomology of}

We recall some of the basic results concerning group cohomology from Appendix~\ref{app:cohomology}. Let $A$ be an Abelian group and $G$ a group acting via automorphisms $\rho: G \to \aut(A)$ on $A$ -- in other words, $A$ is a $G$-module.

By Theorem~\ref{theorem:group_ext_coh} there is an isomorphism
\[
  \coh^2(G,A) \isom \extensions(G,A) .
\]
Note that the cocycles in Theorem~\ref{theorem:group_ext_coh} are not necessarily smooth (besides locally in the neighbourhood of the unity). To properly apply the cohomology theory to (possibly infinite-dimensional) Lie groups, we also require additional conditions for \emph{smoothness}. This means that the action map on $G \times A$ should be smooth, and that the cochain maps $G^n \to A$ should be smooth in an identity neighbourhood. For Theorem~\ref{theorem:group_ext_coh} to hold one needs to assume connectedness of the group $G$, work within the connected component, or consider a suitable subgroup cohomology~\cite[Sec.~II and App.~B]{NEEB2004}. In the following, we tacitly assume such conditions whenever necessary.

\idx{Lie algebra!cohomology of}
There is a corresponding cohomology theory for Lie algebras~\cite[Sec.~7.5--7.6]{HN2012}.\footnote{Here instead of smoothness we should require \emph{continuity} for the action map $\lie{g}^n \to \lie{a}$~\cite[Sec.~I and App.~B]{NEEB2004}.} An important result analogous to Theorem~\ref{theorem:group_ext_coh} is the following isomorphism:
\[
  \coh^2(\lie{g},\lie{a}) \isom \extensions(\lie{g},\lie{a}) ,
\]
for extensions of $\lie{g}$ by an Abelian Lie algebra $\lie{a}$ and with respect to a given representation of these Lie algebras on some vector space. The exact relationship between the second cohomology groups of Lie groups and algebras depends on the setting. In particular, we have Theorem~\ref{theorem:group_algebra_coh_isom}:
\begin{quotation}
	\noindent Let $G$ and $A$ be connected Lie groups, $G$ simply connected and $A$ Abelian. Then
	\[
	  \coh^2(G,A) \isom \coh^2(\lie{g},\lie{a}) .
	\]
\end{quotation}

In general, one can define a derivation map for $\coh^n(G,A) \to \coh^n(\lie{g},\lie{a})$ for $n\geq 2$, which is a monomorphism if $G$ simply connected and $A \isom \lie{a}/\Gamma_A$ for some discrete subgroup $\Gamma_A$ of $\lie{a}$, see Theorem~\ref{theorem:cohom_hom}. Many of the applications in gauge theory can be made to fulfill these requirements, in particular this holds for all current groups $\map{M,G}$ with $G=SU(p)$ for $p\geq3$
if the base manifold $M$ has the dimension $3$. The case of $SU(2)$ is different, since the current group is not connected. However, the extension is in fact a simpler one with the fibre $\Z_2$ and the computation of the relevant cohomology groups is not that tricky, see~\cite{MICKELSSON1987} for details.

\idxloppu{Lie group!cohomology of}

\subsection{Current groups and anomalous cocycles}

The group-theoretic setting for the fermionic field theory necessarily deals with infinite-dimensional Lie groups and their algebras. The gauge current algebra is an example of such an object: a Lie algebra formed of maps $\map{M,\lie{g}}$. We look here at examples for dimensions $\dim(M)=1$ and $\dim(M)=3$: in the first case we see the chiral anomaly arising from a central extension of the Kac-Moody algebra, and the second example introduces the Mickelsson-Faddeev extension we already saw in Chapter~\ref{chap:higher} in the context of $3$-loop groups. We refer to \cite[Ch.~4]{MICKELSSON1989} for more details.

To begin with, we have a projective action on the Fock space through the group extension
\[
	\exact{\map{\A,S^1}}{\gext{\G}}{\G} .
\]
In the case that the extension is cohomologically trivial, the action can be lifted to an honest group action. The corresponding class in the Lie algebra is then the source for the possible anomaly terms: from these cocycles one can derive the local Schwinger terms for the current commutators.

Let us look at examples of this in the the case of $M=S^n$ for odd-dimensional $n$-spheres. Now for the group of gauge transformations $\G = \map{S_e^n,G}$ -- the group of based maps from $S^n$ to $G$ -- the fibration
\[
  \G \to \A \to \A/\G
\]
is equivalent in homotopy to the fibration
\[
  \G \to \paths{S_e^{n-1}G} \to S_e^{n-1}G ,
\]
where the elements in $\A \defeq \paths{S_e^{n-1}G}$ are smooth paths $f$ on $S_e^{n-1}G$ starting from the identity $e$, and if we parametrise the paths by  $t \in [0,1]$, we have $\pi(f) = f(1)$ under the projection $\pi: \A \to \A/\G$. To ensure that the moduli space is a smooth manifold we also need to apply suitable boundary conditions, such as imposing zero radial derivatives on the boundary.

This construction works for $n\geq 2$, and in the case of $n=1$ we have $\G = \bloopg{G}$ and the sequence
\[
  \bloopg{G} \to \paths{G} \to G .
\]
Depending on the topology of the symmetry group $G$, this may lead to nontrivial extensions by the Abelian group $\map{\A,S^1}$.

\idx{Kac-Moody!cocycle}
\begin{example}[$M=S^1$ and the central extension]
In the dimension $n=1$ the relevant group extension is \emph{central}:
\[
  \exact{S^1}{\gext{LG}}{LG}
\]
defining a principal $S^1$-bundle over $LG$. Note that this actually defines a \emph{family} of extensions, which can be characterised by their level $k$ in some representation, see \cite[Ch.~4]{PS1986}. In this case the Lie algebra $2$-cocycle is the central Kac-Moody term:
\[
  \omega(x,y) = \frac{ik}{4\pi} \int_{S^1} \ip{x, \ext y} .
\]
\end{example}

\idx{Mickelsson-Faddeev cocycle}
\begin{example}[$M=S^3$ and Mickelsson-Faddeev extension]
Let the space-time manifold be a $3$-sphere $S^3$. Now the moduli space of based gauge transformations $\A/\G$ is the current group $S_e^2G$~\cite{SINGER1981}.

The obstruction to lift the action is then characterised by the Abelian extension
\[
  \exact{\map{\A,S^1}}{\gext{S^3G}}{S^3G} .
\]
The anomalous commutator term is given by the $2$-cocycle of the corresponding Lie algebra extension
\[
  \exact{\map{\A,i\R}}{\gext{S^3\lie{g}}}{S^3\lie{g}} ,
\]
which is known as the \term{Mickelsson-Faddeev cocycle}~\cite{FADDEEV1984,MICKELSSON1985}:
\[
  \omega(A;x,y) = \frac{i}{24\pi^2} \int_{S^3} \tr A \comm{\ext x}{\ext y} .
\]

\end{example}

% The exact relation of the gauge group extensions to the Schwinger terms can be stated as follows:
% \begin{theorem}
% Let $G$ be a compact Lie group, and let $\lie{h} \defeq \Lie(\G)$ be the Lie algebra of the group of based gauge transformations $\G$. The Schwinger terms of the related current algebra commutators are cohomologous to the $2$-cocyles defining the Abelian extension
% \[
%   \exact{\map{\A,i\R}}{\gext{\lie{h}}}{\lie{h}} .
% \]
% \end{theorem}
% \begin{proof}
% Compare \cite[Sec.~10.7]{AI1995}
% \end{proof}

\subsection{Transgression}\label{sec:transgression} %
\idx{transgression}

Let $\omega$ be a $3$-form on the moduli space $\A/\G$ representing a cohomology class, and let $\pi: \A \to \A/\G$ be the natural projection. Since $\A$ is topologically trivial, there is an \emph{exact} $3$-form on $\A$ given by the pull-back $\ext \mu = \pi^*(\omega)$, where $\mu$ is some $2$-form. Explicitly, for $A\in \A$ and $x,y \in \Lie(\G)$, we can write the Lie algebra $2$-cocycle $c$ as
\[
  c(A;x,y) = \mu_A(x,y) .
\]
Here $x,y$ are interpreted as vertical tangent vectors at $A \in \A$.

This relationship between the topology of the moduli space $\A/\G$ and the cocycles of the extensions of the Lie algebra $\Lie(\G)$ is given by the following map~\cite{CM1994}:
\begin{definition}
	The \termd{cohomology transgression} is the map
	\[
		\coh^p(\A/\G,\Z) \to \coh^{p-1}(\Lie(\G),\map{\A,i\R})
	\]
	for $p>1$.
\end{definition}

The third cohomology group on $\A/\G$ is important since its elements represent the obstruction to the lift of the principal $\G$-bundle to the extension
\[
  \exact{U(1)}{\gext{\G}}{\G} .
\]
In \cite{HMSV2013} it was shown that the fibre in this extension can be replaced with any topological Abelian group. Thus by the transgression map we have a relation between the anomalous $2$-cocycles and the cohomology obstruction of the lifting problem. The latter can also be seen geometrically as the Dixmier-Douady class of the related bundle gerbe; more on this in Section~\ref{sec:gerbes_anomaly}

\idxloppu{current group}
%

%% Det bundle
\section{Determinant bundle} %
\idxalku{determinant bundle}

Determinant (line) bundle is a generalisation of the ordinary determinant of a finite-dimensional linear operator to a determinant-like geometric object related to a vector bundle over an infinite-dimensional manifold. Let us recall the familiar determinant construction for finite-dimensional spaces: Let $P: E \to F$ be a linear operator between $n$-dimensional vector spaces $E$ and $F$ over a field $K$. Any vector space $E$ has a natural \term{top exterior power} $\Top{n} E$ defined by the skew-symmetrised $n$th tensor power of $E$. Then one defines the \term{determinant of $P$} as a linear map $\det P : \Top{n} E \to \Top{n} F$. Note that since the top exterior powers are $1$-dimensional vector spaces, the determinant can be seen as a line $\Top{n} E^* \otimes \Top{n} F$ in $K$
-- in particular, if $E = F$, we have the usual association of the determinant as an element of the field $K$. The map $\det P$ retains many of the usual properties of the determinant, being nonzero if and only if $P$ is invertible.

We want to extend this definition over infinite-dimensional spaces of Fredholm operators $\mathcal{P}$. An early example of this was provided in~\cite{QUILLEN1985}, in which the determinant line bundle was studied over a space of Cauchy-Riemann operators. The gist of this construction is that for a given operator $P$ one can take as a fibre $\Top{n} (\ker P)^* \otimes \Top{n} \cker P$, and a family of such elements defines a line bundle on the space $\mathcal{P}$. If furthermore this determinant bundle can be trivialised, one could identify the bundle sections with functions on the base space and thus obtain a number associated to the operator and call this the determinant.%
	\footnote{For this definition to work, one must assume that the index $\dim(\ker P) - \dim(\cker P)$ is zero; compare the requirement $\dim(E) = \dim(F)$ for the finite-dimensional determinant. In this way one obtains a canonical section of the bundle which can be identified as a determinant map $\det P$ on $\mathcal{P}$, fulfilling the basic property of a vanishing determinant for non-invertible operators.}

Note that the determinant defines an endofunctor in the category of vector spaces:
\[
  (P:E \to F) \mapsto (\det P :\Top{n} E \to \Top{n} F) .
\]
This can be used to induce an endofunctor in the category of vector bundles over a manifold.
Let $E \to X$ be an $n$-dimensional vector bundle, and consider an open cover $\{\U_i\}$ of $X$. The vector bundle is then characterised by the transition functions $g_{ij}$ defined on intersections $\U_i \cap \U_j$ with values in the general linear group $GL(n)$. Taking now determinants of these transition functions we get transition functions for a bundle $\Top{n} E \to X$. We take this as the abstract definition of the determinant line bundle of $E$.

In the following we construct the determinant line bundle related to a Fock bundle over the space of gauge connections. The procedure essentially follows the one given in~\cite{CMM1997}.

Let us consider the odd-dimensional fermionic geometry $\left(M,S\otimes E,G \right)$. Each connection $A\in \A$ defines a massless Dirac operator $D_A$ acting on square-integrable sections of $S\otimes E$. For such operators we can define the space
\[
  \A_0 = \{(A,\lambda) :~ \lambda \notin \spec(D_A) \} \subset \A \times \C
\]
over $\A$. Let then $\U_{\lambda} \subset \A$ be open subsets such that $\U_{\lambda} = \{A\in \A :~ (A,\lambda) \in \A_0 \}$.

With respect to the given non-eigenvalue $\lambda$ there is a \term{spectral decomposition} %
	\idx{spectral decomposition}
\[
  V_A = V_{(A,\lambda)}^+ \oplus V_{(A,\lambda)}^-
\]
where $V_{(A,\lambda)}^+$ is the eigenspace of $D_A$ corresponding to eigenvalues larger than $\lambda$, and $V_{(A,\lambda)}^-$ is the orthogonal complement. Denote by $V_{(A,\lambda,\eta)}$ the intersection of $V_{(A,\lambda)}^+ \cap V_{(A,\eta)}^-$ for non-eigenvalues $\lambda < \eta$. We can then \emph{locally} define a canonical complex line bundle over an intersection $\U_{\lambda,\eta} = \U_{\lambda} \cap \U_{\eta}$,
\[
  \Det_{\lambda,\eta}(A) = \Top{\textrm{top}} V_{(A,\lambda,\eta)}
\]
for all $A \in \U_{\lambda,\eta}$. Since $M$ is assumed to be compact, the spectral subspace corresponding to the interval $[\lambda,\eta]$ is finite-dimensional and this line bundle is well defined. Furthermore, set $\Det_{\eta,\lambda}=\inv{(\Det_{\lambda,\eta})}$ when $\lambda < \eta$. By construction these local line bundles $\Det_{\lambda,\eta}$ fulfill the \term{cocycle condition} with respect to the open cover $\{\U_{\lambda}\}$ of $\A$:
\[
  \Det_{\lambda,\eta} \otimes \Det_{\eta,\mu} = \Det_{\lambda,\mu}
\]
over triple intersections $\U_{\lambda,\eta,\mu}$. The line bundles inherit a natural Hermitian structure from their definition as exterior powers of finite sub-bundles in a Hilbert space, and by restriction to $S^1 \subset \C$ they form $U(1)$-bundles over $\A$.

Now over for each $\U_{\lambda}$ we aim to define line bundles $\Det_{\lambda}$ fulfilling
\[
  \Det_{\eta} = \Det_{\lambda} \otimes \Det_{\lambda,\eta}
\]
over intersections $\U_{\lambda,\eta}$. That such a family exists is in fact equivalent to the cocycle condition of the local line bundles:

\begin{proposition}\label{prop:det_bundle}
	There is a family of line bundles $\{\Det_{\lambda} \to \U_{\lambda} \}$ for all open sets $\U_{\lambda}$ such that
	\[
		\Det_{\eta} = \Det_{\lambda} \otimes \Det_{\lambda,\eta}
	\]
	over intersections $\U_{\lambda,\eta}$ if and only if the line bundles $\Det_{\lambda,\eta} \to \U_{\lambda,\eta}$ fulfill the cocycle property.
\end{proposition}
\begin{proof}
	The family $\{\Det_{\lambda} \to \U_{\lambda} \}$ necessitates the cocycle property by construction.

	Conversely, assume that the line bundles $\Det_{\lambda,\eta} \to \U_{\lambda,\eta}$ fulfill the cocycle property. This local data is then enough to justify existence of the family since the base $\A$ is topologically trivial.
\end{proof}

Let $\F'_{\lambda}(A)$ be the local Fock space over $\U_{\lambda}$. Then, with respect to the local line bundles
\[
  \F'_{\eta}(A) = \F'_{\lambda}(A) \otimes \Det_{\lambda,\eta}(A) .
\]
Now we can define the \term{Fock bundle}
\[
	\F_{\lambda}(A) = \F'_{\lambda}(A) \otimes \Det_{\lambda}(A)
\]
over $\U_{\lambda}$. By construction, the bundle $\F_{\lambda}(A)$ is independent of $\lambda$ and hence gives a well-defined bundle of Fock spaces $\F \to \A$.

The structure of determinant bundles can also be understood in terms of bundle gerbes, as we see in Section~\ref{sec:gerbes_anomaly}.

\subsection{Anomaly and the bundle curvature}

The anomaly is now related to the lift of the gauge action of $\G$ on the base space $\A$ to the total space $\F$. The gauge action on the open sets $\U_{\lambda}$ can be naturally lifted to $\F'_{\lambda}(A)$, but the lift to the determinant bundle hinges on the triviality of the bundle. A non-zero curvature of the determinant bundle is a \emph{topological obstruction}, since it is given by the first Chern class of the determinant bundle and this in turn is the generator of the cohomology group $\coh^2(\A/\G,\Z)$.

First note that the line bundles descend naturally along the projection $\pi:\A \to \A/\G$ from $\U_\lambda \subset \A$ to open subsets $ \V_\lambda \subset \A/\G$, since the transformations in $\G$ act covariantly on the Dirac operator $D_A$ and the determinant bundle relies only on the spectral data of the operator. The question of the anomaly is trivial over the affine $\A$, but on the moduli space much depends on the group $G$.

For connections $\A$ in \emph{odd dimensions} we would expect that the determinant bundle is trivial~\cite{AS1984}, since then $\coh^2(\A/\G,\Z) = 0$. However, given the physical manifold $M$ of dimension $2k+1$, we have \emph{even-dimensional} geometry $M\times I$ since we are interested in \emph{paths} in $\A$ rather than in single fixed potentials. The physical interpretation of the interval $I$ is that it represents the time interval between a given time and the infinite past. The natural Dirac operator on this extended manifold is given as
\[
  D_{A(t)}^{even} = \partial_t + D_{A(t)} ,
\]
where the potential has the form $A(t) \defeq f(t)A + (1-f(t))A_0$ for some fixed smooth function $f:[0,1] \to \R$ such that $f(0) = 0$ and $f(1) = 1$ with the additional condition of being constant near these endpoints. The exact definition of the function $f$ is not important as it serves only as a crutch for the set-up of the curvature computation which does not ultimately depend on $f$. In order to have a well-posed boundary value problem we require that this even Dirac operator satisfies the conditions similar to the Atiyah-Patodi-Singer boundary conditions for the index theorem~\cite{APS1975}:
\begin{enumerate}
	\item at $t=0$, the spinor fields belong to the \emph{negative} eigenspace $V_{(A,\lambda)}^-$, and
	\item at $t=1$, the spinor fields belong to the \emph{positive} eigenspace $V_{(A,\lambda)}^+$.
\end{enumerate}
Due to the boundary conditions, the determinant bundle does not cover all of $\A$, since we must exclude connections with the eigenvalue $\lambda$ -- hence the parameter space is not affine anymore.

Consider then a family of gauge-transformed connections $A^g(p,z)$, where $p \in M$ and $z\in S^2 \subset \A$. The reason for the latter parametrisation is that the Chern classes can be represented as differential forms and determined completely by integrating over a closed surface in the parameter space~\cite{AS1984}. Thus we can consider the product manifold $M \times S^2$, and evaluate the curvature in the gauge directions to get the anomalous cocycles for the current algebra. The computation is described by the families index theorem, see Section~\ref{sec:index}.

\begin{remark}
	The structure of the determinant bundle and its relation to anomalies can also be studied via infinite-dimensional Grassmannian embeddings~\cite{MICKELSSON1990}. In a nutshell, one can think of the Fock spaces $\F'_\lambda(A)$ with the parameter space embedded in a Grassmannian manifold. Naively, the Fock bundle over $\A$ can be defined by pulling back through the map $A \to \F'_\lambda(A)$. The possibly nontrivial determinant bundle arises similarly to the case above when considering the action of the gauge transformation group on the bundle.
\end{remark}

\idxloppu{determinant bundle}
%

%% Bundle gerbes
\section{Anomalous bundle gerbes}\label{sec:gerbes_anomaly} %
\idxalku{bundle gerbe!anomalous}

The bundle gerbes introduced in Chapter~\ref{chap:higher} give us a natural geometric view on the lifting problem and the anomalous cohomology terms. Recall from Section~\ref{sec:gerbes} that a bundle gerbe $(P,Y)$ over a manifold $M$ is characterised by the following conditions:
\begin{enumerate}
  \item $\pi: Y \to M$ is a surjective submersion.
  \item $p: P \to Y^{[2]}$ is a $U(1)$-bundle.
  \item The multiplication $\pi^*_3(P) \otimes \pi^*_1(P) \to \pi^*_2(P)$ defines an associative isomorphism of $U(1)$-bundles over $Y^{[3]}$.
\end{enumerate}
We wish to apply this to the Fock bundle over the space of connections $\A$. Let $\A_0$ be the space of pairs $\{(A,\lambda) :~ \lambda \notin \spec(D_A)\} \subset \A \times \R$. We can then use the construction of the determinant bundle to define a bundle gerbe $(\F,\A_0)$ over the connections $\A$:
\[
	\gerbe{\F}{\A_0}{\A}
\]
The three conditions are fulfilled as follows.

\begin{enumerate}

	\item
	The projection $\pi: \A_0 \to \A$ is a surjective submersion since locally $\A_0 = \U \times \Os$ for open subsets $\U \subset \A$ and $\Os \subset \R$. Note that
	$\pi: \A_0 \to \A$ is not a fibration, however, since it is not locally trivial near the points $A \in \A$ for which $D_A$ has degenerate eigenvalues.

	\item
	We can take the fibre product $\A_0^{[2]}$ as the set $\{(A,\lambda,\eta)\}$ for $\lambda, \eta$ not in the spectrum of $A$. This defines a line bundle over $\A_0^{[2]}$ as
	\[
  	\F \defeq \F_{(A,\lambda,\eta)} = \begin{cases}
																				\Det V_{(A,\lambda,\eta)} ,
																					& \lambda \leq \eta \\
                                      	\Det V^*_{(A,\lambda,\eta)} ,
																					& \lambda \geq \eta
                          						\end{cases} ,
	\]
	where as before $V_{(A,\lambda,\eta)}$ is the sum of all eigenspaces for eigenvalues between $\lambda$ and $\eta$. Now $\F$ is a $U(1)$-bundle over $\A_0^{[2]}$, due to the local line bundles being $U(1)$-bundles over $\A$.

	\item
	For any $\mu > \lambda > \eta$ we recall the cocycle condition,
	\[
  	V_{(A,\mu,\lambda)} \oplus V_{(A,\lambda,\eta)} = V_{(A,\mu,\eta)}
	\]
	so that
	\[
	  \Det V_{(A,\mu,\lambda)} \otimes \Det V_{(A,\lambda,\eta)} = \Det V_{(A,\mu,\eta)} .
	\]
	The multiplication is then given by
	\[
	  \F_{(A,\mu,\lambda)} \otimes \F_{(A,\lambda,\eta)} = \F_{(A,\mu,\eta)} ,
	\]
	and this is associative by construction.
\end{enumerate}

As noted before, the Dixmier-Douady class of the bundle gerbe $(\F,\A_0)$ is an element in the third integral cohomology group of the base manifold, and this class measures the triviality of the bundle gerbe. Now again if we have the affine space $\A$ as the base the bundle gerbe indeed is trivial. We are however interested in the moduli space $\A/\G$, and consequently in the (possibly) anomalous bundle gerbe $(\F/\G,\A_0/\G)$ on $\A/\G$, where the group action of $\G$ extends naturally to $\F$.

The relationship of the bundle gerbe triviality to the curvature of the determinant bundle was explained in \cite{CMM1997}. For the projection $\pi: \A_0 \to \A$ one has a canonical section $A \mapsto (A,\lambda)$ over each subset $\U_{\lambda}$. With respect to this section the bundle gerbe is locally equivalent to the line bundle $\Det_{\lambda,\eta}$ over $\U_{\lambda,\eta}$, and the question of triviality amounts to defining a line bundle $\Det_{\lambda} \to \U_{\lambda}$ such that
\[
  \Det_{\lambda,\eta} = \Det_{\lambda}^* \otimes \Det_{\eta} .
\]
Now if $F_{\lambda}$ is the curvature $2$-form of $\Det_{\lambda}$, we can define the Chern class representative on $\U_{\lambda,\eta}$ for the line bundle $\Det_{\lambda,\eta}$ as
\[
  F_{\lambda,\eta} = F_{\lambda} - F_{\eta} ,
\]
making use of the projection $\pi_{\G}:\A \to \A/\G$ to the subset $\V_{\lambda,\eta} = \pi_{\G}(\U_{\lambda,\eta}) \subset \A/\G$. This means that these $2$-forms on $\U_{\lambda,\eta}$ are cohomologous to the $2$-forms descending to closed forms on $\V_{\lambda,\eta}$. The anomaly terms can be computed using these classes and we can establish the following~\cite[Thm.~4.1]{CMM2000}:
\begin{theorem}\label{theorem:dd_curvature}
	The Dixmier-Douady class of the bundle gerbe $(\F/\G,\A_0/\G)$ is represented by the family of closed curvature $2$-forms $F_{\lambda,\eta}$ on $\V_{\lambda,\eta} \subset \A/\G$. Furthermore, the anomalous Lie algebra $2$-cocycle is cohomologous to the negative of the cocycle given by the curvature of the corresponding determinant bundle along gauge orbits.
\end{theorem}
On the other hand, the existence of the anomaly and the triviality of the bundle gerbe is tied to the prolongation problem for the Abelian extension
\[
  \exact{\map{\A,S^1}}{\gext{\G}}{\G} .
\]
	\idx{bundle gerbe!lifting}
Recall from Section~\ref{sec:gerbes} that to Lie group extensions we can associate a lifting bundle gerbe -- in this case the very same $(\F/\G,\A_0/\G)$. The existence of the anomaly then boils down to the following theorem~\cite[Prop.~2.3]{HMSV2013}:
\begin{theorem}
	The principal $\G$-bundle $\A \to \A/\G$ lifts to a $\gext{\G}$-bundle if and only if the bundle gerbe $(\F/\G,\A_0/\G)$ is trivial.
\end{theorem}
The realisation of the Dixmier-Douady class as a differential form on the base manifold is again related to the families index theorem, see the next section for examples.

\begin{remark}
	As with the determinant bundle, the bundle gerbe and its Dixmier-Douady class extends to the case of infinite-dimensional Grassmannian manifolds~\cite{CM2012}. The point is that to obtain the characteristic class on a space-time manifold as a pull-back of the class on an operator space, one needs to consider determinant bundles over Grassmannians instead of the space of self-adjoint operators. One can then take the related gerbe structure as a \emph{universal gerbe} generalising many of the applications~\cite{CW2006}.
\end{remark}

\idxloppu{bundle gerbe!anomalous}
%

%% index theorem
\section{Index theorem}\label{sec:index} %
\idxalku{index theorem}

Having explored how the Hamiltonian anomaly is geometrically reflected both in the determinant bundle and in the bundle gerbe, we now come to the means of computing the anomaly: the index theorem. As explained in Chapter~\ref{chap:ktheory}, families of Dirac operators parametrised by connections in $\A$ defines an index bundle over the parameter space on the basis of the Atiyah-Singer index theorem; the families index is a formal difference
\[
  \ind D_A = \{\ker D_A\} - \{ \cker D_A \}
\]
as an element in the K-theory $K(\A)$. This difference is equivariant under the gauge transformations in $\G$, and thus the index descends to $K(\A/\G)$. Note that this does \emph{not} define a vector bundle over $\A$ since the nullspaces can vary -- and this harks back to the vacuum problem in quantisation.

In \cite{AS1984} the difference element was studied via the families index theorem and it was shown how one can compute the relevant characteristic classes explicitly via differential forms on an \emph{even-dimensional} base manifold. While the formal difference does not define a vector bundle, it is still possible to define characteristic classes in the sense of K-theory. In particular, the Chern character can be given as
\[
	\chern(\ind D_A) = \int_M \ahatgenus(M) \wedge \chern(\mathcal{E}) ,
\]
where $\mathcal{E}$ is a vector bundle over the product space $M \times \A/\G$.

This can be extended to the \emph{odd-dimensional} case as follows~\cite{CMM1997}: one can consider the odd-dimensional space-time manifold as a \emph{boundary} of an even-dimensional space, so that the index formula may give nontrivial results. Then one can show that the curvature of the determinant bundle is given by the index of the Dirac operator when restricted to the boundary. In the other direction, similar means can be used to express the Dixmier-Douady class of the anomalous bundle gerbe. As in Theorem~\ref{theorem:dd_curvature} we can conclude that these formulations are in fact equivalent in the sense of cohomology.

\idx{$\eta$-invariant}
For manifolds with boundary, in the index formula there are also nonlocal $\eta$-terms depending on the boundary data (see Remark~\ref{rem:aps_index}):
\[
	\ind D_A = \int_M \ahatgenus(M) \wedge  \chern(E) - \frac{1}{2}\eta(D_A) .
\]
To simplify this, we can first of all consider spherical base manifolds $M = S^n$ with a trivial $\ahatgenus$-genus. Secondly, since the $\eta$-terms are defined by the spectral data of the Dirac operator, they are invariant with respect to the gauge transformations. Since we are concerned with the moduli space $\A/\G$ rather than the space of connections $\A$, these nonlocal terms will not contribute. We are left with the evaluation of the Chern character, and this can be tackled with Chern-Simons forms under suitable simplifications.

\subsection{Chern-Simons forms and the boundary geometry} %
\idxalku{Chern-Simons form}

The index-theoretic computation of the curvature of the determinant bundle requires that the base manifold $M$ is the boundary of some larger manifold.
While this aspect of the geometry is a highly nontrivial question in general, many physically interesting problems can be formulated in this manner.

Given the boundary assumption, the anomaly terms in simple cases can be obtained by integration of the appropriate Chern-Simons form. We will not attempt to explore Chern-Simons forms and the related gauge theory to a great depth; the justification for this particular application is essentially as in the following example. See \cite{FREED1992} for a more detailed exposition, \cite{FSS2015} for a modern perspective, and \cite{CS1974} for the original account; finally, for the relationship of Chern-Simons forms to a variety of anomalies, see \cite[Ch.~7]{BERTLMANN2000}.

\begin{example}
	Let $M=S^1$ and $G = U(p)$, and consider a connection $1$-form $A$ on a vector bundle $E$ over $M$ with values in the matrix Lie algebra $\lie{u}_p$. We can define a real-valued $1$-form
	\[
		\frac{1}{2\pi i} \tr A
	\]
	which by integration over $M$ maps to real numbers. Given another connection $1$-form $A'$ (representing a different trivialisation of the bundle $E\to M$), the difference %of the functional
	\[
		\frac{1}{2\pi i} \int_M \tr A - \tr A'
	\]
	is an integer. This can be justified as follows. The $1$-form in question induces a $2$-dimensional functional in terms of the curvature of the connection:
	\[
	  \frac{1}{2\pi i}\int_{S^1} \tr A = \frac{1}{2\pi i}\int_{\bdr B^2} \tr A = \frac{1}{2\pi i}\int_{B^2} \tr F
	\]
	where $F \defeq \ext A$, and the last identity follows from the Stokes' theorem and the fact that the trace operator commutes with the de Rham differential. Choosing two different connections on $B^2$ such that they are connected by a gauge transformation on the boundary, we can glue together the two trivial bundles of these connections. On the base we then have $S^2$ as a disjoint union of the southern and northern hemispheres given by the two discs $B^2$ so that the boundaries become the equator:
	\[
		\frac{1}{2\pi i} \int_{S^1} \tr A - \tr A' = \frac{1}{2\pi i}\int_{B^2_N} \tr \ext A + \frac{1}{2\pi i}\int_{B^2_S} \tr \ext A' = \frac{1}{2\pi i}\int_{S^2} \tr F .
	\]
	But this is the so-called \emph{winding number}, the first Chern class of the bundle, which is a topological invariant with integer values. This particular formula has physical relevance in terms of the \emph{Dirac monopole}, see for example \cite[Sec.~6.4]{BERTLMANN2000}. For the case present the gist of the examples is this: given a principal bundle with a connection over a boundary manifold $M = \bdr N$, we can formulate Chern classes of the bundle in terms of a functional of the curvature $2$-form on the larger manifold $N$.
\end{example}

More generally, consider a product manifold $M \times X$ of dimension $\dim(M) + \dim(X) = n+k$. A closed integral $p$-form $\Omega$ on $M \times X$ defines a closed integral $(p-n)$-form $\Omega_X$ on $X$ through
\[
  \Omega_X = \int_M \Omega .
\]
Now for a Lie algebra valued connection $1$-form $A$ on $\mathcal{E} \to M\times X$ with the corresponding curvature $2$-form $F$, the even Chern classes $c_{2n}$ are invariant polynomials in $F$. To obtain these classes we wish to integrate the form $\Omega_X$ over a suitable closed surface in $X$, the dimension depending on the degree of the cohomology we are after. In the case of the Hamiltonian anomaly this is the $3$-form representation of the Dixmier-Douady class defined on the moduli space $\A/\G$.

\idxloppu{Chern-Simons form}

\subsection{The curvature of the determinant bundle} %
\idx{determinant bundle!curvature of}

We first recall the geometric fundamentals from \cite{AS1984}. Let $P \to M$ be a principal $G$-bundle, $\A$ the space of connections on the bundle and $\G$ the group of gauge transformations.
For the computation of the characteristic classes of the index in terms of forms, we define another bundle
\[
	\Lb \defeq \frac{P \times \A}{\G}
\]
which is a principal $\G$-bundle on $P \times \A$. Furthermore, since the group action of $G$ on $P \times \A$ commutes with the group action of $\G$, the group $G$ acts on $\Lb$. Since $P/G$ is homeomorphic to the base $M$, we can consider a principal $G$-bundle $\Lb$ on $M \times \A/\G$.

We have a Riemannian metric on $P \times \A$, invariant under the action of $G \times \G$. Given a point $(p, A) \in P \times \A$, the metric on the tangent spaces builds from the following data:
\begin{enumerate}
	\item Metric on $T(P, p)$. The metric on the tangent space $T(P, p)$ is given by the metrics of $G$ and $M$. The connection A defines horizontal and vertical projections, and the inner product splits such that the Riemannian metric on $M$ yields the inner product on the horizontal component, and the fixed invariant inner product on the Lie algebra $\lie{g}$ of $G$ yields the inner product on the vertical component (since $G$ is compact).

	\item Metric on $T(\A, A)$. The metric on $T(\A, A)$ is the usual metric on $\sct{\forms^1(M)\otimes \lie{g}}$: given tangent vectors $A_1, A_2$,
	\[
		\int_M \ip{A_1,A_2} \vol_M
	\]
	is the inner product.
\end{enumerate}

The metric on $P \times \A$ descends to a metric on $\Lb$, invariant under $G$. Then the orthogonal complement to orbits of $G$ (with respect to the induced Riemannian metric) yields a horizontal subspace and thus defines a connection. We denote by $\omega$ the related connection form: it is a vertical $1$-form on $\Lb$ with values in the Lie-algebra $\lie{g}$.

\idx{universal bundle}
The bundle $\Lb$ is called the \term{universal bundle}, and we have the following commutative diagram of fibrations:
\[
	\begin{tikzcd}
		P \times \A \arrow{r} \arrow{d} & \Lb \arrow{d} \\
		M \times \A \arrow{r} & M \times \A/\G 
	\end{tikzcd}
\]
where $P \times \A \to M \times \A$ is a principal $G$-bundle with a canonical connection. Moreover, the tangent mapping $T(P \times \A) \to T(\Lb)$ gives the splitting of $T(\Lb)$, and hence a principal connection on $\Lb$.

Earlier we defined Dirac operators on the space-time manifold $M\times I$ as
\[
  D_{A(t)}^{even} = \partial_t + D_{A(t)} ,
\]
where the extended potential is $A(t) = f(t)A + (1-f(t))A_0$, with respect to a fixed connection $A_0$ defining the spectral decomposition. To obtain the Schwinger terms one needs to know the curvature of the determinant bundle along gauge directions on the boundary, so we consider $D_{A(t=1)}^{odd}$. Moreover, we can focus on some closed surface on the moduli space -- take $S^2 \subset \A/\G$ -- since the integral of the first Chern class of the bundle equals the families index. In order to calculate the curvature, we can choose any connection $A$ on the product space $S^2 \times M \times I$, as long as its projection to $M\times I$ is equal to the extended potential given above and it fulfills the necessary boundary conditions.~\cite{CMM1997}

As noted above, the Atiyah-Patodi-Singer index formula has two parts: the local part given by the Chern character and the boundary part given by the $\eta$-invariant. However, since the $\eta$-term is defined by the spectral data of the Dirac operator, it is gauge invariant and thus vanishes when considering the operators modulo gauge transformations. What remains are the local Chern class terms.

Now the connections $A$ are parametrised by $z \in S^2 \subset \A/\G$, and the same holds for the gauge transformations $\G$. We then formulate the index in terms of the Dirac operator $D_A$ subject to these transformations. %
	\idx{Chern-Simons form}
	\idx{Chern class}
In simple situations the Chern classes $c_k$ are given by the integration of the Chern-Simons form of appropriate dimensions: assuming that $M$ is a sphere $S^{2k+1}$, we have the relation
\[
	\ext(CS_{2k+3}) = c_{k+2} ,
\]
where $CS_{2k+3}$ is a Chern-Simons form. By Stokes' theorem we can then write the boundary integral
\[
	\int_{S^2 \times M} CS_{2k+3}(A(z,p,1)) ,
\]
for $z \in S^2$ and $p \in M$. The form is identically zero at $t=0$.

Finally, the curvature $F_A$ of the determinant bundle at the point $A$ in the gauge directions $(x,y)$ is given by integrating over $M$. For instance, if $M=S^1$, we have the third Chern-Simons form
\[
	\frac{1}{8\pi^2} \tr (A \ext A + \frac{2}{3} A^3)
\]
and the curvature form is given by
\[
	F_{A}(x,y) = \frac{1}{4\pi} \int_{S^1} \tr A(p) \comm{x}{y} .
\]
Similarly, for $M=S^3$, one integrates the fifth Chern-Simons form
\[
	\frac{i}{24\pi^3} \tr ( A (\ext A)^2 + \frac{3}{2}A^3 \ext A + \frac{3}{5} A^5 )
\]
and the result is cohomologous to the familiar Mickelsson-Faddeev cocycle.

\subsection{The Dixmier-Douady class} %
\idxalku{Dixmier-Douady class}

Let $M$ now be our odd-dimensional space manifold and take $X = \A/\G$. Now for the Dixmier-Douady class we need to integrate the form $\Omega_X$ over a closed $3$-surface in the moduli space $X$; hence we can pick the $3$-sphere $S^3 \subset X$. By pulling back to $\A$ (both the surface and the connection), the integral becomes
\[
  \int_{S^3} \Omega_X = \int_{M \times B^3} c_{2k}(F) .
\]
Again, this is equal to the odd Chern-Simons form on the boundary $S^2 = \bdr{B^3}$:
\[
  \int_{M \times S^2} CS_{2k-1}(A) .
\]

Evaluations over gauge directions yields a familiar example for $M=S^1$:~\cite{CMM1997}
\[
	\int_{S^3} \Omega_X = \frac{1}{24\pi^2} \int_{S^1 \times S^2} \tr (\inv{g}\ext g)^3 .
\]
As in the above example, the map $g:M\times S^2 \to G$ is a gauge transformation linking the gauge connections on the boundary $S^2$. In higher dimensions for $M = S^{2k+1}$ we get
\[
	\int_{S^3} \Omega_X = \frac{-1}{(k+2)!(2k+3)}\left(\frac{i}{2\pi}\right)^{k+2} \int_{M\times S^2} \tr (\inv{g}\ext g)^{2k+3} .
\]

In general (in conjunction to Theorem~\ref{theorem:dd_curvature}) one can show that this indeed is the Dixmier-Douady class of the bundle gerbe:
\begin{theorem}[{\cite[Thm.~4.2]{CMM2000}}]
	The class $\Omega_X \defeq \Omega_{\A_0/\G}$ is a representative of the Dixmier-Douady class of the bundle gerbe $(\F/\G,\A_0/\G)$.
\end{theorem}
\begin{proof}
	See~\cite{CMM1997}.
\end{proof}

\begin{remark}
Note that in these index computations possible torsion information may not be accounted; no general methods are known since the map from K-theory to de Rham cohomology tends to lose this information. See \cite{CMM2000} for examples on how to deal with the torsion.
\end{remark}

\idxloppu{index theorem}
\idxloppu{Dixmier-Douady class}
%

%% Schwinger terms
\section{Modified gauge transformations and current commutators}
\label{sec:general_comm} %
\idxalku{connection!moduli space of}
\idxalku{connection!gauge}

We now come back to the question of current commutators when we have not only the internal symmetry group $G$ and the induced gauge transformations $\G$, but the modified transformation group $\Gbig$ which accounts for all components of the full space-time geometry. By this we mean the group of all (based) vertical automorphisms of the vector bundle $S \otimes E \to M$. To this we could associate the principal $\spin \otimes G$ -bundle;
here we are abusing the notation slightly in that the tensor product $\spin \otimes G$ refers to the tensor product of linear representations on the typical fibre, rather than to the tensor product of groups themselves in the sense of \cite{BL1987}.
\idx{vertical automorphism}

In general, the connections $\omega$ on the bundle $S \otimes E \to M$ are defined as
\begin{equation}\label{eq:all_connections}
	\omega = \omega_S \otimes 1_E + 1_S \otimes \omega_E ,
\end{equation}
where $\omega_S$ is a \term{spin connection on $S$} and $\omega_E$ is a \term{Yang-Mills connection on $E$}. On the fermion fields $\psi \defeq \sigma \otimes e \in \sct{S\otimes E}$ the connection acts as
\[
	\omega(\psi) = \omega_S(\sigma) \otimes e + \sigma \otimes \omega_E(e) .
\]
As usual, we can write the connection as a $1$-form in $\forms^1(M)\otimes\lie{g}_S$, in which the extended Lie algebra $\lie{g}_S$ is the tensor product $\lie{spin}\otimes\lie{g}$ with $\lie{spin}$ generated by gamma matrices $\gamma_{\mu}$ of a dimension matching with that of $M$. In this modification we then use the general form of the gauge currents
\[
   j_{\mu}^a(x) = \psi^{\dagger}(x)( \alpha_{\mu} \otimes \tau^a ) \psi(x)
\]
with $\alpha_{\mu} = \gamma_0 \gamma_{\mu}$.

If we now apply vertical transformations given by $\Gbig$ on general connections of the form \eqref{eq:all_connections}, we will get tensor products of mixed components defining all the (orbits of) possible connections on the product bundle. Let us call this space $\tilde \A$. By definition, any connection $\omega' \in \tilde A$ is gauge equivalent to another connection $\omega$ of the form \eqref{eq:all_connections}, and there is a natural \term{restriction map}
\[
	\rho : \omega \mapsto \omega_E \in \A,
\]
which is surjective by the definition of $\tilde A$. In general, the map $\rho$ does not extend to a bijection between equivalence classes, since for any given connection $\omega_E$ in $\A$ the preimage $\inv{\rho}(\omega_E)$ in $\tilde \A$ may contain a number of connections not related by gauge transformations. This can be seen concretely by choosing another spin connection $\omega'_S \not\sim \omega_S$ in Equation~\eqref{eq:all_connections} while keeping $\omega_E$ fixed. Moreover, there is no homotopy
\[
	\tilde \A/\Gbig \simeq \A/\G ,
\]
unless $\tilde \A / \Gbig$ is topologically trivial.

The situation is not as clear if we have a \emph{fixed spin connection} $\Sigma$ on $S$, induced from an affine connection on $M$. %
Now for any connection $\omega_E$ on $E\to M$, we can write
\begin{equation}\label{eq:fixed_spin}
	\omega \defeq \Sigma \otimes 1_E + 1_S \otimes \omega_E .
\end{equation}
The action of $\Gbig$ on connections of this form for all $\omega_E \in \A$ then induces the \term{modified space of gauge connections} $\Abig$:
\[
	\Abig \defeq \left\{ (\Sigma \otimes 1_E + 1_S \otimes \omega_E)^g :~ \omega_E \in \A , \, g \in \Gbig \right\}.
\]
It is the union of all $\G_S$-orbit spaces thus obtained, and as such a subspace of the space of all connections on the product bundle. Again any connection $\omega' \in \Abig$ is by definition equivalent to a connection $\omega$ in the basic form \eqref{eq:fixed_spin} for some $\omega_E \in \A$, and there is a surjective restriction map $\rho : \omega \mapsto \omega_E$.

Let us denote the equivalence classes in $\A$ with respect to the action of $\G$ by $[\omega_E]^{\G}$, and similarly for the classes in $\Abig$ with respect to $\Gbig$ by $[\omega]^{\Gbig}$. Note that we can injectively map each $g\in \G$ and each $\omega_E \in \A$ as follows
\begin{align*}
	g &\mapsto 1_S\otimes g \in \Gbig \\
	\omega_E &\mapsto 1_S \otimes \omega_E \in \Abig .
\end{align*}
Now the question is if the restriction map $\rho$ extends to a map
\[
	\rho_* : \Abig/\Gbig \to \A/\G : [\omega]^{\Gbig} \mapsto [\omega_E]^{\G}
\]
between the two moduli spaces. 
In other words, if we have a transformation $h \in \Gbig$ such that
\[
	h: \Sigma \otimes 1_E + 1_S \otimes \omega_E \mapsto \Sigma \otimes 1_E + 1_S \otimes \omega_E' ,
\]
we would want to find such $g \in \G$ that the following diagram commutes:
\[
	\begin{tikzcd}[column sep=large,row sep=large]
		\Sigma \otimes 1_E + 1_S \otimes \omega_E \arrow[r,"h"] \arrow[d,"\rho"] & \Sigma \otimes 1_E + 1_S \otimes \omega_E'  \arrow[d,"\rho"]\\
    \omega_E \arrow[r,"g"] & \omega_E'
	\end{tikzcd}
\]
We see that the induced map $\rho_*$ is in fact a homeomorphism of the moduli spaces.

\begin{proposition}\label{prop:moduli_space}
	The moduli space of gauge connections is homeomorphic to that of the modified space,
	\[
		\Abig/\Gbig \simeq \A/\G .
	\]
\end{proposition}
\begin{proof}

	Consider a transformation $h\in \Gbig$ such that $\omega^h = \Sigma \otimes 1_E + 1_S \otimes \omega_E'$.
	Note that since the gauge transformations are based, such $h$ is unique. Now locally we can write the transformation in terms of $\End(S)\otimes\End(E)$-valued differential forms:
	\begin{align*}
		\Sigma \otimes 1_E  + 1_S \otimes \omega_E' &= \inv{h}(\Sigma \otimes 1_E + 1_S \otimes \omega_E) h + \inv{h}\ext h \\
		\Longleftrightarrow~ \ext h &= h(\Sigma \otimes 1_E  + 1_S \otimes \omega_E') - (\Sigma \otimes 1_E  + 1_S \otimes \omega_E) h .
	\end{align*}
	This is a linear first order differential equation, and its solution factorises to elements in the subspaces $\End(S)$ and $\End(E)$. Since the spin connection $\Sigma$ remains unchanged, it follows that the solution must be in the linear subspace $1_S \otimes \End(E)$.
	This is possible if and only if $h = 1_S \otimes g$ for some $g \in \G$. Thus $\omega_E' = \omega_E^g$, and we can define a map
	\[
		\rho_* : \Abig/\Gbig \to \A/\G .
	\]
	By definition this map is surjective, and it is also injective since for any equivalence $\omega_E \sim \omega_E^g$ in $\A$, there is an equivalence in $\Abig$:
	\[
		\omega = \Sigma \otimes 1_E  + 1_S\otimes \omega_E \sim \Sigma \otimes 1_E  + 1_S \otimes \omega_E^g = \omega^{1\otimes g}.
	\]
	Thus $\rho_*$ is a bijection, and the moduli spaces $\Abig/\Gbig$ and $\A/\G$ are homeomorphic.
\end{proof}

\begin{remark}
	Later on we consider the case with $\dim(M)=3$. Then if we have a representation in the space of unitary matrices $U(p)$, the modified group is represented in the larger space $U(2p)$, where the factor of $2$ is due to the defining representation of Weyl fermions in $\C^2$. By the same token, we can replace the gamma matrices by the Pauli matrices $\sigma_{\mu}$, with $\sigma_0 = I$, and write $\lie{g}_S = \lie{u}(n) \otimes \lie{u}(2)$ for the modified Lie algebra.
	Moreover, for $M=S^3$ there is a diffeomorphism $M \isom SU(2)$, and the spin bundle $S = M \times SU(2)$ is trivial. The structure group for the tangent bundle of $M$ is $SO(3)$, for which the spin group $\spin(3) \isom SU(2)$ is the double-covering group.
\end{remark}

\idxloppu{connection!moduli space of}
\idxloppu{connection!gauge}
\idxalku{Dixmier-Douady class}
From Chapter~\ref{chap:ktheory} we know that the Dixmier-Douady class represents a K-theory element of the family and hence it is homotopy invariant. Now, if the two moduli spaces are homeomorphic, the modification of the transformation group does not change the topological basis for the anomaly terms on the side of the moduli space.

Consider the natural projection of the connections to the moduli space
\[
	\pi : \A \to \A/\G ,
\]
where we define the Dixmier-Douady representative as a $3$-form $\omega_3$ over the moduli space $\A/\G$. 
By the definition of the class $\omega_3$ is closed in $\A/\G$. The pullback $\pull{\omega_3}$ through the projection is closed in the direction of gauge orbits, and since $\A$ is affine, there is a $2$-form $\theta$ on $\A$ representing a $2$-cocycle of the Lie algebra $\lie{g} \defeq \Lie(\G)$ and fulfilling
\[
	\ext \theta = \pull{\omega_3} .
\]
This follows from the basic idea that the Lie algebra $\lie{g}$ generates smooth vector fields on $\A$ via the right action of $\G$. Then, if there is a $2$-form $\theta$ closed in the vertical directions dictated by the gauge orbits, we obtain a $2$-cocycle on the Lie algebra with values in $\map{\A,i\R}$ by evaluating the form along the generated vector fields. Thus we have the \emph{cohomology transgression} $\coh^3(\A/\G,\Z) \to \coh^2(\lie{g},\map{\A,i\R})$ as in Subsection~\ref{sec:transgression}.

 \idx{universal bundle}
Assume that the base manifold $M$ is an odd-dimensional sphere. Following the outline given in the previous sections, the $3$-form representation $\Omega$ as a characteristic class is evaluated over a closed orientable $3$-surface in the moduli space; we can choose $S^3 \subset \A/\G$ without loss of generality. Then,
\[
	\int_{S^3} \Omega = \int_{M\times B^3} c_{2n}(\pull{\F})
\]
where the Chern class is given in terms of the curvature $2$-form $\F$ of the \term{universal connection} $1$-form $\Auni_{G}$ over $M\times \A/\G$ with values in the Lie algebra $\Lie(G)$. The ball $B^3$ is the pullback of the surface $S^3$ on the space of connections $\A$. In the general case such an evaluation of a closed form on a product manifold $M\times \A/\G$ leads to a rather involved expression with nonlocal terms coming from the Green's operator~\cite{AS1984}. %
Formally, we have components according to the decomposition on $M\times \A$, which descends to $M\times \A/\G$:
\[
	\forms^2(M\times\A) = \forms^2(M) \oplus(\forms^1(M)\otimes\forms^1(\A)) \oplus \forms^2(\A) .
\]
Here the last component is not local. At a point $(x,A) \in M \times \A/\G$ we can write for tangent vectors $a,b$ in gauge directions $\F^{(0,2)}_{x,A}(a,b) = -2 G_{A} \comm{a}{b}(x)$, where $G_A \defeq \inv{(\ext_A^{*} \ext_A^{})}$ is the Green's operator for the Laplacian of the exterior covariant derivative, expressed by $\ext_A \defeq \ext + \comm{A}{\cdot}$ in the adjoint representation. The formal adjoint of the derivative is locally defined by $\ext_{A}^{*}(b) \defeq \partial_i b^i + \comm{A^i}{b^i}$
for $b$ in the tangent space of $\A/\G$ at the point $A$.%
	\footnote{What we mean by \term{gauge directions} is the following: The condition $\ext_A^{*}(b) = 0$ at a point $A \in \A$ sets the tangent vector $b\in T_A(\A)$ in the \term{background gauge}. Since any tangent vector at $\pi(A) \in \A/\G$ is by definition in the background gauge, we can make this assumption without loss of generality.}
In principle this curvature can then be used to evaluate the characteristic classes on the gauge bundle. For instance, if $G=SU(p)$ then in three dimensions $M=S^3$ we can represent a closed $3$-form $\Omega$ on the moduli space as
\[
	\Omega = k \int_M \tr \F^3 ,
\]
where $k$ is the normalisation for the integral class.

However, if we restrict to $S^3 \subset \A/\G$ and assume that the gauge bundle associated to the connection is trivial so that the pullback $\Auni$ of the universal connection to $M\times \A$ is globally defined, we need only the local terms of the curvature. Now we can evaluate the Chern class by the boundary integral of a suitable Chern-Simons form:
\[
	\int_{S^3} \Omega = \int_{M\times B^3} c_{2n}(\pull{\F}) = \int_{M\times S^2} CS_{2n-1}(\Auni) .
\]
For $M=S^3$, the integrand is the fifth Chern-Simons form
\[
	CS_5(\Auni) = \frac{i}{24\pi^3} \tr \left( \Auni(\ext \Auni)^2 + \frac{3}{2}\Auni^3 \ext \Auni + \frac{3}{5}\Auni^5 \right) ,
\]
and the Lie algebra $2$-cocycle is revealed by transgression: for $M=S^3$ it can be shown to be cohomologous to the Mickelsson-Faddeev cocycle.~\cite{CMM1997}

Due to the homotopy invariance, we can consider the same $3$-form representative $\Omega$ on $\Abig/\Gbig$. Now also the pullback to $\Abig$ fulfills the same coboundary relation, that is, there is a $2$-form $\eta$ on $\Abig$ such that $\pull{\Omega} = \ext \eta$. This $2$-form $\eta$ is the (possibly anomalous) representation of the current algebra $2$-cocycle given as a characteristic class on the gauge bundle. All of the above then holds true also when we use the modified moduli space $\Abig/\Gbig$. The main difference is that now the local formulas for the Schwinger terms will be slightly more complicated since the gauge algebra is a tensor product space.

We can illustrate this in the case of $\dim(M) = 3$. Our main objective is to recover the local Schwinger terms for all components of the current commutators by working in the larger transformation group $\G_S$.

\begin{example}[$M=S^3$] %
	\idxalku{current group}
When the base manifold is a sphere $S^n$ the geometry of the moduli space simplifies considerably. In particular, the moduli space $\A/\G$ is homotopic to the space of based maps $\map{S^{n-1},G}$, and the gauge transformation group $\G$ can be identified with the based loop group of the moduli space~\cite{SINGER1981}. Note that in general $\A/\G$ is not simply connected, nor is $\G$ connected; we can restrict to the connected component of the identity if necessary.

Let us first consider the standard case of connections in $\A/\G$ for $M=S^3$. The moduli space is homotopic to $S^2G$ and the space of connections $\A$ can be identified with the group of based smooth paths $B^3G$, that is, a connection $A\in \A$ is a smooth path on $S^2G$ starting from the identity, and the projection $\pi:B^3G \to S^2G$ is given by the evaluation of the path at its end point.

\idx{Bott periodicity}
What we need is a concrete representative of the Dixmier-Douady class on the moduli space. If $G=SU(p)$ with $p>2$, the odd cohomology groups are generated by elements $\alpha_{2k-1}$ for small $k\in\N$ (due to the Bott periodicity, see Remark~\ref{rem:bottcohomology}). Following the basic principles of the Chern-Simons theory~\cite{CS1974}, the standard representation can be written in terms of the Maurer-Cartan form $\inv{g}\ext g$, where $g: B^3 \to G$ is a smooth map. Here we assume that the maps $(x,z) \mapsto g(x,z) : B^3 \times S^2 \to G$ identify the gauge connections on the boundary
$\bdr B^3 = S^2$. We can then recover the Dixmier-Douady form by evaluating the Chern-Simons form on the product manifold $M \times S^2$. %
	\idx{Chern-Simons form}
Along gauge directions we get
\[
	\int_{S^3 \times S^2} \Omega = 	c_3 \int_{S^3 \times S^2} \tr (\inv{g}\ext g)^5 ,
	%= \frac{c_3}{2} \int_{S^3\times S^2} \tr (\inv{g}\ext g \wedge \comm{\inv{g}\ext g}{\inv{g}\ext g}^2),
\]
where
\[
	c_{2k+1} \defeq -\left(\frac{i}{2\pi}\right)^{k+2} \frac{1}{(k+2)!(2k+3)} .
\]
is the normalisation coefficient ensuring that the class is indeed integral.~\cite{CMM1997} %

The transgression follows from evaluating the Chern-Simons form over $M$ for tangent vectors $u,v,w:S^2 \to \lie{g}$. The generator of the Lie algebra cohomology on the moduli space in its simplest form is the following:
\[
	\int_{S^2G} \omega = c_2 \int_{S^2} \tr u \comm{\ext v}{\ext w} ,
\]
for $c_2 \defeq i/(24\pi^2)$. A direct computation verifies that this is indeed cohomologous to the transgressed form of $\Omega$.
Lifting then to a $3$-form on $B^3$ by Stokes' theorem we get
\[
	\int_{B^3} \tr \ext u \comm{\ext v}{\ext w}  = \int_{B^3} \tr \left( \ext u \comm{\ext v}{\ext w} - \ext v \comm{\ext w}{\ext u} + \ext w \comm{\ext u}{\ext v} \right) .
\]
In the Lie algebra cohomology this is the coboundary of the familiar Mickelsson-Faddeev cochain: %
	\idx{Mickelsson-Faddeev cocycle}
\[
	\theta_2(\inv{f}\ext f;u,v) = c_2 \int_{B^3} \tr \inv{f}\ext f \comm{\ext u}{\ext v} ,
\]
where the Maurer-Cartan form $\inv{f}\ext f$ is defined for a smooth map $f : B^3 \to G$. The coboundary is given by the formula
\begin{align*}
	\cob \theta_2(A;u,v,w) \defeq &~\lder_u \theta_2(A;v,w) - \lder_v \theta_2(A;w,u) + \lder_w \theta_2(A;u,v) \\
		& - \theta_2(A;\comm{u}{v},w) + \theta_2(A;\comm{w}{u},v) - \theta_2(A;\comm{v}{w},u) ,
\end{align*}
where the Lie derivative acts on the potential $A\defeq \inv{f}\ext f$ as $\lder_u A \defeq \comm{A}{u} + \ext u$.

Since we are interested in the gauge directions with the boundary points of $\A = B^3G$ identified, the $2$-form $\theta_2$ is evaluated over tangent vectors $u,v$ vanishing on the boundary of $B^3$: hence it is \emph{closed in the gauge directions}. %
	\idx{transgression}
The transgression then gives a $2$-\emph{cocycle} on $S^3$:
\[
	\theta_2(\inv{f}\ext f;u,v) = c_2 \int_{S^3} \tr \inv{f}\ext f \comm{\ext u}{\ext v} .
\]
This the essence of the transgression: we have a map between classes given by
\[
	\int_{S^2G} (\cdot) \to \int_{S^3} (\cdot)
\]
and thus can obtain the desired Lie algebra $2$-cocycle.
\idxloppu{Dixmier-Douady class}
%

%\newpage

\begin{remark}%
		\idx{Bott periodicity}
	Note that the use of standard generators for the group cohomology as above applies only with $G=SU(p)$ with $p$ suitably large. The cohomology representation is in general more complicated for arbitrary Lie groups.
\end{remark}

\idxloppu{current group}

Let us then move to the modified group $\Gbig$ in the case of Weyl fermions. The generalised connections $\A_S$ can be recovered simply by letting the larger group $\Gbig$ act on the subspace $\A$. By Proposition~\ref{prop:moduli_space} the modified moduli space is equivalent to the original one, and so we have the same pullback mechanism for the Dixmier-Douady class and the cohomology transgression has the same form. The tensor products in the Lie algebra basis will not change the fact that pullback will be the coboundary of a $2$-form representing a Lie algebra $2$-cocycle cohomologous to Mickelsson-Faddeev cocycle.

\idxalku{Schwinger terms}
\idxalku{current algebra}
We can now have a closer look at the arising Schwinger terms. Now the Lie algebra valued forms $u,v$ and $A\defeq \inv{f}\ext f$ can be written with respect to linear combinations of the basis $\sigma_{\mu} \otimes \tau^a$ built from Pauli matrices and the generators of the current algebra -- as before, we will call the components of the tensor product \term{spin components} and \term{Yang-Mills components}, correspondingly. To illustrate how this looks, we can write simple two-component tensor commutators as
\begin{align}\label{eq:tensor_commutators}
	\comm{x_s \otimes x^t}{y_s \otimes y^t} &= (x_s y_s)\otimes(x^t y^t) - (y_s x_s)\otimes(y^t x^t) \nonumber \\
	&= \comm{x_s}{y_s} \otimes x^t y^t + y_s x_s \otimes \comm{x^t}{y^t} ,
\end{align}
where we have used the mixed Kronecker product. If we write $x_s \otimes x^t \defeq x^{\mu}_a \sigma_{\mu} \otimes \tau^a$, 
we can split the commutator into separate cases depending on the spin indices relevant to the gauge currents. In the following we use \emph{fixed indices $\mu$ and $\nu$}, and no summation is intended to be carried over them.
\begin{enumerate}
	\item $\mu = \nu$: the only contributing term is
		\[
			y_b^{\mu}x_a^{\mu}I \otimes \comm{\tau^a}{\tau^b} .
		\]
	\item $\mu = 0$ and $\nu \neq 0$: the contributing term is
		\[
			y_b^{\nu} x_a^0 \sigma_{\nu} \otimes \comm{\tau^a}{\tau^b} .
		\]
	\item $\mu \neq \nu$ and neither are zero: the contributing terms are
		\begin{align*}
			& 2i\epsilon_{\mu\nu\eta} \, x_a^{\mu} y_b^{\nu} \sigma_{\eta} \otimes \tau^a \tau^b + i\epsilon_{\nu\mu\eta} y_b^{\nu} x_a^{\mu}\sigma_{\eta} \otimes \comm{\tau^a}{\tau^b} \\
			&~= i\epsilon_{\mu\nu\eta} x_a^{\mu} y_b^{\nu} \sigma_{\eta} \otimes \acomm{\tau^a}{\tau^b} .
		\end{align*}
\end{enumerate}
Here the Levi-Civita symbol $\epsilon_{\mu\nu\eta}$ goes through the indices $\{1,2,3\}$, and we have defined $\sigma_0 \defeq I$.

We can now formulate the modified current commutators for the general current components defined by
\[
	j_a^{\mu}(x) = \psi^{\dagger}(x)\left(\sigma_{\mu}\otimes\tau^a \right)\psi(x) .
\]
First note the naive relations involving the non-anomalous terms in the commutators $\comm{j_a^{\mu}(x)}{j_b^{\nu}(y)}$ for fixed indices $\mu$ and $\nu$:
\begin{enumerate}
	\item\label{spin-index:same} $\mu = \nu$:
		\begin{align*}
			\comm{j^a_{\mu}(x)}{j^b_{\mu}(y)} &= \psi^{\dagger}(x) (I \otimes \comm{\tau^a}{\tau^b}) \psi(x) \delta(x-y) \\
				&= \lambda^{ab}_c j^c_0(x) \delta(x-y) .
		\end{align*}

	\item\label{spin-index:zero} $\mu = 0$ and $\nu \neq 0$:
		\begin{align*}
			\comm{j^a_{0}(x)}{j^b_{\nu}(y)} &= \psi^{\dagger}(x) (\sigma_{\nu} \otimes \comm{\tau^a}{\tau^b}) \psi(x) \delta(x-y) \\
			&= \lambda^{ab}_c j^c_{\nu}(x) \delta(x-y) .
		\end{align*}

	\item\label{spin-index:nontrivial} $\mu \neq \nu$ and neither are zero:
		\begin{align*}
			\comm{j^a_{\mu}(x)}{j^b_{\nu}(y)} &= i\epsilon_{\mu\nu\eta} \psi^{\dagger}(x) (\sigma_{\eta} \otimes \acomm{\tau^a}{\tau^b}) \psi(x) \delta(x-y) \\
			&= i \epsilon_{\mu\nu\eta}d_{abc} j_{\eta}^c(x) \delta(x-y).
		\end{align*}

\end{enumerate}
Assuming $G=SU(p)$ in the purely spatial case~\ref{spin-index:nontrivial}, the \emph{defining representation}
gives the anticommutator
\[
	\acomm{\tau^a}{\tau^b} = \frac{1}{p}\delta_{ab} I + d_{abc} \tau^c ,
\]
where $d_{abc} \defeq \frac{1}{2} \tr \acomm{\tau^a}{\tau^b}\tau^c$. Note that we are now indeed outside the scope of the original Lie algebra $\lie{g}$: the anticommutator is not in the Lie algebra $\lie{su}(p)$ but in the universal enveloping algebra. For the modified gauge algebra we can make the assumption of working within $\lie{u}(p)$ and then the anticommutator is simply
\[
	\acomm{\tau^a}{\tau^b} = d_{abc} \tau^c .
\]

\idxalku{anomaly!commutator}
These commutators need to be amended by the appropriate Schwinger terms $\alpha^{ab}(x,y)$. Plugging the tensor products into the Mickelsson-Faddeev cochain formula gives cases similar to the above in terms of the potential $A \in \Lambda^1(M) \otimes \lie{g}_S$ and of the exterior derivatives for $u,v \in \Lambda^0(M) \otimes \lie{g}_S$. Note that in the ordinary case we can write the local expression $A = A_i(x)\ext x^i$, where $A_i(x)$ is a function on $M$ with values in the Lie algebra $\lie{g}$ so that $A_i(x) = A_{i,c}(x) \tau^c$. In the modified Lie algebra $\lie{g}_S$ the local forms are
\begin{align*}
		A &= A^{\eta}_{k,c}(x) (\sigma_{\eta}\otimes \tau^c) \ext x^k ,\\
		\ext u &= \partial_i u^{\mu}_{a}(x) (\sigma_ {\mu} \otimes \tau^a) \ext x^i, \text{ and }\\
		\ext v &= \partial_j v^{\nu}_b(x) (\sigma_{\nu}\otimes \tau^b) \ext x^j .
\end{align*}
This maze of indices needs some clarification. The coordinate $x$ is naturally a point in the base manifold $M$, and the coordinate maps $x^i:M\to\R^n$ are local with respect to this given point. In the coefficient function $A^{\eta}_{k,c}(x)$ the indices $\eta$ and $c$ are in relation to the current algebra built from the tensor product $\sigma_{\eta}\otimes \tau^c$. The index $k$ is the local index of the differential form defined on $M$ with respect to the chosen coordinates, and with slight abuse of notation we then denote by $A^{\eta}_c$ the $1$-form on $M$ with respect to the basis in $\lie{g}_S$. As before we follow the convention for summing over repeated indices, unless otherwise stated.%

We can then write the local expression involving the commutator for the derivatives as follows:
\begin{align*}
	\comm{\ext u}{\ext v} =& \comm{ \ext u_a^{\mu} ( \sigma_{\mu} \otimes \tau^a)}{\ext v_b^{\nu} (\sigma_{\nu} \otimes \tau^b)} \\
	=& \partial_i u_a^{\mu}(x) \partial_j v_b^{\nu}(x) \left(\comm{\sigma_{\mu}}{\sigma_{\nu}}\otimes \tau^a \tau^b + \sigma_{\nu} \sigma_{\mu} \otimes \acomm{\tau^a}{\tau^b} \right)\ext x^i \ext x^j .
\end{align*}
In the last term the anticommutator follows from the antisymmetry of the wedge product of forms $\ext x^i$ and $\ext x^j$; compare to the basic form in Equation~\eqref{eq:tensor_commutators}. Splitting the calculation now into separate cases depending on the spin indices, we get the following components of the cocycle and the corresponding local Schwinger terms. Note that again the {indices $\mu$ and $\nu$ are fixed}.
\begin{enumerate}
	\item\label{ST-spin-index:same} $\mu = \nu$:
		\begin{align*}
			\theta_2(A;u,v) &= c_2 \int_M  \tr A^{\eta}_c(x) \partial_i u_a^{\mu}(x) \partial_j v_b^{\nu}(x) \sigma_{\eta} \otimes \tau^c \acomm{\tau^a}{\tau^b} \, \ext x^i \ext x^j \\
			  &= c_2 \int_M A^{\eta}_{k,c}(x) \partial_{i} u^{\mu}_a(x) \partial_{j} v^{\nu}_b(x)  \tr(\sigma_{\eta}) \tr(\acomm{ \tau^a}{\tau^b} \tau^c) \, \ext x^k \ext x^i \ext x^j \\
				&= c_2 \epsilon_{ijk} \int_M A^{\eta}_{k,c}(x) \partial_{i} u^{\mu}_a(x) \partial_{j} v^{\nu}_b(x)  \tr(\sigma_{\eta}) \tr(\acomm{ \tau^a}{\tau^b} \tau^c) \, \vol_M.
		\end{align*}
		Here the summation is carried over the local indices $i,j,k$, the spin index $\eta$, and the Yang-Mills indices $a,b,c$. Since the Pauli matrices $\sigma_{\eta}$ are traceless, we see that the Schwinger term appears only for $\eta = 0$:
		\[
			\alpha^{ab}(x,y) = 2 c_2 \epsilon_{ijk} \tr(\acomm{ \tau^a}{\tau^b} \tau^c) \partial_{i} A^0_{j,c}(x) \partial_{k} \delta(x-y) .
		\]
	\item\label{ST-spin-index:zero} $\mu = 0$ and $\nu \neq 0$:
		\begin{align*}
			\theta_2(A;u,v) &= c_2 \int_M  \tr A^{\eta}_c(x) \partial_i u_a^0(x) \partial_j v_b^{\nu}(x) \sigma_{\eta}\sigma_{\nu} \otimes \tau^c \acomm{\tau^a}{\tau^b} \, \ext x^i \ext x^j \\
				&= 2 c_2 \delta_{\eta\nu}\epsilon_{ijk} \int_M A^{\eta}_{k,c}(x) \partial_{i} u^0_a(x) \partial_{j} v^{\nu}_b(x) \tr(\acomm{ \tau^a}{\tau^b} \tau^c) \, \vol_M ,
		\end{align*}
		where we have used $\tr(\sigma_{\eta}\sigma_{\nu}) = 2\delta_{\eta\nu}$. This gives the Schwinger term:
		\[
			\alpha^{ab}(x,y) = 2 c_2 \epsilon_{ijk} \tr(\acomm{ \tau^a}{\tau^b} \tau^c) \partial_{i} A^{\nu}_{j,c}(x) \partial_{k} \delta(x-y) .
		\]
	\item\label{ST-spin-index:nontrivial} $\mu \neq \nu$ and neither are zero:
		\begin{align*}
			\theta_2(A;u,v) &= i c_2 \epsilon_{\mu \nu \eta} \int_M  \tr A^{\zeta}_c(x) \partial_i u_a^{\mu}(x) \partial_j v_b^{\nu}(x)  \sigma_{\zeta} \sigma_{\eta} \otimes \tau^c \comm{ \tau^a}{\tau^b} \, \ext x^i \ext x^j \\
			&= i c_2 \epsilon_{\mu \nu \eta} \epsilon_{ijk} \int_M A^{\zeta}_{k,c}(x)  \partial_i u_a^{\mu}(x) \partial_j v_b^{\nu}(x) \tr(\sigma_{\zeta} \sigma_{\eta})\tr(\comm{ \tau^a}{\tau^b} \tau^c) \, \vol_M\\
			&=  2 i c_2 \epsilon_{\mu \nu \eta} \delta_{\zeta \eta} \epsilon_{ijk} \int_M A^{\zeta}_{k,c}(x) \partial_i u_a^{\mu}(x) \partial_j v_b^{\nu}(x)  \tr(\comm{ \tau^a}{\tau^b} \tau^c) \, \vol_M .
		\end{align*}
		Now the Schwinger term is:
		\[
			\alpha^{ab}(x,y) = 2 i c_2 \epsilon_{\mu \nu \eta} \epsilon_{ijk}  \tr(\comm{ \tau^a}{\tau^b} \tau^c)  \partial_i A^{\eta}_{j,c}(x) \partial_k \delta(x-y) .
		\]
\end{enumerate}
As expected, in the case~\ref{ST-spin-index:same} for $\mu=\nu=0$ we have recovered the well-known Schwinger term involving only the component $A^0_c(x)$: compare to Equation~\eqref{eq:faddeev_schwinger}. In the derivation of the local expressions from the cocycle we have used the techniques described in \cite[Ch.~10]{AI1995}.

\idxloppu{anomaly!commutator}
\idxloppu{current algebra}
\idxloppu{Schwinger terms}

\end{example}

\idxloppu{anomaly!Hamiltonian}
%

% conclusion
%!TEX root = teesirunko-arxiv.tex
\chapter{Conclusion and outlook}

The substance of this thesis roughly divides into two parts, in reverse chronological order: the Hamiltonian anomaly, and higher structures relevant to symmetry groups and anomalies. Concerning the former, we have sketched the relationships between different mathematical structures surrounding the Hamiltonian anomaly in Figure~\ref{fig:hanomaly}.
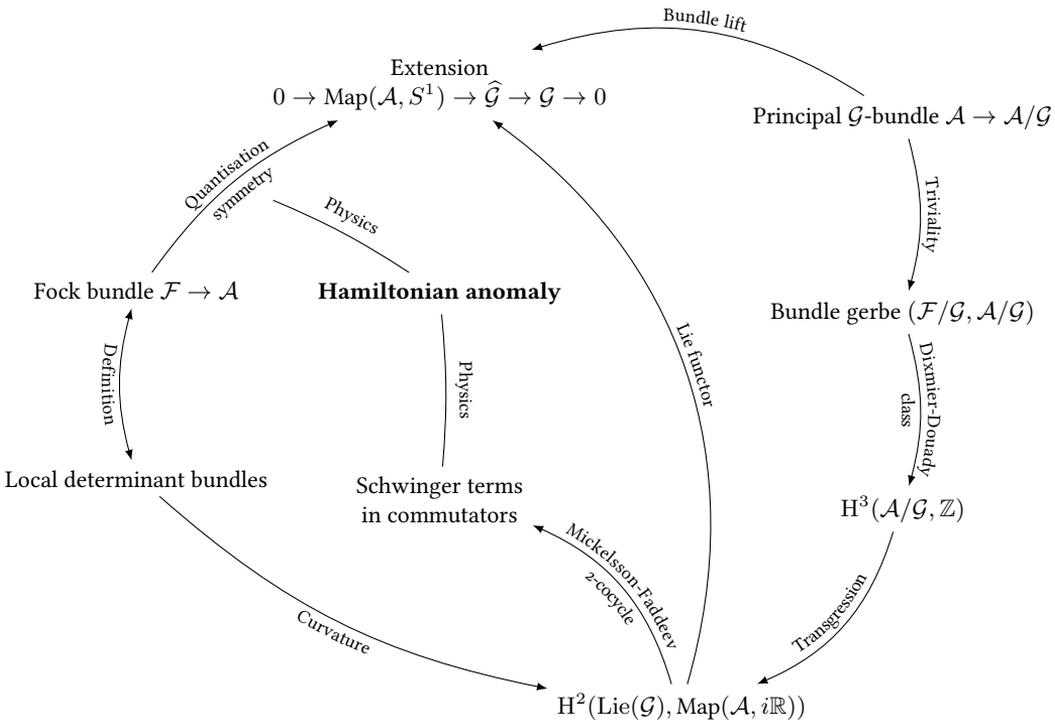
\begin{figure}[!ht]
  \tikzmark{figurecorner}%
    %!TEX root = teesirunko-arxiv.tex

%\tikzset{external/export next=false}
\begin{tikzpicture}[remember picture,overlay,shift={( $ (pic cs:figurecorner) + (-12mm,0mm) $ )},x=6mm,y=4mm]
%\useasboundingbox (0,0);
  %\path[use as bounding box] (0,0) rectangle (10,10); % adjust to fit
  % start node
  \node[circle,draw=black, fill=white, inner sep=0pt,minimum size=5pt,opacity=0] (alphaCorner) at (0,0) {};
  % end node
  \node[circle,draw=black, fill=white, inner sep=0pt,minimum size=5pt,opacity=0] (omegaCorner) at (25,-26) {};
  % help lines: grid
  %\draw[help lines,step=6mm,gray!20] (alphaCorner) grid (omegaCorner);
  %\draw[help lines,step=6mm,gray!20] ( $ (pic cs:figurecorner) + (0,0) $) grid (5,-6);

\end{tikzpicture}

% Relative to node in previous picture:
\begin{tikzpicture}[remember picture,shift={(alphaCorner)},%
    x=6mm,y=4mm]]%
    \path[use as bounding box] (0,0) rectangle (21,22.5); % adjust to fit

\end{tikzpicture}

% Relative to node in previous picture:
\begin{tikzpicture}[remember picture,overlay,shift={(alphaCorner)},%
    every node/.style={anchor=center},every text node part/.style={align=center},%
    node distance=20mm and 2mm,%
    font=\small,%
    x=6mm,y=4mm]
    \tikzset{nuolet/.style={sloped, anchor=center, font=\scriptsize}}
    \tikzset{>=latex}
    %\path[use as bounding box] (0,0) rectangle (10,10); % adjust to f

    % central node: hanomaly
    \node (CENTER) at (10.25,-10) {\bfseries Hamiltonian anomaly};
    \node[below = of CENTER] (ST) {Schwinger terms\\in commutators};
    \node[below right = of ST] (LIECOH) {$\coh^2(\Lie(\G),\map{\A,i\R})$};
    \node[above right = of LIECOH] (COH) {$\coh^3(\A/\G,\Z)$};
    \node[above = of COH] (GERBE) {Bundle gerbe $(\F/\G,\A/\G)$};
    \node[above = of GERBE] (LIFT) {Principal $\G$-bundle $\A \to \A/\G$};
    \node[above = of CENTER] (EXT) {Extension\\$\exact{\map{\A,S^1}}{\gext{\G}}{\G}$};
    \node[below left = of EXT] (FOCK) {Fock bundle $\F \to \A$};
    \node[below = of FOCK] (DET) {Local determinant bundles};

    % shifting text above/beloew to arrows
    \def\aboveshift#1{\raisebox{1ex}}
    \def\belowshift#1{\raisebox{-2.25ex}}

    \draw[<-,color=white,bend left=25,postaction={decorate,decoration={text along path,text align=center,text={|\scriptsize\aboveshift|Mickelsson-Faddeev}}}] (ST) to (LIECOH);
\draw[<-,bend left=25,postaction={decorate,decoration={text along path,text align=center,text={|\scriptsize\belowshift|2-cocycle}}}] (ST) to (LIECOH);

    \draw[<-,bend right=25,postaction={decorate,decoration={text along path,text align=center,text={|\scriptsize\aboveshift|Transgression}}}] (LIECOH) to (COH);

    \draw[->,color=white,bend left=15,postaction={decorate,decoration={text along path,text align=center,text={|\scriptsize\aboveshift|Dixmier-Douady}}}]  (GERBE) to (COH);
    \draw[->,bend left=15,postaction={decorate,decoration={text along path,text align=center,text={|\scriptsize\belowshift|class}}}]  (GERBE) to (COH);

    \draw[->,bend left=15,postaction={decorate,decoration={text along path,text align=center,text={|\scriptsize\aboveshift|Triviality}}}] (LIFT) to (GERBE);

    \draw[<-,bend left=25,postaction={decorate,decoration={text along path,text align=center,text={|\scriptsize\aboveshift|Bundle lift}}}] (EXT) to (LIFT);

    \draw[->,color=white,bend left=15,postaction={decorate,decoration={text along path,text align=center,text={|\scriptsize\aboveshift|Quantisation}}}] (FOCK) to node[nuolet, above] (QUANT) {} (EXT);
    \draw[->,bend left=15,postaction={decorate,decoration={text along path,text align=center,text={|\scriptsize\belowshift|symmetry}}}] (FOCK) to (EXT);

    \draw[<->,bend right=15,postaction={decorate,decoration={text along path,text align=center,text={|\scriptsize\belowshift|Definition}}}] (FOCK) to (DET);

    \draw[->,bend right=15,postaction={decorate,decoration={text along path,text align=center,text={|\scriptsize\belowshift|Curvature}}}] (DET) to (LIECOH);

    \draw[<-,bend left=35,postaction={decorate,decoration={text along path,text align=center,text={|\scriptsize\aboveshift|Lie functor}}}] (EXT) to (LIECOH);
%    \draw[<->,bend left=25] (EXT) to node[nuolet,above] {Lie functor} (LIECOH);

    \draw[-,bend left=5] (CENTER) to node[nuolet,above] {Physics} (ST);
    \draw[-,bend right=5] (CENTER) to node[nuolet,above] {Physics} ( $ (CENTER)!25mm!(QUANT) $ );

\end{tikzpicture}
  \caption{The mathematical structure of the Hamiltonian anomaly: an informal diagram.}
  \label{fig:hanomaly}
\end{figure}

To this picture we could also add a more recent development, that of anomalous categories~\cite{FREED2014,MONNIER2015}. This approach will probably prove to be fruitful for those wishing to tackle anomalies in higher gauge theories. The general framework for the anomalous categories is \emph{topological quantum field theory}, in which the central idea is to form a precise connection between the category of Hilbert spaces and the category of differentiable manifolds. Thus one unifies these two different structures in physics in such a way that a geometric construction of the quantum state spaces can be achieved.

Current groups and algebras as defined in Chapter~\ref{chap:gauge} have proven to be apt tools for modelling the intricacies of gauge symmetry. Yet these structures are still in many ways \emph{classical}, and there is much to be said on what emerges if we quantised these symmetries as well. This is not a new idea as such: while we have only fleetingly touched the topic here,%
  \footnote{There are nontrivial connections to $C^*$-algebras and Kac-Moody algebras, for instance.}
the study of quantum groups is alive and active. There are also examples of higher quantum groups, and it remains to be seen if the categorification of currents groups will have a role in this endeavour.

Indeed the question is: where to go from here? Many parts in this thesis are in some sense old news and rely on already established mathematical notions. Yet the prominence of bundle gerbes already shows the usefulness of rising to a higher level of abstraction. The study of higher structures in quantum physics and categorification is one of the road signs potentially pointing to a deeper understanding of not only mathematical objects but also physics -- perhaps, at the end of the path, it will even offer true glimpses of \emph{new} physics about to emerge. %
So we come back to the Great Question alluded in Chapter~\ref{chap:higher}: not only should we strive to understand \emph{how} to quantise, but also \emph{what} are we really quantising in the first place.

This quest is not exclusively field theoretic; much could be said for instance on the physics of condensed matter in some higher framework.
Perhaps there will never be a single theory of everything, but at least history speaks volumes of the usefulness of stronger mathematical tools. In this light it is easy to conjecture that these higher and more refined structures will make a difference not only in the abstract but also for many concrete applications.

% appendix
\appendix
%!TEX root = teesirunko-arxiv.tex
\chapter{Cohomology theories}\label{app:cohomology}

We present a quick review of various cohomology theories used in the main text. To understand the Dixmier-Douady class of bundle gerbes one needs to have some grip on sheaf cohomology -- we do this by introducing \v{C}ech cohomology; and for the applications of current groups one needs to understand the basics of Lie group and Lie algebra cohomology, including the infinite-dimensional setting which often requires a more refined approach.

We tackle the group cohomology first, since it is simpler and provides necessary background material for the \v{C}ech cohomology.

\section{Group cohomology} %
  \idxalku{Lie group!cohomology of}

Let us recall the basics of group cohomology and illustrate the ideas with a couple of examples. For more details on the extensions and cohomology theory, see \cite[Sec.~18.1--18.3]{HN2012} and \cite[Ch.~4]{MICKELSSON1989} for gauge-theoretic applications; a good source for many physics-driven considerations is \cite{AI1995}.

Let $A$ be an Abelian group and $G$ a group acting via automorphisms $\rho: G \to \aut(A)$ on $A$ -- we say that $A$ is a \term{$G$-module}. For $p\in \N_0$, define an \term{$p$-cochain} for this action as a function $c^p : G^p \to A$ with the condition of being zero if at least one of the arguments is the identity of $G$. The set of $p$-cochains $\coch^p(G,A)$ has a natural group structure by point-wise addition. Consider then a \term{coboundary map}
\[
  \cob_G : \coch^p(G,A) \to \coch^{p+1}(G,A)
\]
defined by the formula\footnote{With respect to the action of $G$ on $A$ which could be either left or right -- we assume that it is left.}
\begin{align*}
  (\cob_G f)(g_1,\dots,g_p,g_{p+1}) \defeq&~ g_1 f(g_2,\dots,g_{p+1}) \\
    &~+ \sum_{i=1}^{p} (-1)^i f(g_1,\dots,g_{i-1},g_ig_{i+1},g_{i+2},\dots,g_{p+1}) \\
    &~+ (-1)^{p+1}f(g_1,\dots,g_p) .
\end{align*}
Define as a subgroup of $p$-cochains the group of \term{$p$-cocycles} by the kernel $\coc^p(G,A) \defeq \ker \cob_G|_{\coch^p(G,A)}$. It is easy to establish that $\cob_G f$ vanishes if any of the arguments is the identity, and that $\cob_G^2 = 0$ when combined as a mapping $\coch^p(G,A) \to \coch^{p+2}(G,A)$. From this we see that the group of $p$-cocycles has a subgroup $\cobo^p(G,A) \defeq \cob_G(\coch^{p-1}(G,A))$ called the \term{$p$-coboundaries}.
\begin{definition}
The quotient group
\[
  \coh^p(G,A) \defeq \coc^p(G,A) / \cobo^p(G,A)
\]
with respect to the homomorphism $\rho:G \to \aut(A)$ is the \termd{$p$th cohomology group of $G$ with values in the $G$-module $A$}.
\end{definition}

An \term{extension of the group $G$ by the group $N$} is a short exact sequence
\[
  \exact{N}{\gext{G}}{G}
\]
of group morphisms $\iota:N \to \gext{G}$ and $g:\gext{G} \to G$ such that $\iota: N \to \gext{G}|_{\ker{g}}$ is an isomorphism. Two extensions $\gext{G}_1$ and $\gext{G}_2$ of $G$ by $N$ are equivalent if there exists a group morphism $\phi:\gext{G}_1 \to \gext{G}_2$ such that the following diagram commutes.
\[
  \begin{tikzcd}[column sep=huge]
    N \arrow[r,"\iota_1"] \arrow[d,"\id_N"] & \gext{G}_1 \arrow[r,"g_1"] \arrow[d,"\phi"] & G \arrow[d,"\id_G"] \\
    N \arrow[r,"\iota_2"] & \gext{G}_2 \arrow[r,"g_2"] & G
  \end{tikzcd}
\]
The morphism $\phi$ must then be an isomorphisms of groups. Let us denote the set of equivalence classes of group extensions by $\extensions(G,N)$.

\begin{theorem}\label{theorem:group_ext_coh}
Let $A$ be an Abelian group. Then there is an isomorphism
\[
  \coh^2(G,A) \isom \extensions(G,A) .
\]
\end{theorem}
\begin{proof}
The bijection is constructed by associating to a cocycle $f \in \coc^2(G,A)$ the Abelian extension $A\times_f G$ as a set $A\times G$ with the product
\[
  (a,g)(a',g') = (a + g.a' + f(g,g'),gg') .
\]
Any Abelian extension of $G$ by $A$ is of this form, and they are equivalent precisely when the defining cocycles $f$ coincide. For details, see \cite[Sec.~18.3]{HN2012}.
\end{proof}
More generally, one can use the cohomology group $\coh^2(G,Z(N))$ -- where $Z(N)$ is the centre of $N$ -- to classify any group extension of $G$ by $N$.

\idxalku{Lie algebra!cohomology of}
There is a corresponding cohomology theory for Lie algebras~\cite[Sec.~7.5--7.6]{HN2012}. An important result analogous to Theorem~\ref{theorem:group_ext_coh} is the following isomorphism:
\[
  \coh^2(\lie{g},\lie{a}) \isom \extensions(\lie{g},\lie{a}) ,
\]
for extensions of a Lie algebra $\lie{g}$ by an Abelian Lie algebra $\lie{a}$ and with respect to a given representation of these Lie algebras on some vector space. The exact relationship between the second cohomology groups of Lie groups and algebras depends on the setting. In particular, we have the following theorem:
\begin{theorem}\label{theorem:group_algebra_coh_isom}
Let $G$ and $A$ be connected Lie groups, $G$ simply connected and $A$ Abelian. Then
\[
  \coh^2(G,A) \isom \coh^2(\lie{g},\lie{a}) ,
\]
where $\lie{g}$ is the Lie algebra of $G$, and $\lie{a}$ is the Lie algebra of $A$.
\end{theorem}
\begin{proof}
See for instance~\cite[Thm.~6.9.1]{AI1995}.
%\hox{Maybe there is a better source...}
\end{proof}

\subsection{Infinite-dimensional groups and their algebras}

We now turn to the Abelian extensions of infinite-dimensional Lie groups and algebras, such as is the case with current groups and algebras. Our canonical source is \cite{NEEB2004}; see also the extensive review \cite{NEEB2006}. Certain results needed in Chapter~\ref{chap:higher} come from \cite{MW2016}.

Let $H$ be a Lie group. Recall that an Abelian Lie group $A$ is a smooth $H$-module, if it is a $H$-module and the action map $H\times A \to A$ is smooth. Assuming that $N$ is a normal subgroup of $H$, the group of $N$-invariant elements of $A$ is denoted by
\[
	A^N = \{ a \in A : (\forall n \in N) ~ n.a = a \} ,
\]
which is a $H$-submodule of $A$.

We can then define more refined cohomology following~\cite[Appendix D]{NEEB2004}. Let us denote \term{smooth group cohomology} by $\coh^p_s(N,A)$ and \term{continuous Lie algebra cohomology} by $\coh^p_c(\lie{n},\lie{a})$, and similarly for all the groups involved (cochains, cocycles, coboundaries). The gist is that the maps $N^p \to A$ (resp. $\lie{n}^p \to \lie{a}$) are smooth (resp. continuous). The smoothness and continuity are defined \emph{locally} in a neighbourhood of the identity.

\begin{definition}
A cocycle $f \in \coc^p_s(N,A)$ is \termd{smoothly cohomologically invariant with respect to $H$} if there is a map
\[
	\phi: H \to \coch_s^{p-1}(N,A) \quad \text{such that} \quad  \ext(\phi(h)) = h.f - f \quad \forall h \in H ,
\]
and the map
\[
	H \times N^{p-1} \to A : (h,n_1,\dots,n_{p-1}) \mapsto \phi(h)(n_1,\dots,n_{p-1})
\]
is smooth in an identity neighbourhood of $H \times N^{p-1}$. This gives us \termd{smoothly invariant cohomology classes of $N$ with values in $A$}.
\end{definition}

\begin{remark}
  In the main text we implicitly assume the smoothness (resp. continuity) for the cohomology groups where applicable, and do not distinguish these groups in notation.
\end{remark}

A connection between the smooth group cohomology and continuous algebra cohomology characterising extensions can be made as follows.
\begin{theorem}[Cohomology homomorphism~{\cite[Thm. VII.2]{NEEB2004}}]\label{theorem:cohom_hom}
Let $N$ be a connected Lie group and $A \isom \lie{a}/\Gamma_A$ a smooth $N$-module, where $\Gamma_A \subset \lie{a}$ is a discrete subgroup of the sequentially complete locally convex space $\lie{a}$. Then there is an exact sequence
\[
	\hom(\pi_1(N),A^N) \to \coh^2_s(N,A) \to \coh^2_c(\lie{n},\lie{a}) .
\]
\end{theorem}

In Chapter~\ref{chap:higher} we discuss the lifting of a group action to an Abelian extension. Consider an Abelian Lie group extension
\[
  \exact{A}{\gext{N}}{N} ,
\]
where $N \subset H$ is a normal split subgroup of some Lie group $H$. Let us denote by $\aut(\gext{N},A)$ the group formed from automorphisms of the extension $\gext{N}$ that preserve the split Lie subgroup $A$. To this extension we want to lift the action given by
\[
  \psi: H \to \aut(A) \times \aut(N) .
\]
Define a map $\ext_{\psi}$ as the coboundary operator in the cochain complex of maps $f : H^p \to \coch_s^1(N,A)$, in relation to the action of the group $H$ on $\coch_s^1(N,A)$ by $h.f = \psi(h).f$. We then have the following result.

\begin{proposition}[Lifting homomorphism~{\cite[Prop. 3.8]{MW2016}}]\label{proposition:aut_lift}
Let $H$ be a Lie group, $N$ a connected normal Lie subgroup of $H$, $\theta \in \coc^2_s(N,A)$ a smooth $2$-cocycle and $\gext{N}$ the corresponding Lie group extension by an Abelian group $A$. Then the smooth group homomorphism $\psi: H \to \aut(A) \times \aut(N)$ lifts to a smooth homomorphism $\gext{\psi}:H\to \aut(\gext{N},A)$ if and only if
\begin{enumerate}
	\item $\theta$ is smoothly cohomologically invariant with respect to $H$, and
	\item the corresponding cohomology class $[\ext_{\psi}\phi] \in \coh^2_s(H,\coc^1_s(N,A))$ is trivial, where the $1$-cocycle $\phi$ is defined via $\ext_N (\phi(h)) = h.\theta - \theta$ for any $h \in H$.
\end{enumerate}
\end{proposition}
In \cite{MW2016}, the corresponding notation is $\ext_{S_{\psi}}$, where $S_{\psi}$ denotes the action.

\idxloppu{Lie algebra!cohomology of}
\idxloppu{Lie group!cohomology of}

\section{Sheaves of groups and \v{C}ech cohomology} %
\idxalku{Ce@\v{C}ech cohomology}

The cohomology classification of (bundle) gerbes rests on defining cocycles via local transition functions on a covering of the base manifold $X$: this leads to cohomology theories with respect to sheaves of Abelian groups. We mostly rely on \cite{BRYLINSKI1993} as a source on sheaf and \v{C}ech cohomology with focus on the gerbe-theoretic applications.

Let $X$ be a topological space. For any open set $\U$ of $X$, define a set $\F(\U)$ (with $\F(\emptyset) = \{0\}$ defined as a one-element set). If then $\U_j \subset \U_i$ is an inclusion of open sets, we further define a \term{restriction map}%
  \footnote{In the index $\rho_{\U_j,\U_i}$ we will abbreviate $(\U_j,\U_i)$ by $(j,i)$ if possible without creating undue confusion. We will also abbreviate $\U_i\cap\U_j$ by $\U_{i,j}$.}
\[
  \rho_{j,i}: \F(\U_i) \to \F(\U_j) .
\]
Now, if for open sets $\U_k \subset \U_j \subset \U_i$ the restriction map satisfies the \emph{transitivity condition}
\[
  \rho_{k,i} = \rho_{k,j}\rho_{j,i} \quad \text{and} \quad \rho_{i,i} = 1 ,
\]
the collection $\F$ of sets is called a \term{presheaf of sets} over $X$. A set $\F(\U)$ is called a \term{section of $\F$} over $\U$.

%Let $\U = \{\U_i\}_{i\in \J}$ an open covering

Let $\U$ be an open set of $X$, and let $\{\U_i\}$ be its open covering. Consider a family $\{\sigma_i\}$ of elements $\sigma_i \in \F(\U_i)$ satisfying the \emph{glueing condition}
\[
  \rho_{\U_{i,j},\U_i}(\sigma_i) = \rho_{\U_{i,j}, \U_j}(\sigma_j) .
\]
If for all such families there exists a unique $\sigma \in \F(\U)$ such that $\rho_{\U_i,\U} = \sigma_i$, the presheaf $\F$ is called a \term{sheaf} over $X$. %
\idx{sheaf}

Let then $Y$ be a topological space. Any presheaf $\F$ of continuous maps from open sets of $X$ to $Y$ is also a sheaf (of maps). For example, when $X$ is a smooth manifold, we have the sheaf of complex-valued functions $\sheaf{\C}$ on $X$. In particular $Y$ can be a topological group, and then the sheaf is called a sheaf of groups.
In the following we assume that $X$ is a smooth manifold and $Y=A$ is an Abelian group.

There exists a cohomology theory defined as a functor from the category of sheaves of Abelian groups to the category of Abelian groups. However, we will here opt to introduce \v{C}ech cohomology instead. Its construction is more explicit, and the theory provides useful computational tools for sheaf cohomology. For details about the construction of the sheaf cohomology we refer to \cite[Sec.~1.1]{BRYLINSKI1993}.

Consider a presheaf $\sheaf{A}$ of Abelian groups with respect to an open covering $\U = \{\U_i\}_{i\in\J}$. Let us write $\U_{i_0,\dots,i_p}$ for an finite intersection of open sets $\U_{i_0}\cap \U_{i_1} \cap \cdots \cap \U_{i_p}$. Given $p\in \N_0$, we define the group of \term{$p$-cochains} as
\[
  \coch^p(\U,\sheaf{A}) = \prod_{i_0,\dots,i_p} \sheaf{A}(\U_{i_0,\dots,i_p}) .
\]
Thus a $p$-cochain $\sheaf{\alpha}$ is a family of maps $\alpha_{i_0,\dots,i_p} \in \sheaf{A}(\U_{i_0,\dots,i_p})$. A \term{coboundary map} $\coch^p(\U,\sheaf{A})  \to \coch^{p+1}(\U,\sheaf{A})$ can then be defined as
\[
  \cob(\sheaf{\alpha})_{i_0,\dots,i_{p+1}} = \sum_{k=0}^{p+1} (-1)^k (\alpha_{i_0,\dots,i_{j-1}, i_{j+1},\dots,i_{p+1}})_{|\U_{i_0,\dots,i_{p+1}}} ,
\]
where the sum is taken on a restriction to the intersection $\U_{i_0,\dots,i_{p+1}}$. As usual, the $(p+1)$-cochain $\cob(\sheaf{\alpha})$ is called the \term{coboundary} of $\sheaf{\alpha}$, and if $\cob(\sheaf{\alpha}) = 0$, the $p$-cochain $\sheaf{\alpha}$ is called a \term{cocycle}.

It is not difficult to see that $\cob^2 = 0$~\cite[Prop.~1.3.1, p.~25]{BRYLINSKI1993}, and we can define the cohomology with respect to the covering:
\begin{definition}
  The complex of groups
  \[
    \begin{tikzcd}
      \dots \arrow[r,"\delta"] & \coch^p(\U,\sheaf{A}) \arrow[r,"\delta"]  & \coch^{p+1}(\U,\sheaf{A}) \arrow[r,"\delta"] & \dots
    \end{tikzcd}
  \]
  defines the \termd{\v{C}ech cohomology groups of the covering $\U$} with coefficients in the (pre)sheaf $\sheaf{A}$, denoted by $\ccoh^p(\U,\sheaf{A})$.
\end{definition}
This definition can be extended from a given covering to the whole space $X$. We define:
\begin{definition}
  The degree $p$ \term{\v{C}ech cohomology group $\ccoh^p(X,\sheaf{A})$} is the direct limit
  \[
    \ccoh^p(X,\sheaf{A}) = \lim_{\U} \ccoh^p(\U,\sheaf{A})
  \]
  taken over the ordered set of open coverings $\U$ of $X$.
\end{definition}
What we mean by the ordered set of coverings is the following. An open covering $\V = \V_{i\in\I}$ is defined to be \term{finer} than $\U$ if there is a map $\phi:\I \to \J$ such that $\V_i \subset \U_{\phi(i)}$ for all $i\in\I$: this defines an order relation $\U \prec \V$ in the set of all open coverings of $X$. This ordering then induces a morphism of \v{C}ech complexes
\[
  \phi_* : \coch^{}(\U,\sheaf{A}) \to \coch^{}(\V,\sheaf{A})
\]
satisfying, for any given $p$-cochain $\sheaf{\alpha}$ for the covering $\U$, the following condition:
\[
  \phi_*(\sheaf{\alpha})_{i_0,\dots,i_p} = (\alpha_{\phi(i_0),\dots,\phi(i_p)})_{|\V_{i_0,\dots,i_p}} .
\]

Finally, and while we have not explicitly defined the sheaf cohomology groups $\coh^{p}(X,\sheaf{A})$, we can note the following useful results connecting \v{C}ech cohomology and sheaf cohomology.~\cite[Prop.~1.3.4 and Thm.~1.3.13]{BRYLINSKI1993}
\begin{theorem}
  There is a canonical group homomorphism
  \[
    \ccoh^{p}(\U,\sheaf{A}) \to \coh^{p}(X,\sheaf{A})
  \]
  If $X$ is paracompact, this map is an isomorphism.
\end{theorem}

\idxloppu{Ce@\v{C}ech cohomology}
%

% line numbering ends
%\nolinenumbers

% - testihomma !
%\clearpage
%\begingroup
%\renewcommand{\cleardoublepage}{}
%\chapter{Acknowledgments}
%\newpage
%\endgroup
% - testihomma ends !

% fiksataan lähdeluettelolle sisälmyksen linkit, sivunumerointi ja otsikko
\cleardoublepage
\phantomsection
\addcontentsline{toc}{chapter}{\bibname}
\printbibliography

% \cleardoublepage
% \phantomsection
% \addcontentsline{toc}{chapter}{Index}
% \printindex

\end{document}